\documentclass[fleqn,letterpaper,UKenglish,11pt]{article}
\usepackage[margin=1in]{geometry}

\usepackage[l2tabu, orthodox]{nag}
\usepackage[utf8]{inputenc}
\usepackage[pdfusetitle]{hyperref}
\usepackage{microtype}
\usepackage{tcolorbox}

\newcommand{\weight}{\mathbf{w}}
\newcommand{\eq}{\mathsf{Eq}}

\newcommand*{\EE}{\mathcal{E}}

\newcommand*{\GG}{\mathcal{G}}

\newcommand*{\MM}{\mathcal{M}}
\newcommand*{\NN}{\mathcal{N}}

\usepackage{xcolor}
\usepackage{mathtools}

\mathtoolsset{
showonlyrefs=true % or false in draft mode
}%use \equation{} everywhere; only the referenced ones are numbered

\usepackage{amsmath,amssymb,amsthm}
\theoremstyle{plain}
\newtheorem{theorem}{Theorem}[section]
\newtheorem{conjecture}[theorem]{Conjecture}
\newtheorem*{conjecture*}{Conjecture}
\newtheorem{proposition}[theorem]{Proposition}
\newtheorem{lemma}[theorem]{Lemma}

\newtheorem{corollary}[theorem]{Corollary}
\theoremstyle{definition}
\newtheorem{definition}[theorem]{Definition}
\newtheorem{example}[theorem]{Example}
\newtheorem{remark}[theorem]{Remark}

\usepackage{fullpage}

\usepackage{enumerate} %Gives ability, e.g., \begin{enumerate}[Case
\usepackage{authblk}

\usepackage{xspace}

\newcommand{\cclass}[1]{\ensuremath{\mbox{\textup{#1}}}\xspace}

\newcommand{\NP}{\cclass{NP}}

\newcommand{\FPT}{\cclass{FPT}}

\newcommand{\Weft}{{\cclass{W}}}
\newcommand{\W}[1]{{\Weft}{{[#1]}}}

\newcommand{\N}{\ensuremath{\mathbb N}\xspace}

\newcommand{\Q}{\ensuremath{\mathbb Q}\xspace}
\newcommand{\Z}{\ensuremath{\mathbb Z}\xspace}

\newcommand{\cB}{\ensuremath{\mathcal B}\xspace}
\newcommand{\cC}{\ensuremath{\mathcal C}\xspace}
\newcommand{\cP}{\ensuremath{\mathcal P}\xspace}

\newcommand{\ar}{\ensuremath{ar}}

\newcommand{\all}{\ensuremath{\mathrm{All}}}
\newcommand{\oh}{\ensuremath{\mathrm{OH}}}
\newcommand{\mo}{\ensuremath{\mathrm{Mo}}}
\newcommand{\id}{\ensuremath{\mathrm{Id}}}

\newcommand{\pmin}{\ensuremath{\mathsf{min}}}
\newcommand{\pmax}{\ensuremath{\mathsf{max}}}
\newcommand{\pmedian}{\ensuremath{\mathsf{median}}}
\DeclareMathOperator{\median}{median}

\newcommand{\pcon}{\ensuremath{\mathsf{con}}}
\newcommand{\fcon}{\ensuremath{\mathrm{con}}}

\newcommand{\ppol}{\ensuremath{\mathrm{pPol}}}
\newcommand{\mpol}{\ensuremath{\mathrm{mPol}}}
\newcommand{\tmpol}{\ensuremath{\mathrm{tmPol}}}
\newcommand{\pol}{\ensuremath{\mathrm{Pol}}}
\newcommand{\inv}{\ensuremath{\mathrm{Inv}}}

\newcommand{\pcclone}[1]{\ensuremath{\langle #1 \rangle_{\not \exists}}}

\newcommand{\ucsp}{\textsc{Unrestricted-CSP}\xspace}
\newcommand{\ohcsp}{\textsc{One-Hot-CSP}\xspace}
\newcommand{\mcsp}{\textsc{Monotone-CSP}\xspace}
\newcommand{\udcsp}{\textsc{udCSP}\xspace}

\newcommand{\cclone}[1]{\ensuremath{\langle #1 \rangle}}
\newcommand{\fgpp}[2]{\ensuremath{\langle #1 \rangle^{#2}_{\mathrm{fgpp}}}}
\newcommand{\efpp}[2]{\ensuremath{\langle #1 \rangle^{#2}_{\mathrm{efpp}}}}

\newcommand{\img}[2]{\ensuremath{\mathcal{H}(#1, #2)}}

\newcommand{\Nat}{\mathbb{N}}

\DeclareMathOperator{\Pol}{Pol}
\DeclareMathOperator{\Inv}{Inv}

\newcommand{\Poly}{\cclass{P}}

\newcommand{\numdom}[1]{\ensuremath{\mathtt{D}_{#1}}}

\newcommand{\problemDef}[3]
{%
    \begin{tcolorbox}[arc=0.1mm,boxsep=-0.6mm,left=1.9mm,right=1.9mm,bottom=1.4mm,top=1.4mm,adjusted title={\strut \sc#1},colback=white!5]

    \noindent\textbf{Instance:} #2
    \noindent\textbf{Question:} #3
    \end{tcolorbox}
}

\newcommand{\problemDefP}[4]
{%
    \begin{tcolorbox}[arc=0.1mm,boxsep=-0.6mm,left=1.9mm,right=1.9mm,bottom=1.4mm,top=1.4mm,adjusted title={\strut \sc#1},colback=white!5]

    \noindent\textbf{Instance:} #2
  
    \noindent\textbf{Parameter:} #3
    
    \noindent\textbf{Question:} #4
    \end{tcolorbox}
}

\newcommand{\thegraph}{assignment\xspace}

\usepackage{todonotes}
\usepackage{booktabs}

\begin{document}
\title{Going Beyond Twin-width? CSPs with \\Unbounded Domain and Few Variables}
 \author[1]{Peter Jonsson}
 \affil[1]{Link\"oping University, Sweden}
 \author[1]{Victor Lagerkvist}
 \author[1]{Jorke M.\ de Vlas}
  \author[2]{Magnus Wahlstr\"om}
 \affil[2]{Royal Holloway, University of London, UK}

% \pagenumbering{gobble} 
\date{}
\maketitle

\pagenumbering{roman}

\begin{comment}
\textsc{Paired Min Cut} is commonly used as a 
source problem for proving 
W[1]-hardness (see e.g.~\cite{Dabrowski:etal:ipec2023,KimKPW21flow,marx2009constant,Osipov:Wahlstrom:esa2023}).
The W[1]-hardness of this problem follows from work by Marx and Razgon~\cite[Theorem~7]{marx2009constant}, but they
formulate the problem in logical terms. A graph-theoretic formulation, slightly different from the above, can be found in~\cite[Lemma~5.7]{Kim:etal:arxiv2022}. In particular, Kim et al.~\cite{Kim:etal:arxiv2022} assume the graph to be a DAG, however, W[1]-hardness of \textsc{Paired Min Cut} as defined above follows easily from this.

\end{comment}

\begin{abstract}    
We study a model of constraint satisfaction problems geared
  towards instances with few variables but with domain of unbounded size ($\udcsp$).
  Our model is inspired by recent work on FPT algorithms for MinCSP
  where, surprisingly, both upper and lower bounds on the parameterized
  complexity frequently correspond to solving $k$-variable $\udcsp$s.
  For example, both the FPT algorithms for \textsc{Boolean MinCSP} (Kim et al., SODA 2023) 
  and \textsc{Directed Multicut} with three cut requests (Hatzel et al., SODA 2023)
  rely directly on such a reduction, and the canonical W[1]-hardness construction
  in the area (\textsc{Paired Min Cut} by Marx and Razgon (IPL 2009)) corresponds to a $k$-variable $\udcsp$. 
The $\udcsp$ framework represents constraints with unbounded domains via a collection
of unary maps $\MM$ into a finite-domain base language $\Gamma$ -- this places $\udcsp$ in the {\em terra incognita} between finite-
and infinite-domain CSPs, rendering previously studied algebraic approaches inapplicable.
We develop an algebraic theory for studying the complexity of $\udcsp(\Gamma,\MM)$ based on {\em partial multifunctions}, and present a Galois connection with a natural
definability notion for $\udcsp$. 
We illustrate by graph separation: we can now view problems such as {\sc Paired Min Cut}, {\sc 3-Chain Sat}, and {\sc Coupled Min-Cut} in a general framework that explains their inherent parameterized complexity.
These insights guide us when simplifying a key step in the aforementioned
  FPT algorithm for the \textsc{Boolean MinCSP}: a clutter of 
  branching steps is replaced
  by a clean reduction to $\udcsp$.

In more technical detail, we study three different types of maps --
  unrestricted maps, \emph{one-hot} maps, and \emph{monotone} maps.
  For unrestricted maps, the problem is W[1]-hard for all but trivial cases,
  and for one-hot maps, the characterization coincides with
  Marx' FPT dichotomy for Boolean \textsc{Weighted CSPs} (Computational Complexity 2005).
For the monotone maps $\mo$, the project turns out to be much more challenging. In this setting,
  the algebraic objects boil down to {\em ordered polymorphisms} definable by restricted identities. We identify the \emph{connector} polymorphism
    that characterizes when $\udcsp(\Gamma,\mo)$
    can define all permutations or not, and prove that its absence implies
    W[1]-hardness. 
    When $\Gamma$ is binary, then its presence implies membership in FPT; we prove this
    by using results for
    {\em twin-width} of ordered structures (Twin-Width IV;
    Bonnet et al., JACM 2024).  
  For non-binary languages,
  the twin-width connection breaks down completely.
One rescue idea is to replace twin-width by a more restrictive
width notion better suited for non-binary
structures. We show that the obvious example {\em rank-width}
is not feasible. Another idea
is to investigate width notions for general structures
that specializes to twin-width for binary structures, and we introduce the parameter {\em projected grid-rank} and prove that it coincides with (1) the connector polymorphism and (2) twin-width of binary projections, via a Ramsey-style argument. 
Exploiting projected grid-rank algorithmically is not straightforward, however, and a future resolution of $\udcsp(\Gamma, \mo)$ likely needs a proper generalization of twin-width or a functionally equivalent width-measure to higher-arity relations.
%We identify barriers for this approach to work and this, unfortunately, 
%forces us to leave the question of FPT algorithms
% to future work.\todo{PJ: Extend with Ramsey.}
\end{abstract}

%To this end,
%we introduce {\em projected grid-rank}, guided by the restrictions incurred by the %line polymorphism.
%The poorly understood nature of width measures for
%  non-binary structures forces us to leave the question of FPT algorithms
%  to future work, unfortunately. 

\newpage

\tableofcontents

%\listoftodos

\bigskip

%Status: PJ

%\begin{enumerate}
%\item
%Connector. Fixed. Use macros \textbackslash pcon (pattern) and
%\textbackslash fcon (function).

%\item
%Permutation notation. Fixed.

%\item
%New name for incidence graph. Use macro \textbackslash thegraph. For now, I call it
%{\em assignment graph}. 

%\item
%Function difference. Fixed.

%\item
%Figure 3 needs improvements. Not anymore.

%\end{enumerate}

%%%%%%%%%%%%%%%%%%%%%%%%%%%%%%%%%%%%%%%%%%%%%%%%%%%%%%%%%%%%%%%%%%%%%%%%%%%%%%%%%%%%%%%%%%%%

\newpage

\pagenumbering{arabic}

\section{Introduction}

In this paper we study the finite-domain {\em constraint satisfaction problem}, i.e., the problem of deciding whether a set of constraints over a set of variables $V$ with domain $D$ admits at least one solution. We are specifically interested in the scenario when $|V|$ is expected to be small but where $|D|$ can be large. We present the background for studying this problem in Section~\ref{sec:background}, summarize the main results of the paper in Section~\ref{sec:summary}, and discuss applications and related research in Section~\ref{sec:applications}.

\subsection{Background} \label{sec:background}
In classical complexity investigations into CSPs, the most typical
setting is to consider a restricted variant of the CSP problem where a typically finite set of relations (a {\em constraint language}) is used to determine the allowed constraints (CSP$(\Gamma)$). 
%where the language $\Gamma$
%and the problem domain are fixed and finite, and analyse its complexity 
%based on the number of variables.
There is a rich toolbox for
analysing the complexity of such problems: 
on the relational side there are many closure operators on relations which imitate different forms of gadget reductions: {\em primitive positive definability} (pp-definability)
gives simple gadget reductions based on replacing constraints by something equivalent,
and on a higher level one obtains {\em primitive positive constructions/implementations} (pp-constructions/implementations). 
In terms of gadget metaphors, this corresponds to NP-hardness
reductions ``between problem domains'', e.g., a problem over some
domain $D$ can be reduced into a problem over a domain $D'$ by
mapping variables over $D$ to tuples of variables over $D'$. For example, the classical reduction from \textsc{$k$-coloring} to \textsc{$3$-SAT} can be formulated as a pp-construction of the $k$-clique to the Boolean language where each relation is the set of models of a 3-clause. These notions can also be witnessed by algebraic counterparts and the resulting theory often goes under the name the {\em algebraic approach}~\cite{barto2018}. The basic object corresponding to pp-definability is then that of a {\em polymorphism} which describes permissible ways to combine solutions, and to obtain an analogoue of pp-constructions one considers polymorphisms definable by restricted forms of identities ({\em strong Maltsev conditions}). Via this approach the classical complexity of CSP$(\Gamma)$ has been fully classified depending: it is tractable if and only if $\Gamma$ does not have a polymorphism satisfying a non-trivial identity~\cite{Bulatov:focs2017,Zhuk:jacm2020}.

%In a direction that is quite orthogonal to the classical
%investigations, in parameterized complexity one encounters
%problems that are essentially definable as CSPs where the number of
%{\em variables} is a parameter, say $k$, and the domain is some large (but
%finite) input-defined size $n$; refer to these as \emph{$k$-variable,
%  unbounded-domain CSPs}. In one way, the fundamental method of
%colour-coding can be seen as precisely a reduction of a problem into a
%$k$-variable setting; e.g., colour-coding can be used to reduce
%\textsc{$k$-Clique}, the most fundamental W[1]-hard problem, to
%\textsc{Multicoloured $k$-Clique}, where the latter is precisely
%equivalent to a $k$-variable arbitrary 2-CSP with input-defined domain
%size $n$. 
%\todo{MW: This probably needs expanding on.} 
%However, this perspective may be largely vacuous, since the most
%famous positive applications of colour-coding (such as the first
%$O^*(2^{O(k)})$-time algorithm for \textsc{$k$-Path}\todo{CITE})
%cannot really be described in CSP terms.
%More recently $k$-variable unbounded-domain CSPs have
%started to become important as a building block for many
%problems in FPT, especially including the \textsc{MinCSP} setting.
In a direction that is quite orthogonal to the classical
investigations, in parameterized complexity one encounters
problems that are essentially definable as CSPs where the number of
{\em variables} is a parameter, say $k$, and the domain is some large (but
finite) input-defined size $n$; refer to these as \emph{$k$-variable,
  unbounded-domain CSPs} ($\udcsp$). 
  Certainly, parameterizing by the number of variables has been considered earlier in the literature. For a classical example, integer programming is known to be FPT with respect to the number of variables~\cite{Lenstra:mor83}, a fact which has seen some FPT applications (see, e.g.,~\cite{GavenciakKK22ip}).
In a less formal way (i.e., not directly corresponding to a CSP), one can also view the classical method of \emph{colour coding}~\cite{AlonYZ95} as reducing a problem of finding a $k$-vertex subgraph $H$ in a given graph $G$ to a problem of assigning target ``values'' (vertices in $G$) to $k$ ``variables'' (vertices from $H$) in a way such that the edge adjacencies are preserved.
This applies to the hardness side as well, where the classical W[1]-hard problem \textsc{Multicoloured $k$-Clique} is frequently treated as a $k$-variable 2-CSP with input-defined domain size $n$.
Machine learning in fixed dimension --- here, dimension refers to the number of
input neurons --- is also directly related to finite-variable
problems and its computational and parameterized
complexity is an active area of research~\cite{Chen:etal:focs2021,Froese:Hertrich:neurips2023,Froese:etal:jair2022,Khalife:etal:mp2024}.
Yet another example comes from Bringmann et al.~\cite{Bringmann:etal:jcss2016} 
who study the parameterized complexity of
{\sc Steiner Multicut} and 
introduce the problem {\sc  NAE-Integer-3-SAT} 
as a  handy source of hardness (it has subsequently been used
in several works:
for instance, 
Galby et al.~\cite{Galby:etal:sidma2023} use it for studying the {\sc Metric Dimension} problem, 
Ganian et al.~\cite{Ganian:etal:iclr2025} use it for a problem in machine learning, and
Misra et al.~\cite{Misra:etal:algorithmica2024} for problems in game theory.)
Their W[1]-hardness proof is a
reduction from \textsc{Multicoloured $k$-Clique} and, not surprisingly,
{\sc  NAE-Integer-3-SAT} can be viewed as an
instance of the $\udcsp$ problem (see Example~\ref{ex:firstexample}).

However, more recently, this phenomenon has started occurring in a more systematic manner, especially in graph separation
problems and the study of \textsc{MinCSP}s. This is primarily due to the 
breakthrough tool of \emph{flow augmentation} due to Kim et al.~\cite{KimKPW25fa1},
through which graph separation problems can be reduced to instances where the solution is an \emph{$st$-min cut} in an augmented graph. 
Via Menger's theorem, the set of $st$-min cuts in a graph can in turn be encoded as the set of solutions to a $k$-variable CSP:
Given an $st$-max flow $\cP=\{P_1,\ldots,P_k\}$ in a graph $G$,
there is a CSP with a variable $x_i$ for each path $P_i \in \cP$, 
where the value of $x_i$ encodes where $P_i$ will be cut,
and binary constraints between the variables can enforce that a satisfying assignment corresponds to an $st$-min cut in $G$ (see Section~\ref{sec:applications}).
Effectively, this reduces the search space from $\{0,1\}^n$ (all $st$-cuts)
to $[n]^k$ (solutions to the CSP). 
This perspective turns out to be very fruitful. 
For example, the long-standing open problem \textsc{$\ell$-Chain SAT}~\cite{ChitnisEM17} was solved by Kim et al.~\cite{KimKPW25fa1} directly by this method,
and for a more complex application, the FPT algorithm for \textsc{Directed Multicut} with three cut requests~\cite{doi:10.1137/1.9781611977554.ch123} 
was solved by a combination of methods, including flow augmentation to reduce to a $k$-variable CSP, and the use of \emph{twin-width}~\cite{bonnet2024twin}
together with structural arguments to show that the resulting CSP can be solved in FPT time. 
An even more direct connection is for \textsc{MinCSP}. 
For a constraint language $\Gamma$, the problem \textsc{MinCSP$(\Gamma)$} asks, 
given a CSP with constraints from $\Gamma$ and an integer $k$, whether there is an assignment that satisfies all but at most $k$ constraint in the instance.
By varying the constraint language $\Gamma$, one gets a rich variety of constraint optimization problems, where we in particular are interested in knowing for which languages $\Gamma$ the problem \textsc{MinCSP$(\Gamma)$} is FPT parameterized by $k$. 
For Boolean languages, Kim et al.~\cite{KimKPW23fa3} gave a dichotomy of \textsc{MinCSP$(\Gamma)$} as being FPT or W[1]-hard, and here, both the upper and
lower bounds go via $k$-variable $\udcsp$s.
For the upper bound, the most interesting case is for \emph{bijunctive} languages $\Gamma$, where every relation in $\Gamma$ is defined by a 2-CNF formula. 
The FPT algorithm for tractable languages of this type has three steps:
(1) a standard reduction to a graph cut problem; (2) heavy use of flow-augmentation to reduce to an instance where the solution is an $st$-min cut; and (3) a complex branching algorithm to solve such a min-cut instance. 
Step (3) can be seen as a $\udcsp$; in fact, we argue that the results of this paper lead to a simplified, less ad-hoc solution to this step. 
The lower bounds are primarily reductions from 
\textsc{Paired Min Cut}~\cite{MarxR09}\footnote{Marx and Razgon formulate this problem in logical terms. A more common graph-theoretic formulation can be found in~\cite[Lemma~5.7]{Kim:etal:arxiv2022}.}, which is precisely a simple, W[1]-hard $\udcsp$ problem. 
Similar outlines can be given for many other parameterized \textsc{MinCSP}s -- frequently, what decides tractability is whether a problem can encode a W[1]-hard $k$-variable $\udcsp$ such as \textsc{Paired Min Cut}~\cite{Dabrowski:etal:ipec2023,DabrowskiJOOW23soda,OsipovPW24pointalgebra,OsipovW23equality}.%\todo{Good refs? More? PJ: I added two more.}

In this paper we wish to make this perspective explicit and precise,
and we seek to develop the machinery to analyse the parameterized
complexity of problems with few variables but with unbounded domain size. To be able to study this problem in a systematic way we follow the classic approach and first fix a finite constraint language $\Gamma$ (over some fixed finite domain). Such restrictions are necessary since the problem is otherwise trivially W[1]-hard.
However, since the domain is variable and part of the input, constraints cannot merely be given as constraints over $\Gamma$, and to circumvent this we consider constraints of the form $R(m_1(x_1), \ldots, m_r(x_r))$ where $x_1, \ldots, x_r$ are variables (at most $k$ many) and $m_1, \ldots, m_r$ are maps from the input domain $\{0, 1, \ldots, n-1\}$ to the domain of $R$. Thus, even if the base language $\Gamma$ is fixed we can still describe make use of unbounded domains by using maps to the base domain. We let $\udcsp(\Gamma)$ denote this generalization of CSP and are then primarily interested in analyzing the parameterized complexity when the the number of variables is used as parameter.  %\todo{PJ: Should we contrast our results with Lenstra here?}

The driving force behind our approach is now what we can learn about $k$-variable problems when studied in the unifying and very general $\udcsp(\Gamma)$ framework. Thus, is it possible to obtain interesting FPT/W[1] dichotomies or is the number of variables as a parameter inherently too frail? As we will see, this question can be greatly simplified by universal algebra, and by studying restricted problems in this setting we unravel new connections between CSPs, ordered polymorphisms, and the important {\em twin-width} 
parameter~\cite{bonnet2024twin,Bonnet:etal:jacm2022}.

\subsection{Summary of Results} \label{sec:summary}

We begin by describing the most relevant types of maps that we allow in $\udcsp(\Gamma)$ instances. 
We primarily consider constraint languages $\Gamma$ defined over some base  domain $\numdom{d} = \{0, 1, \ldots, d - 1\}$ and then define a {\em map family} to be a set of maps which for each $(a,b) \in \mathbb{N}^2$ contains exactly one set of maps $\subseteq \numdom{a}^{\numdom{b}}$. We extend $\udcsp(\Gamma)$ to use a map family $\MM$ as an additional parameter ($\udcsp(\Gamma, \MM)$) and each instance is then only allowed to use maps from $\MM$. The requirement that the map family contains maps for all possible combinations $(a,b) \in \mathbb{N}^2$ is mainly a convenience to make it easier to speak of $\udcsp(\Gamma, \MM)$ problems with different base domains, and to obtain a more robust algebraic theory.

First, it is straightforward to show that if the map family $\MM$ is completely arbitrary and contains all possible maps then $\udcsp(\Gamma, \MM)$ is W[1]-hard hard except when $\Gamma$ is a trivial language definable over unary relations. For example, even the Boolean equality relation $\eq = \{00, 11\}$ results in a W[1]-hard problem. In contrast, equality constraints are trivial to handle in the standard CSP problem and on the algebraic side one typically allows them to be used regardless of whether $\Gamma$ can actually express it or not. Hence, to obtain a reasonable problem we have to (1) fix a more restrictive map family $\MM$ and (2) attempt to classify $\udcsp(\Gamma, \MM)$ for all possible base languages $\Gamma$. Inspired by the two classical Boolean encodings {\em one-hot} and {\em unary} we define the following two types of maps (for the precise definitions of the associated map families $\oh$ and $\mo$ we refer the reader to Section~\ref{sec:statement}).

\begin{enumerate}
  \item 
    $\oh$: map $m : \numdom{n} \to \numdom{2}$ is called \emph{one-hot} if $m(x) = 1$ for exactly one $x \in \numdom{n}$. We study the resulting problem $\udcsp(\Gamma, \oh)$ for Boolean base languages $\Gamma$.
  \item
    $\mo$: a map $m : \numdom{n} \to \numdom{d}$ is called \emph{monotone} if $x \leq y$ implies that $m(x) \leq m(y)$ according to the natural ordering on $\numdom{n}$ and $\numdom{d}$. The problem $\udcsp(\Gamma, \mo)$ is relevant for any finite-domain base language $\Gamma$. Sometimes we also need anti-monotone maps and we let $\mo'$ denote this map family. Conditions of the type $x \geq a$ are
    monotone while $x \leq a$ are anti-monotone.
\end{enumerate}

We often use Iverson bracket notation for describing maps to $\{0,1\}$. Every one-hot map can be viewed as $[x=a]$ while $[x \geq a]$ and $[x \leq a]$ are
common examples of monotone and anti-monotone maps, respectively.
In fact, every monotone (anti-monotone) map from a domain $D$ to $\{0,1\}$ equals $[x \geq a]$ ($[x \leq a]$) for some $a \in D$.

\begin{example} \label{ex:firstexample}
Let us consider the {\sc NAE-Integer-3-SAT} problem by Bringmann et al~\cite{Bringmann:etal:jcss2016}.
They define the problem as follows: one is given variables $\{x_1,\dots,x_k\}$ that
each take a value in $\{1, \dots, n\}$ and clauses $\{C_1, \dots, C_m\}$ of the form
${\rm NAE}(x_{i_1} \leq a_1, x_{i_2} \leq a_2, x_{i_3} \leq a_3)$,
where $a_1, a_2, a_3 \in \{1, \dots, n\}$, and such a clause is satisfied if not all three inequalities are true and
not all are false (i.e., they are ``not all equal''). The goal is to find an assignment of the variables
that satisfies all given clauses. In our notation this problem can be formulated as $\udcsp(\{R_{\mathrm{NAE}}\},\mo')$
where $R_{\mathrm{NAE}}=\{0,1\}^3 \setminus \{000,111\}$.
%and for all
%$n \geq a$, $\MM$ contains the map $f_{n,a}:\{0,\dots,n-1\} \rightarrow \{0,1\}$
%such that $f(x)=1$ if and only if $x \leq a$.
\end{example}

To obtain general results in this setting it is necessary to treat constraint languages in a unified way, i.e., to avoid proving hardness and tractability through exhaustive case analyses. However, it is far from clear how e.g.\ universal algebra can be used for this purpose since there are numerous complicating factors --- whether equality is allowed on the relational side, that the underlying map family needs to be taken into account, and to what extent it is possible to restrict the number of cases by considering certain symmetries --- and no existing algebraic theory fits this bill.

\subsubsection{A New Algebraic Approach}
\label{sec:intro_algebra}

The classical algebraic theory for CSPs is fundamentally based on associating each set of relations $\Gamma$ with a closure operator $\langle \Gamma \rangle$ containing all {\em primitive positive definable} (pp-definable) relations, i.e., if $R$ can be expressed as the set of models of $\exists y_1, \ldots, y_n \colon \varphi(x_1, \ldots, x_r, y_1, \ldots, y_n)$ where $\varphi(\cdot)$ is a conjunctive, equality-free formula over atoms from $\Gamma \cup \{\eq_D\}$, where $\eq_D$ is the equality relation over the domain $D$. The main point of introducing this closure operator is that the sets of the form $\langle \Gamma \rangle$ can be dually defined as sets of {\em polymorphisms} $\pol(\Gamma)$, i.e., all functions over $D$ such that $f$ applied component-wise to tuples in $R$ stays inside $R$ for every $R \in \Gamma$. Intuitively, a polymorphism then provides permissable ways to combine solutions of a CSP$(\Gamma)$ instance, and via the CSP dichotomy theorem~\cite{Bulatov:focs2017,Zhuk:jacm2020} we know that the only source of tractability of CSP$(\Gamma)$ is a non-trivial way to combine solutions, which can always be witnessed by an operation in $\pol(\Gamma)$. Unfortunately, nearly all aspects of pp-definitions (and, as a consequence, polymorphisms) are unusable in the $\udcsp(\Gamma, \MM)$ setting, since 
\begin{enumerate}
\item we cannot allow existentially quantified variables since our complexity parameter $|V| = k$ blows up if one does the typical gadget replacement and introduces fresh variables for each constraint, 
\item cannot allow the equality constraint to be used implicitly since $\udcsp(\Gamma, \MM)$ and $\udcsp(\Gamma \cup \{\eq_D\}, \MM)$ may have different complexity, and
\item atoms $R(\cdot)$ should be allowed to use maps from $\MM$.
\end{enumerate}

Some, or at least parts of these difficulties, have been addressed before, and lead to generalizations of the classical algebraic theory.

\begin{enumerate}
  \item By removing existential quantification we get {\em partial polymorphisms} rather than totally defined polymorphisms, i.e., if $f$ is a partial polymorphism we only require the component-wise application of $f$ to tuples of a relation $R$ to stay inside $R$ if each application is defined.
  \item If we additionally do not allow equality then the basic object is a {\em partial multipolymorphism}, i.e., a function $D \to 2^{D}$ where $2^D$ is the powerset of $D$, and we identify the condition that $f$ returns the empty set to mean that it is undefined.
  \item The generalization to atoms of the form $R(m_1(x_1), \ldots, m_r(x_r))$ for maps $m_i \colon \numdom{n} \to \numdom{d}$ (where $\numdom{n}$ is the domain of the relation that we are defining and $\numdom{d}$ the domain of the relation $R$ in the base language) is largely unexplored but has been considered for the specific case when all maps are over the same domain, and when equality is allowed, called {\em functionally guarded pp-definitions\footnote{While Carbonnel refers to them as pp-definitions they can be formulated in a quantifier-free way via our maps terminology.}}~\cite{carbonnel2022Redundancy}. On the functional side we then get partial operations satisfying {\em polymorphism patterns} which can be viewed as a restricted form of identities that were instrumental in resolving the CSP dichotomy theorem.
\end{enumerate}

By combining these different notions we land on the definability notion {\em equality-free functionally guarded primitively positively definable using maps from $\MM$} which we for the sake of brevity refer to as {\em $\MM$-fgpp-definitions} but stress that the definitions are quantifier-free. Naturally, the complicated aspect in this algebraic undertaking is not to combine the notions, which is trivial, but to investigate whether such an assortment of seemingly unrelated requirements actually leads to a useful theory. In particular, arbitrary sets of partial polymorphisms are poorly understood even in the Boolean domain, and virtually nothing general is known about partial multifunctions for arbitrary finite domains. To have any chance of applying the theory we thus need to be able to make simplifying assumptions and describe partial multifunctions in a more effective way.

Here, we dare say that our theory is surprisingly well-behaved even if no particular assumptions are made on $\MM$, and remarkably powerful in the monotone setting. The basic idea to get a functional correspondence to $\MM$-fgpp-definitions over $\Gamma$ is to not directly consider the partial multipolymorphisms of $\Gamma$ ($\mpol(\Gamma)$) but rather to interpret these over different domains with the map family $\MM$ ($\mpol(\Gamma, \MM)$), and we say that each such function {\em $\MM$-preserves} $\Gamma$. This can formally be achieved by generalizing the well-known algebraic notion of {\em concrete homomorphic image} to the multifunction setting, as well as an inverse operation which generalizes {\em retractions}~\cite{barto2018}. We obtain the following general statement.

\begin{theorem} \label{thm:intro1}
(Informal)
Let $\Gamma$ be a constraint language and $R$ a relation. Then $\Gamma$ $\MM$-fgpp-defines $R$ if and only if each multifunction in $\mpol(\Gamma, \MM)$ preserves $R$.
%VL: I formulated it like this to avoid having to introduce the identity map family here.
%$\mpol(\Gamma, \MM) \subseteq \mpol(R)$.
\end{theorem}

For the specific case when $\MM$ contains all maps then $\mpol(R, \MM)$ collapses to partial multifunctions definable by polymorphism patterns, and when $\MM = \mo$ contains all monotone maps then we get a description in terms of {\em ordered polymorphism patterns}. Abstractly, these can be defined by systems of identities of the form $f(\ldots) \approx x$ where the term $f(\ldots)$ on the left-hand side is not nested, and where the variables are equipped with an ordering, and we can construct partial multifunctions over some concrete domain $\numdom{d}$ by considering monotone maps from the pattern variables to $\numdom{d}$. This, for example, makes it possible to formulate {\em median}, {\em minimum} and {\em maximum} operations in a straightforward way. Here it might be interesting to note that the latter two operations cannot be defined via the classical algebraic toolbox~\cite{DBLP:conf/dagstuhl/BartoKW17} but are instead treated as special instances of {\em semilattice} operations. This is not a problem when studying the classical complexity of CSPs but preserving the order becomes crucial in the $\udcsp(\Gamma, \mo)$ setting.

%\begin{enumerate}
%  \item Classical algebraic approach is based on polymorphisms and pp-definitions.
%\item We need multifunctions instead of functions.
%\item We do not allow existential quantification (partial multifunctions).
%\item allow maps. We have carbonnell using maps over same domain. But far from clear that we get something well behaved if we don't have all maps. 
%\item Far from clear that this is well-behaved. If the map family is very weak we obtain just partial multifunctions as a degenerate case, which is not well understood. 
%\item Remarkably, all of this is greatly simplified by considering homomorphic images (not explored in the partial multifunction setting). Since we are considering multifunctions also the inverse operation is well-defined. This generalizes patterns and allows for even more precise statements which we return to in the monotone setting. 
%\end{enumerate}

\subsubsection{One-Hot CSP}
As a straightforward first problem we consider $\udcsp(\Gamma, \oh)$. We prove that the parameterized complexity of this problem coincides with the parameterized complexity of finding a satisfying assignment of weight $k$~\cite{Marx05CSP}, which is FPT if and only if the language is {\em weakly separable}.
%where $\Gamma$ is Boolean and $\oh$ is the set of one-hot maps of the form $m(x) = 1$ for exactly one $x \in \numdom{n}$. 
%We represent these maps by Iverson brackets of the form $[x = a]$. Then, we may immediately observe that in a $k$-variable instance of $\udcsp(\Gamma, \oh)$ at most $k$ of these literals may be simultaneously true, and thus have a superficial similarity to the Boolean satisfiability problem parameterized by the weight (the number of ones) of the satisfying assignment. The parameterized complexity of this problem was systematically studied by Marx~\cite{Marx05CSP} who obtained an FPT/W[1] dichotomy via the criterion {\em weak separability}. This criterion was later~\cite{Kratsch:etal:toct2016} refined into a partial polymorphism condition based on two operations and we prove that these two operations also govern the parameterized complexity of $\udcsp(\Gamma, \oh)$.

\begin{theorem} \label{thm:intro2}
Let $\Gamma$ be a Boolean constraint language. Then $\udcsp(\Gamma, \oh)$ is FPT if $\Gamma$ is weakly separable and W[1]-hard otherwise.
\end{theorem}

We also resolve the classical complexity of $\udcsp(\Gamma, \oh)$ which, curiously, differs from the complexity of finding a satisfying assignment of weight $k$~\cite{Marx05CSP}. Hence, the two problems are different but happen to follow the same parameterized complexity, and the most important message from Theorem~\ref{thm:intro2} is that we obtain natural and relevant problems by restricting map families. Importantly, $\udcsp(\Gamma, \oh)$ seems to be a robust problem to reduce from and to, and in comparison to weighted SAT one can avoid finicky reductions that need to preserve the number of ones.
%
%for other parameterized problems where weak separability decides FPT. Given that our dichotomy in several crucial steps are often easier than those of Marx~\cite{Marx05CSP}  could simplify parameterized complexity dichotomies. 
%\todo[inline]{Could we make any similarities between reducing from clique vs multicoloured clique? I want to argue that our problem is easier to use as an auxilliary problem than Marx's problem, which can be messy to reduce from/to.}

\subsubsection{Monotone CSP}
Let us now turn to the main contribution of the paper, namely, $\udcsp(\Gamma, \mo)$ where $\mo$ is the set of monotone maps. In light of Theorem~\ref{thm:intro1} we expect the complexity of this problem to be captured by (partial) polymorphisms definable by systems of identities with order, and indeed, this proves to be the case for the P/NP dichotomy.

\begin{theorem} \label{thm:intro3}
    Let $\Gamma$ be a constraint language. If $\Gamma$ is $\mo$-preserved by $\pmin$, $\pmax$ or $\pmedian$, then $\udcsp(\Gamma, \mo)$ is in $P$. Otherwise, $\udcsp(\Gamma, \mo)$ is NP-hard.
\end{theorem}

In addition, up to FPT-reductions, we show that for every language $\Gamma$ such that $\udcsp(\Gamma,\mo)$ is \NP-hard, $\udcsp(\Gamma,\mo)$ and $\udcsp(\Gamma,\mo \cup \mo')$ have the same complexity. Thus we may assume that we are working with both monotone and anti-monotone maps.

For the parameterized tractability, we first note that if $\Gamma$ can fgpp-define all permutations\footnote{To be precise we require for every permutation $\sigma \colon \numdom{n} \to \numdom{n}$ that $\Gamma$ can $(\mo \cup \mo')$-fgpp-define the {\em graph} $\sigma^\bullet = \{(x, \sigma(x)) \mid x \in \numdom{n}\}$.}, then $\udcsp(\Gamma,\mo \cup \mo')$ is W[1]-hard through standard reductions.
We therefore seek an algebraic criterion that serves as an obstruction towards this.  
We define an ordered polymorphism pattern, the \emph{connector polymorphism}, that serves this role. 
%The tractability part is based around local consistency and the NP-hardness is intrinsically linked to  definability of certain permutations. Inspired by this we begin the parameterized classification of defining an algebraic criterion which determines whether $\Gamma$ can fgpp-define all permutations\footnote{To be precise we require for every permutation $\sigma \colon \numdom{n} \to \numdom{n}$ that $\Gamma$ can $\oh$-fgpp-define the {\em graph} $\sigma^\bullet = \{(x, \sigma(x)) \mid x \in \numdom{n}\}$.} with monotone or anti-monotone maps. This question, in turn, can be related to definability of the so-called {\em flip permutation} $F((x,y)) = (y,x)$ over domain $[\numdom{n}]^2$, which is a {\em universal permutation} for permutations over $\numdom{n}$. We define an obstruction towards fgpp-definability of all permutations by the {\em connector polymorphism}. 
This operation can be abstractly defined over $\numdom{3}$ as $\pcon(0, 0, 1, 2, 2) = 0, \pcon(0, 2, 1, 0, 2) = 1$, $\pcon(2, 2, 1, 0, 0) = 2$, and $\pcon(2, 0, 1, 2, 0) = 1$, and via the algebraic machinery described in Section~\ref{sec:intro_algebra} this pattern can be interpreted over any domain $\numdom{d}$ to produce a concrete operation $\fcon_d$. 
We also define a canonical W[1]-hardness relation $R_3=\{00, 02, 11, 20, 22\}$, and note the following (somewhat informally).
\begin{enumerate}
    \item $R_3$ can $\mo$-fgpp-define the graph of every permutation.
    \item Every language $\Gamma$ not preserved by $\pcon$ can $(\mo \cup \mo')$-fgpp-define $R_3$.
\end{enumerate}
%The first evidence that the connector polymorphism indeed captures the borderline between FPT and W[1]-hardness is then given by the following theorem. 
We get the following theorem. 
Recall by Theorem~\ref{thm:intro3} that if $\Gamma$ is invariant under $\min$ or $\max$ then $\udcsp(\Gamma, \mo)$ is in P, so the assumption in this theorem does not diminish the hardness claim.

\begin{theorem} \label{thm:intro4}
  Let $\Gamma$ be a constraint language not $\mo$-preserved by $\pmin$ or $\pmax$. If $\pcon \notin \mpol(\Gamma,\mo)$ then $\udcsp(\Gamma, \mo)$ is W[1]-hard.
\end{theorem}

In the other direction, we show that for every finite language $\Gamma$ preserved by $\pcon$, there is a permutation $\sigma$ not fgpp-definable in $\Gamma$. 
In fact, we show something that on the surface may seem stronger. 
As mentioned, \emph{twin-width} is a notion of structural complexity that is very well suited for (binary) structures with an order, and the canonical example of a binary structure with unbounded twin-width is precisely the ability to define all permutations~\cite{Bonnet:etal:jacm2022}.
We use the notion of twin-width over ordered domains, which is particularly well-behaved~\cite{bonnet2024twin}.
There is then a notion of \emph{grid-rank}, which serves as a canonical obstruction towards a binary structure having bounded twin-width.
%\footnote{Grid-rank is functionally equivalent to twin-width~\cite{bonnet2024twin} and can for a matrix be defined as the the largest $k$ such that a matrix can be divided into $k^2$ zones (using $k-1$ vertical and horizontal lines) where each zone has at least $k$ distinct rows or columns. The grid-rank of a binary relation over an ordered domain is then just the grid-rank of its defining matrix.}
Then we show, by a surprisingly simple argument, that for every binary language $\Gamma$ preserved by $\pcon$, the set of binary relations $(\mo \cup \mo')$-fgpp-definable over $\Gamma$ has constant grid-rank, and thus bounded twin-width. 
By results from~\cite{bonnet2024twin}, and a simple dynamic programming algorithm,
this is sufficient for a tractability result for binary languages. 

% We complement this by an FPT result for {\em binary} relations invariant under the connector polymorphisms. A priori, it is not clear why this assumption is necessary, since CSPs over binary relations are as expressive as CSPs over arbitrary constraint languages. However, as we will show, this question is linked to the (likely very hard) task of finding a correct generalization of twin-width~\cite{bonnet2024twin} to arbitrary non-binary, ordered structures. 

\begin{theorem} \label{thm:intro5}
  Let $\Gamma$ be a binary constraint language. If $\pcon \in \mpol(\Gamma,\mo)$ then $\udcsp(\Gamma, \mo)$ is in FPT.
\end{theorem}

% The algebraic part of the proof is based on the observation that the connector polymorphism implies that the \thegraph graph of any instance has bounded {\em grid-rank}\footnote{Grid-rank is functionally equivalent to twin-width~\cite{bonnet2024twin} and can for a matrix be defined as the the largest $k$ such that a matrix can be divided into $k^2$ zones (using $k-1$ vertical and horizontal lines) where each zone has at least $k$ distinct rows or columns. The grid-rank of an ordered graph is then just the grid-rank of its adjacency matrix.} which, using a known construction~\cite{bonnet2024twin}, can be converted to a contraction sequence. The final ingredient in the proof is then a twin-width style dynamic programming algorithm over the contraction sequence. 

Extending Theorem~\ref{thm:intro5} to relations of arbitrary arity meets unexpected difficulties. Twin-width is intrinsically a property of binary structures\footnote{and non-binary structures with \emph{first-order transductions} from binary structures of bounded twin-width~\cite{Bonnet:etal:jacm2022}.}. Furthermore, based on conversations with experts in the area, establishing a useful width notion for non-binary structures that generalizes twin-width is considered extremely challenging, or possibly impossible, given the tendency of non-binary structures to very rapidly become uncontrollably expressive. 
Thus we have three options -- hoping that no tractable non-binary languages exist; moving on from twin-width to other width notions more suitable for non-binary structures; or developing such a width-notion ourselves to suit our needs. 

The first is a non-starter. Consider the basic 1-in-3 relation $R_{1/3}=\{100, 010, 001\}$.
It is easily verified to be preserved by $\fcon_2$, so the W[1]-hardness proof of Theorem~\ref{thm:intro4} does not apply. There is also no obvious way to reduce $\udcsp(R_{1/3}, \mo)$ to a binary structure while preserving bounded twin-width. In fact, the FPT-status of $\udcsp(R_{1/3}, \mo)$ is left as an open problem.

For the second, the most natural parameter to investigate is {\em clique-width} or one of its functionally equivalent parameters such as {\em boolean width} or {\em rank-width}. Rank-width is possible to extend to non-binary relations so a reasonable first question is whether
the (graphs/matrices corresponding to a) set of binary, $\mo$-fgpp-definable relations over a relation $R$ preserved by the connector polymorphism has bounded rank-width or not.  Here, we give a simple relation $R = \{01, 10, 12, 21\}$ over $\numdom{3}$ which together with Boolean implication $\{00,01,11\}$ can be used to give a construction in terms of {\em diamond graphs} which has unbounded rank-width (Theorem~\ref{thm:unbounded_rankwidth}). 

Thus, we are left with designing a non-binary width parameter that generalizes twin-width. 
We consider such a parameter, inspired by the properties of $\pcon$, based on extending grid-rank to non-binary structures which we call {\em projected grid-rank}.
Towards the definition of this parameter we first note that simply excluding $k$-grids is not enough: the relations $S_n(x,y,z) \equiv (x + y = z)$ over all domains $\numdom{n}$ together with unary relations give a W[1]-hard CSP but, viewed as a 3-dimensional grid, each relation does not even contain a $2 \times 2 \times 2$ subgrid where every cell is non-constant. 
Instead, we define a \emph{projected grid}
  of $R \subseteq \numdom{n}^r$ w.r.t.\ a partition $A \cup B$ of $[r]$ to be a relation $R' \subseteq S \times T$ where $S$ and $T$ are two ordered subsets of $\numdom{n}^A$, and  $\numdom{n}^B$, respectively, such that $R'=\{(s,t) \in S \times T \mid s \cup t \in R\}$
where $s \cup t$ denotes the tuple in $\numdom{n}^r$ whose entries in coordinates $A$
match $s$ and whose entries in coordinates $B$ match $t$. We then say that $R$ has \emph{projected grid-rank $k$} if it has a
  projected grid with grid-rank at least $k$. The definition generalizes to constraint languages in the obvious way and we then relate bounded projected grid-rank to the connector polymorphism as follows.

\begin{theorem}  \label{thm:intro6}
  Let $\Gamma$ be a finite base language. The set of $(\mo \cup \mo')$-fgpp-definable relations over $\Gamma$ has bounded projected grid-rank if and only if $\Gamma$ satisfies the connector property.
\end{theorem}

If we return to the troublesome relation $S_n(x,y,z) \equiv (x + y = z)$ over $\numdom{n}$ 
then we find that (1) it has the connector property, but (2) its projective grid-rank is growing with $n$. Thus, any finite language $\Gamma$ with the connector property fails to fgpp-define $S_n$ for some $n \in \N$. (This is analogous to the situation with permutations -- every individual permutation has the connector property, but any finite language that can fgpp-define all of them does not.)

All evidence thus far suggests that the connector property is the correct boundary for tractability, even for higher-arity relations. We give one further piece of evidence.  
We show that for any binary \emph{projection} of a $k$-ary relation with bounded projected grid rank, the projected relation has bounded grid-rank (and thus bounded twin-width). 
This in particular rules out every W[1]-hardness proof we can think of for a language with the connector property. 
%, and naturally raises the question of whether there is a concrete link between projected grid-rank and twin-width. Again, since there is no known generalization of twin-width to non-binary structures, we investigate whether bounded twin-width of the binary projections coincides with bounded projected grid-rank (which by Theorem~\ref{thm:intro6} coincides with the connector property). We prove the following (recall that twin-width is functionally equivalent to grid-rank).

\begin{theorem} \label{thm:intro7}
    There is a function $f \colon \N \times \N \to \N$ 
  such that the following holds:
  For all $d, k \in \N$ and every relation $R \subseteq \numdom{n}^k$
  with projected grid-rank at most $d$, over $\numdom{n}$
  each binary projection of $R$ has grid-rank at most $f(d,k)$.   
\end{theorem}

The proof is based on a multi-layered Ramsey argument which uses the existence of large high-rank grids (analogous to applications of the Marcus-Tardos theorem in twin-width theory~\cite{bonnet2024twin,Bonnet:etal:jacm2022,PilipczukSZ22compact}) with the requirement that higher dimensions exhibit regular lexicographic-like ordering, inspired by the two-dimensional Erdős-Szekeres theorem of Fishburn and Graham~\cite{FishburnG93lexicographic}. 
We defer the details to the main paper. 
% The technical core involves induction over {\em weak projected grids with hidden dimensions} where we systematically reduce the dimensionality by exploiting either cells with diverse labels (leading to immediate rank reduction) or cells with many labels having disjoint spans (enabling application of two-dimensional Erdős-Szekeres to recover lexicographic structure). This is applied via a product Ramsey argument~\cite{Bodirsky:Book} which ensures that we either have the desired subgrid or enough combinatorial regularity to iterate the reduction process.
See Figure~\ref{fig:property_table} for a summary of our width results. Based on the strong connection between the connector property and bounded width we pose the following explicit conjecture (note that the hardness part of the conjecture already follows from our results and that we have verified the FPT part for binary relations).

\begin{conjecture} \label{conjecture}
$\mcsp(\Gamma)$ is FPT if and only if $\Gamma$ has the connector property.
\end{conjecture}

%lemma:projected-grid-rank
\begin{figure}[tbp]
\centering
\begin{tabular}{@{}lllll@{}}
\toprule
\textbf{Property} & \textbf{fgpp-definability} & \textbf{Projected GR} & \textbf{TW (binary proj.)} & \textbf{RW} \\
\midrule
\textbf{Connector } &   \begin{tabular}[t]{@{}l@{}}
                             Not all permutations \\
                             \textit{Lemma~\ref{lemma:all_permutations}}
                           \end{tabular}  & \begin{tabular}[t]{@{}l@{}} Bounded \\ \textit{Lemma~\ref{lemma:projected-grid-rank}} \end{tabular} & \begin{tabular}[t]{@{}l@{}}
                                                           Bounded \\ \textit{Corollary~\ref{corollary:twin_width_bounded}} 
                                                         \end{tabular} & \begin{tabular}[t]{@{}l@{}} Unbounded \\ \textit{Theorem~\ref{thm:unbounded_rankwidth}} \end{tabular}\\
\midrule
\textbf{Not connector} & \begin{tabular}[t]{@{}l@{}} All permutations \\ Lemma~\ref{lemma:not_line_gives_r3} \end{tabular} & \begin{tabular}[t]{@{}l@{}} Unbounded \\ \textit{Lemma~\ref{lemma:projected-grid-rank}} \end{tabular} & \begin{tabular}[t]{@{}l@{}}
                                                             Unbounded \\ \textit{Corollary~\ref{corollary:twin_width_bounded}}
                                                           \end{tabular} & \begin{tabular}[t]{@{}l@{}} Unbounded \\ \textit{Theorem~\ref{thm:unbounded_rankwidth}} \end{tabular}\\
\bottomrule
\end{tabular}
\caption{Connections between the connector polymorphism and width-properties. In the table, GR stands for grid-rank, TW for twin-width, and RW for rank-width.}
\label{fig:property_table}
\end{figure}

\subsection{Applications} \label{sec:applications}

%\todo[inline]{PJ: Suggestions. Start with graph separation (we have unifying framework/algorithm) and continue with improved algorithm for bijunctive Min2CSP. This only uses monotone maps and focuses on FPT/W[1]-results. Continue with signed logics --- here we additionally use one-hot maps and focus on P/NP-results. Try to use Bringmann et al. and Lenstra in some entertaining way.}

We finish the introductory part by highlighting some interesting
applications of the $\udcsp$ framework.
We
begin with the parameterized complexity of graph separation problems and its connections to $\mcsp$. We broaden the perspective to other map families and computational complexity in our second example where we consider
many-valued logics.

\paragraph{Graph separation.}
Let us consider one of the most natural applications of $\udcsp$:
min-cuts in digraphs. 
  Hence, let $G$ be a digraph with $s, t \in V(G)$, and 
   let $\cP=\{P_1, \ldots, P_k\}$ be an arc-disjoint $st$-max flow in $G$. 
   For
  each path $P_i \in \cP$, we enumerate the edges of $P_i$ from $s$ to $t$
  as $e_{i,1}, \ldots, e_{i,m_i}$.
  We claim that there is an
  instance of $\udcsp(\Gamma,\mo)$ over the Boolean
  language $\Gamma=\{\textsf{Impl}, 0\}$ with $k$ variables whose solutions are in
  bijection with the $st$-min cuts in $G$ (where $\textsf{Impl} = \{00,01,11\}$ and where we write $0$ for the constant Boolean relation $\{0\}$).

We abuse notation slightly and let $V(\cP),E(\cP)$ denote the vertices and edges appearing in $\cP$, respectively.
We define a graph $G'$ where the paths in $\cP$ have been removed, i.e.
$V(G')=V(G) \setminus V(\cP)$ and $E(G')=E(G) \setminus E(\cP)$.
For every pair of
  vertices $u, v \in V(\cP) \setminus \{t\}$ and $i, j \in [k]$
  such that $u \in V(P_i)$, $v \in V(P_j)$, $(u,i) \neq (v,j)$,
  and there is a path from $u$ to $v$ in $G'$ with internal vertices
  (if any) from $V'$, let $e_{i,a}$ and $e_{j,b}$ be the edges of $P_i$
  following $u$ and $P_j$ following $v$, respectively, and
  add a constraint $\textsf{Impl}([X_i \geq a],[X_j \geq b])$. 
  Note that this includes the case $u=v$ if $i \neq j$. 
  
  The satisfying assignments of the resulting $\udcsp(\Gamma,\mo)$ instance are in bijection with
  $st$-min cuts in the following sense. If $Z \subseteq E(G)$ is an
  $st$-min cut, then for each $i \in [k]$ let $Z \cap E(P_i)=\{e_{i,n_i}\}$
  and let $\varphi$ be the assignment where $\varphi(X_i)=n_i$ for each $i$. 
  Then $\varphi$ satisfies all constraints, since otherwise there is a
  path in $G'$ from the $s$-side of $Z$ to the $t$-side of $Z$.
  Conversely, for every satisfying assignment $\varphi$ of the $\udcsp$ instance,
  let $Z=\{e_{i,n_i} \mid X_i=n_i, i \in [k]\}$. Then $Z$ is an $st$-min cut,
  since any path from the $s$-side of $Z$ to the $t$-side of $Z$
  implies a $\textsf{Impl}$-constraint violated by $\varphi$. 

With this correspondence in mind, we can view other min-cut
problems through the lens of $\udcsp$.
First consider a pair of
edges $e_{i,a}$ and $e_{j,b}$ under the constraint that a min-cut $Z$
should contain $e_{i,a}$ if and only if it contains $e_{j,b}$. This
can be captured by the binary relation $R = \{00, 02, 11, 20, 22\}$
extended with unary maps as follows:
\[
  R(f_a(X_i), f_b(X_j)) \quad \text{ where } \quad
  f_i(x) = 
  \begin{cases}
    0 & x < i \\
    1 & x = i \\
    2 & x>i.
  \end{cases}
\]
%We refer to the monotone maps $f_i$ as \emph{threshold maps} and let ${\rm Tm}$
%denote the set of them.
Now, $\udcsp(\{R,\mathsf{Impl},0\},\mo)$
contains the \textsc{Paired Min-Cut} problem and
is thus W[1]-hard. 
We consider a couple of variations on this theme.
Let $e_{i,a}=uv$ and $e_{j,b}=u'v'$ be two edges.

\begin{itemize}
\item Replacing $R$ with $R_a=\{00, 02, 11, 22\}$ yields a
  coordination constraint that additionally forbids that $uv$ is on
  the $s$-side and $u'v'$ on the $t$-side of the cut,
  as in the \textsc{3-Chain SAT} constraint $(u \to v \to u' \to v')$~\cite{ChitnisEM17,KimKPW25fa1}.
  
\item Replacing $R$ with $R_b=\{02, 11, 20, 22\}$ yields a
  coordination constraint that additionally forbids that both $uv$ and
  $u'v'$ are on the $t$-side of the cut,
  as in the \textsc{Coupled Min-Cut} problem~\cite{KimKPW23fa3}.
  
  \item Replacing $R$ with $R_c=\{00, 02, 11, 20\}$ yields a
  coordination constraint that additionally forbids that both $uv$ and
  $u'v'$ are on the $s$-side of the cut (as a kind of dual to
  \textsc{Coupled Min-Cut}).
\end{itemize}

It is readily checked that each of these three relations satisfies the connector property, hence $\udcsp(\{R_a,R_b,R_c,\mathsf{Impl},0\}, \mo)$ is FPT.
While this does not represent a completely new algorithm for these problems,
since the above approach is only applicable after the exhaustive application of flow augmentation, it is interesting that we can solve them with a unified method, and that the resolution (unlike the original algorithms) involves the application of bounded twin-width.

In Section~\ref{sec:boolean_mincsp}, we give a more exhaustive application of these methods, showing how a significant portion of the \textsc{Boolean MinCSP} algorithm of Kim et al.~\cite{KimKPW23fa3} can be captured and simplified using a $\udcsp(\Gamma,\mo)$ formulation, where $\Gamma$ is a binary language with the connector polymorphism.

\paragraph{Many-valued logic.}
%The previous graph separation examples have focused on parameterized complexity and %monotone maps in the $\udcsp$ framework.
%We now take the opportunity 
%to broaden the perspective and consider $\udcsp$ based on other map families and
%with computational complexity in mind. Our starting point is a particular approach to
%many-valued logic.
Consider clauses which are disjunctions of atomic propositions of the form $x \in S$,
where $x$ is a variable and $S$ is a set from a family of subsets ${\cal S}$ of a finite domain $D$. This is the basis for a variant of many-valued logics known as {\em signed logic};
in this context, the sets of $S$ are called signs. 
Clearly, Boolean propositional logic and the SAT problem can be viewed as having domain $\{0,1\}$
and where a positive literal $x$ corresponds to $x \in \{1\}$ while 
a negated literal $\neg x$ corresponds to $x \in \{0\}$.
A major motivation behind signed logics is that they offer a classical logic framework for working with many-valued logics~\cite{Hahnle:book94,Lu:etal:jar98,Murray:Rosenthal:fi94}.
%This framework is quite powerful: in a certain sense, signed CNF formulas
%can be viewed as a generic representation for finite-valued logics~\cite{Hahnle:jlc94}.
%A class of membership constraints is characterized by
%the family of sets that are allowed to occur as signs in the formulas. 
Computational aspects of signed logics have been studied extensively:
the survey by Beckert, Hähnle, and Manyá~\cite{Beckert:etal:survey2000}
covers early results
whereas more recent complexity-oriented results can be found in, for instance,
\cite{Ansotegui:Manya:ismvl2003,Charatonik:Wrona:ismvl2007,Chepoi:etal:ejc2010,Gil:etal:sicomp2008,Jonsson:Nordh:tocs2010}.

Most complexity results in the literature focus on CNF formulas.
Naturally, more complex formulas can be constructed if by using other
relations than clauses. If the domain $D$ is fixed --- the standard setting for signed logics --- then
the complexity of the satisfiability problem for all choices of allowed relations and signs follows immediately from the finite-domain CSP classification by Bulatov~\cite{Bulatov:focs2017}
and Zhuk~\cite{Zhuk:jacm2020}. While a fixed domain may be sufficient in some cases, there are indeed contexts
where domain flexibility is desirable; database applications and knowledge representation in AI are
conceivable examples. Another motivation is that certain families of many-valued logics
can be treated in a uniform manner. It is common that such a family is parameterized
by $k \geq 2$ and each logic has a value domain of size $k$. One example is the family  $P_2,P_3,\dots$
introduced by Post~\cite[Section~11]{Post:ajm21}
and other examples have been provided by, for instance, Gödel~\cite{Godel:aaww32} and
\L ukasiewicz (cf.~\cite[Section~III]{Karpenko:book}).
The idea is supported by Hähnle~\cite{Hahnle:jlc94}: signed logics are, in a computationally meaningful way, a generic representation for many-valued logics

By viewing signs as unary maps,
we can drop the restrictions above: the domain is
part of the input and the relations are not restricted to clauses.
It is common that
the signs have the form $\uparrow a = \{d \in D \; | \; d \succeq a\}$ and 
$\downarrow a = \{d \in D \; | \; d \preceq a\}$, where
$\prec$ is some partial order on $D$; formulas defined on such signs are called {\em regular} signed
formulas. 
If we go back to Example~\ref{ex:firstexample} and {\sc NAE-Integer-3-SAT}, we see that it is
a problem of this kind where the underlying partial order is a total order. 
Formulas with signs of the form $\{d\}$ for some $d \in D$
are referred to as {\em monosigned} formulas.
Note that monosigned
formulas are regular with respect to the empty ordering.
Consequently, regular formulas over total orders and monosigned formulas
are two extreme cases: the maximum number of element pairs
are related in the first case and the minimum number in the second case.
To analyze the complexity of the monosigned case, then we have to understand
the complexity of $\ohcsp(\Gamma)$ for Boolean $\Gamma$ since
every monosign can be expressed by a
one-hot map. A P/NP dichotomy follows directly from Theorem~\ref{thm:p-vs-np-oh}.
We can obtain a P/NP dichotomy for regular signed formulas over the totally ordered domain ${\mathbb N}$, too. Note first that $\uparrow a(x)$ corresponds to the monotone map $[x \geq a]$
while $\downarrow a(x)$ corresponds to the anti-monotone map $[x \leq a]$.
The connection with $\udcsp(\Gamma,\mo \cup \mo')$ is obvious and we get
a P/NP-dichotomy from Theorem~\ref{thm:p-vs-np-antimono}.

There are other results in the literature that, in one way or another, are connected to
signed logics.
One example is due to Bodirsky and
Mottet~\cite{Bodirsky:Mottet:lmcs2018}. They consider an infinite
domain $D$ (such as the natural numbers ${\mathbb N}$) and a {\em finite} set $U$
of unary relations $U_1,\dots,U_n \subseteq D$.
They prove a P/NP dichotomy for CSP$(\Gamma)$ where the finite constraint language $\Gamma$
only contains relations that are first-order definable in $\{U_1,\dots,U_n\}$.
This may be viewed as a variant of signed logic where
the allowed signs
are those that are first-order definable in $U$. Since $U$ is a finite set, this
naturally restricts the possible signs. For instance, $U$ can only contain
a finite number of monosigned $x \in \{d\}$ literals, so Bodirsky and
Mottet's result is incomparable to ours.

\bigskip

\paragraph{Roadmap.}
The rest of the paper is structured as follows. In Section \ref{sec:prelims}, we give an overview of  constraint satisfaction and parameterized complexity. In Section \ref{sec:algebra}, we 
formally introduce the $\udcsp$ problem together with the algebraic framework.
The complexity of $\udcsp$ is studied in Sections \ref{sec:always_hard}--\ref{sec:monotone} where we consider arbitrary maps, one-hot maps, and 
monotone maps, respectively. Figure~\ref{fig:summary} contains a summary of our complexity results.
We investigate $\mcsp(\Gamma)$ when $\Gamma$ is a
non-binary base language with the connector polymorphism in Section~\ref{sec:non_binary}.
We conclude the paper in Section~\ref{sec:discussion} with a discussion of our results and future research directions.

\begin{figure}[tbp]
\centering
\begin{tabular}{@{}lll@{}}
\toprule
\textbf{Map Type} & \textbf{Computational Complexity} & \textbf{Parameterized Complexity} \\
\midrule
\textbf{Unrestricted maps} & \begin{tabular}[t]{@{}l@{}}
                             in P and FPT if $\Gamma$ is essentially unary\\
                             \textit{Theorem~\ref{thm:arbitrary-dichotomy}}
                           \end{tabular} & (the same) \\
\midrule
\textbf{One-hot maps} & \begin{tabular}[t]{@{}l@{}}
                        in P iff $\Gamma \subseteq \langle \{\mathsf{OR}_2,0,1\} \rangle_{\neq}$ \\
                        \textit{Theorem~\ref{thm:p-vs-np-oh}}
                      \end{tabular} & 
                      \begin{tabular}[t]{@{}l@{}}
                        FPT iff $\Gamma$ weakly 0-separable \\
                        \textit{Theorem~\ref{thm:one-hot-fpt-w1}}
                      \end{tabular} \\
\midrule
\textbf{Monotone maps} & \begin{tabular}[t]{@{}l@{}}
                         in P iff $\Pol(\Gamma) \cap \{\min,\max,\median\} \neq \emptyset$ \\
                         \textit{Theorem~\ref{thm:p-vs-np-mono}}
                       \end{tabular} & 
                       \begin{tabular}[t]{@{}l@{}}
                         \textit{restriction: $\Gamma$ binary} \\
                         FPT iff $\Gamma$ has connector property \\
                         \textit{Theorem~\ref{thm:fpt-vs-w1-mono}}
                       \end{tabular} \\
\bottomrule
\end{tabular}

\caption{Complexity of $\udcsp(\Gamma,\MM)$ for various map families. Recall that $\Gamma$ is Boolean in the one-hot case. The parameterized
complexity of $\udcsp(\Gamma,\mo)$ for non-binary $\Gamma$ is discussed in Section~\ref{sec:non_binary}.}
\label{fig:summary}
\end{figure}

%%%%%%%%%%%%%%%%%%%%%%%%%%%%%%%%%%%%%%%%%%%%%%%%%%%%%%%%%%%%%%%%%%%%%%%%%%%%%%%%%%%%%%%%%%%%

\section{Preliminaries}\label{sec:prelims}

We let $\Nat=\{0,1,2,\dots\}$ denote the natural numbers
and 
for every non-zero $c \in $ we let $[c]=\{1,2,\dots,c\}$. This notation can be extended to ranges: for $a, b \in \Nat$, $[a,b]=\{a,a+1,\ldots,b\}$. Let $D$ and $E$ denote two sets.
For a map $g \colon E \to D$, we let $g^\bullet$ be the {\em graph} $\{(x, g(x)) \mid x \in E\}$ of $g$.
Let $\sigma:D \rightarrow D$ be a permutation. When convenient, we implicitly view it as its graph $\sigma^\bullet$,
and allow ourselves to write things like "permutation 021" when referring to the relation $\{00,12,21\}$.
We let $E^D$ denote the set of total functions from $D$ to $E$.

We use the following graph-theoretic terminology.
Let $G$ be an undirected graph.
We write $V(G)$ and $E(G)$ to denote the vertices and edges of $G$, respectively.
If $U \subseteq V(G)$, then the {\em subgraph of $G$ induced
by $U$} is the graph $G'$ with
$V(G')=U$ and $E(G')=\{ \{v, w\} \mid v,w \in U \; {\rm and} \; \{v, w\} \in E(G)\}$. 
We denote this graph by $G[U]$.
If $Z$ is a subset of edges in $G$, we write $G-Z$ to denote the graph~$G'$ with $V(G') = V(G)$ and $E(G') = E(G) \setminus Z$.
For $X,Y \subseteq V(G)$, 
an \emph{$(X,Y)$-cut} is a subset of edges $Z$
such that $G - Z$ does not contain a path with 
one endpoint in $X$ and another in $Y$.
When $X,Y$ are singleton sets $X=\{x\}$ and
$Y=\{y\}$, we simplify the notation and write $xy$-cut instead
of $(X,Y)$-cut. We say that a path $P$ is an $xy$-{\em path}  if its endpoints are the vertices $x$ and $y$. We sometimes consider directed graphs $G$ with vertex set $V(G)$ and arc set $A(G)$; the terminology introduced above generalizes to this case.

The rest of this section begins 
with a short primer on parameterized
complexity (Section~\ref{sec:param-comp}).
We continue by introducing some basics of constraint satisfaction problems (CSP) 
(Section~\ref{sec:csp}).
Section~\ref{sec:prelim_algebra} give a short overview of terminology and some results
concerning the algebraic approach for CSPs.

%\todo[inline]{PJ: More notation, assumptions, etc.}
%
%\begin{itemize}
%\item
%...
%
%\end{itemize}

\subsection{Parameterized Complexity}
\label{sec:param-comp}

We assume throughout the paper that the complexity classes
\Poly\ and \NP are distinct. Given two computational problems $L_1 \subseteq \Sigma_1^*$ and $L_2 \subseteq \Sigma_2^*$ with input alphabets $\Sigma_1,\Sigma_2$,
we write $L_1 \leq_{\rm poly} L_2$ if there is a polynomial-time many-one reduction
from $L_1$ to $L_2$. If $I$ is an instance of a problem $L$, then we let $||I||$
denote its size.

We will use the framework of 
{\em parameterized complexity}~\cite{book/DowneyF99}.
A parameterized problem is a subset of $\Sigma^* \times \Nat$ where $\Sigma$ is the input alphabet. 
The running time of an algorithm is studied with respect to the parameter
$p\in\Nat$ and the input size~$n$---the parameterized complexity of a problem is viewed as a function of both.
The underlying idea is that the parameter captures the structure of
the instance in a computationally meaningful way.
Reductions between parameterized problems thus need to take
the parameter into account. To this end, we will use {\em FPT-reductions}.
Let $L_1$ and $L_2$ denote parameterized problems with $L_1 \subseteq \Sigma_1^* \times {\mathbb N}$
and $L_2 \subseteq \Sigma_2^* \times {\mathbb N}$. 
An FPT-reduction from $L_1$ to $L_2$ is a
mapping $P: \Sigma_1^* \times {\mathbb N} \rightarrow \Sigma_2^* \times {\mathbb N}$
such that
\begin{enumerate}[(1)]
  \item
      $(x, k) \in  L_1$ if and only if $P((x, k)) \in L_2$, 
 \item
 the mapping can be computed in time $f(k) \cdot |x|^{O(1)}$ for some computable function $f : {\mathbb N} \rightarrow {\mathbb N}$, and
  \item there is a computable function $g : {\mathbb N} \rightarrow {\mathbb N}$ 
such that for all $(x,k) \in L_1$ if $(x', k') = P((x, k))$, then $k' \leq g(k)$.
\end{enumerate}

We write $L_1 \leq_{\rm FPT} L_2$ when
there is an FPT-reduction from $L_1$ to $L_2$.

The most favourable complexity class is \FPT\
(\emph{fixed-parameter tractable}),
which contains all problems that can be decided 
in $f(p)\cdot n^{O(1)}$ time with $f$ being some computable
function. It is easy to verify that if $L_1$ and $L_2$ are parameterized
problems such that $L_1$ FPT-reduces to $L_2$ and $L_2$ is in \FPT, then
it follows that $L_1$ is in \FPT, too.
The class \W{1} contains all problems that are FPT-reducible to \textsc{Independent Set} when parameterized
by the size of the solution, i.e. the number of vertices in the independent set.
Showing \W{1}-hardness (by an FPT-reduction) for a problem rules out the existence of a fixed-parameter
algorithm under the standard assumption $\FPT \neq \W{1}$. We assume this
assumption to be true from now on.
We will use the following W[1]-hard problem~\cite{Pietrzak03} several times
during the course of the paper.
  
   \problemDefP{Multicoloured Clique}
{ An undirected graph $G$, an integer $k \geq 1$, and a proper vertex colouring $\chi:V(G) \to [k]$}
{$k$}
{Does $G$ contain a clique of size $k$?}

\noindent
We will also use the following two close relatives (which both are 
W[1]-hard).

\problemDefP{Clique}
{ An undirected graph $G$ and an integer $k \geq 1$}
{$k$}
{Does $G$ contain a clique of size $k$?}

\problemDefP{Independent Set}
{ An undirected graph $G$ and an integer $k \geq 1$}
{$k$}
{Does $G$ contain an independent set of size $k$?}

\subsection{Constraint Satisfaction Problems}
\label{sec:csp}

A {\em (relational) structure} $\Gamma$ is a tuple $(D; \tau, I)$ where
$D$ is the {\em domain}, or {\em universe},
$\tau$ is a relational signature, and $I$ is a function from $\sigma$ to the
set of all relations over $D$ which assigns each relation symbol a
corresponding relation over $D$. If $R \subseteq D^k$ for some $k \geq 0$, then we write $\ar(R) = k$ for its arity. All structures in this paper are relational
and we assume that they have a finite signature unless otherwise stated.
Typically, we do not need to make a sharp distinction between relations and the corresponding relation symbols, so we usually simply write $(D; R_1, \ldots, R_m)$, where 
each $R_i$ is a relation over $D$, to denote a structure. 
We sometimes do not make a sharp distinction between structures and sets of relations when the signature is not important. We often use a compressed notation
for tuples in relations and write $a_1a_2\dots a_k$ instead of $(a_1,a_2,\dots,a_k)$.
We let $a_1a_2\dots a_k[i] = a_i$ for $1 \leq i \leq k$.

Let $\Gamma$ be an arbitrary $\tau$-structure $(D;R_1,R_2,\dots)$. The {\em constraint satisfaction problem} over $\Gamma$ (CSP$(\Gamma)$) is defined as follows.

\problemDef{CSP($\Gamma$)}
{$I = (V,C)$ where $V$ is a set of variables and $C$ a set of constraints of the form $R(x_1, \ldots, x_k)$ where $x_1, \ldots, x_k \in V$ and $R \in \Gamma$}
{Does there exist a function $f : V \rightarrow D$ that satisfies all constraints, i.e., $(f(x_1), \ldots, f(x_k)) \in R$ for all $R(x_1, \ldots, x_k) \in C$?}

We usually refer to the structure $\Gamma$ as the {\em constraint language}.
An important result was obtained independently by Bulatov~\cite{Bulatov:focs2017} and by Zhuk~\cite{Zhuk:jacm2020}: every finite-domain CSP$(\Gamma)$ is in P or it is NP-complete.

Without loss of generality, we work almost exclusively with numeric domains so we let $\numdom{d}$, $d \geq 1$,
denote the $d$-element set $\{0,\dots,d-1\}$.
We always assume that $\numdom{d}$ is implicitly ordered by the ordinary less-than
relation $<$.
%Given an instance $I=(V,C)$ of CSP$(\Gamma)$, we let $\Sol(I)$ be the set of solutions to $I$.
We say that a relation is {\em Boolean} if it has
domain $\numdom{2}=\{0,1\}$. We will use the following Boolean relations in the sequel. Let $m \geq 2$.

\begin{itemize}
\item
$\textsf{Eq}=\{00,11\}$,

\item
$\textsf{Impl}=\{00,01,11\}$,

\item
$\textsf{Or}_m=\{x_1\dots x_m \in \numdom{2}^m \; | \; 1 \in \{x_1,\dots,x_m\}\}$,

\item
$\textsf{Nand}_m=\{x_1\dots x_m \in \numdom{2}^m \; | \; 0 \in \{x_1,\dots,x_m\}\}$, and

\item
$\mathsf{Even}_m=\{x_1\dots x_m \in \numdom{2}^m \; | \; \sum_{i=1}^m x_i=0 \pmod 2\}$.
\end{itemize}

Standard tools for solving or simplifying instances of a CSP are local consistency-enforcing algorithms. The basic idea is to reduce variable domains or tighten constraints in a solution-preserving way.
Let $I=(V,C)$ be an instance of CSP$(\Gamma)$ where $\Gamma$ has
domain $D$ and $V=\{x_1,\dots,x_k\}$. We associate a variable domain $D_i \subseteq D$ 
with every variable $x_i \in V$.
We say that $I$ is {\em arc-consistent} if
 for every variable $x_i$ and every value $d \in D_i$,
  every constraint involving $x_i$ has a tuple that is consistent with $x_i=d$.
  We can easily (and in polynomial time) extract the reduced variable domains
  from a CSP instance $I=(V,C)$. 
  Begin by setting $D_i=D$ for $1 \leq i \leq k$. If a variable $x_i \in V$ is not consistent
  with respect to value $d \in D_i$ and the current variable domains $(D_1,\dots,D_k)$, then let $D_i := D_i \setminus \{d\}$.
  Repeat this until a fixed point is reached, and then output the vector $(D_1,\dots,D_k)$.
  Clearly, if $f:V \rightarrow D$ is a solution to $I$, then $f(x_i) \in D_i$ for every $x_i \in V$. In particular, if $D_i = \emptyset$ for some $x_i \in V$, then $I$
  has no solution.

A stronger notion of consistency is \emph{singleton arc consistency (SAC)},
which ensures that for every variable $x_i$ and every value $d \in D_i$,
the unary constraint $(x_i=d)$ is not rejected by arc consistency;
i.e., adding $(x_i=d)$ and enforcing arc consistency does not lead to
an instance with an empty variable domain. See~\cite[Chapter~5.6]{DBLP:conf/dagstuhl/BartoKW17}.
SAC can be enforced in polynomial time, if the maximum arity $r$ of a
constraint is bounded; in particular, for binary languages with $k$
variables, $m$ constraints and domains of size $n$, SAC can be
enforced in $O(mkn^3)$ time~\cite{bessiere2011efficient}.

%More general kinds of local consistency can be captured via a concept known as $(k,l)$-{\em minimality}. We follow the presentation in~\cite[Section~5.5]{DBLP:conf/dagstuhl/BartoKW17}.
%Let $1 \leq k \leq l$ be integers. An instance of the CSP is $(k, l)$-minimal if
%\begin{enumerate}
%\item
%no scope of a constraint contains repeated variables,
%
%\item
%every $l$-element set of variables is within the scope of some constraint, and
%
%\item
%for any at most $k$-element set of variables $W$ and any two constraints whose scope contains
%$W$, the projections of these constraints onto $W$ coincide.
%\end{enumerate}
%
%For fixed $k$ and $l$, there is a straightforward polynomial-time algorithm $A_{k,l}$ that 
%transforms any CSP instance into a $(k, l)$-minimal instance with the same set of
%solutions. Similarly to the arc consistency algorithm, this algorithm can also transform
%$\udcsp$ instances into $(k,l)$-minimal instances in polynomial time.
%Let $I$ denote an arbitrary instance of CSP$(\Gamma)$ and
%note the following: (1)
%if $A_{k,l}(I)$ contains an empty constraint relation,
%then $I$ is not satisfiable and (2)
%the instance $A_{k,l}(I)$ is not necessarily an instance of CSP$(\Gamma)$.
%We say that a constraint language $\Gamma$ has {\em width} $(k,l)$ if $A_{k,l}$ refutes every unsatisfiable instance of CSP$(\Gamma)$ (in the sense
%that an instance $I$ of CSP$(\Gamma)$ is not satisfiable if and only if $A_{k,l}(I)$ contains an empty constraint relation).

\subsection{Primitive Positive Definitions and Polymorphisms}
\label{sec:prelim_algebra}

We give a succinct summary of terminology and some results from the borderland
between logic and algebra. A more detailed treatment can be found in~\cite{DBLP:conf/dagstuhl/BartoKW17}.
Let $\Gamma$ be a structure with signature $\tau$ and domain $D$. 
First-order formulas over $\Gamma$ can be used to define relations: 
for a formula $\phi(x_1,\ldots,x_k)$ 
with free variables $x_1,\ldots,x_k$, the corresponding relation $R$
is the set of all $k$-tuples $(t_1,\ldots,t_k) \in D^k$
such that $\phi(t_1,\ldots,t_k)$ is true in $\Gamma$. 
In this case we say that $R$ is {\em first-order definable} in $\Gamma$.
Our definitions of relations are always parameter-free, i.e. we do not allow
the use of domain elements within them.
%Various restricted definitions are of central importance in this paper of which the most important one is that of {\em primitive positive definability}. Other types of definability are discussed in Section~\ref{sec:general-fpt}.

A certain restricted type of definition has gained a lot of attention. A $k$-ary relation $R$ is said
to have a {\em primitive positive definition} (pp-definition) over
$\Gamma$ if \[R(x_1, \ldots, x_{k}) \equiv \exists y_1, \ldots,
y_{k'}\, \colon \, R_1(\mathbf{x}_1) \wedge \ldots \wedge
R_m({\mathbf{x}_m})\] where each $R_i \in \Gamma \cup \{\eq_D\}$ and
each $\mathbf{x}_i$ is a tuple of variables over
$x_1,\ldots, x_{k}$, $y_1, \ldots, y_{k'}$ matching the arity of
$R_i$. Here, $\eq_D \: = \{(a,a) \mid a \in D\}$ is the equality relation over $D$.  Thus, $R$ is definable by a first-order formula consisting only
of existential quantification and conjunction over positive atoms from
$\Gamma$ and equality constraints. If $\Gamma$ is a set of relations
then we let $\cclone{\Gamma}$ be the inclusion-wise
smallest set of relations containing $\Gamma$ closed under
pp-definitions. The importance of pp-definitions in the context of constraint satisfaction problems stems from the following general reducibility condition.

\begin{theorem}[\cite{Jeavons:tcs98}] \label{thm:pp-red}
Let $\Gamma$ and $\Delta$ be finite structures with the same domain. 
If every relation of $\Gamma$ has a primitive positive definition in $\Delta$, then there is a polynomial-time reduction from CSP$(\Gamma)$ to CSP$(\Delta)$.
\end{theorem}

This result is the cornerstone of the algebraic approach to CSPs~\cite{barto2018} but is less useful in our setting since we are interested in studying CSPs where the domain is large but the number of variables are few. We thus also introduce the quantifier-free fragment of pp-definitions as follows.
A $k$-ary relation $R$ is said
to have a {\em quantifier-free primitive positive definition} (qfpp-definition) over
$\Gamma$ if \[R(x_1, \ldots, x_{k}) \equiv  R_1(\mathbf{x}_1) \wedge \ldots \wedge
R_m({\mathbf{x}_m})\] where each $R_i \in \Gamma \cup \{\eq_D\}$ and
each $\mathbf{x}_i$ is a tuple of variables over
$x_1,\ldots, x_{k}$, $y_1, \ldots, y_{k'}$ matching the arity of
$R_i$. We let $\cclone{\Gamma}_{\not \exists}$ be the
smallest set of relations containing $\Gamma$ closed under
qfpp-definitions.

In the context of reductions for CSPs the presence of equality constraints is not an issue since CSP$(\Gamma)$ and CSP$(\Gamma \cup \{\eq_D\})$ have the same classical complexity. However, this is not always true in the unbounded-domain, few variables setting, and to circumvent this we make the following definitions, where we say that pp-definition is {\em equality-free} if only atoms from $\Gamma$ are used.

\begin{definition}
  Let $\Gamma$ be a constraint language over a domain $D$.
  \begin{enumerate}
    \item 
      We let $\cclone{\Gamma}_{\neq} \subseteq \cclone{\Gamma}$ be the set of equality-free pp-definable relations over $\Gamma$.
    \item
      We let $\cclone{\Gamma}_{\nexists, \neq} \subseteq \pcclone{\Gamma}$ be the set of equality-free qfpp-definable relations over $\Gamma$.
  \end{enumerate}
\end{definition}

%\begin{theorem}[\cite{Jonsson:etal:jcss2017}] \label{thm:qfpp-red}
%Let $\Gamma$ and $\Delta$ be structures with the same domain.
%Assume that CSP$(\Delta)$ can be solved in time $t(n) \cdot poly(||I||)$ where $n$ is the number of variables
%in the given instance.
%If every relation of $\Gamma$ has a qfpp-definition in $\Delta$, then CSP$(\Gamma)$
%can be solved in time $t(n) \cdot poly(||I||)$, too.
%\end{theorem}

We now describe the corresponding algebraic objects.
An operation 
$f \colon D^{m}\to D$ is a {\em polymorphism} of a relation 
$R\subseteq D^{k}$ if, for any choice of $m$ tuples 
$(t_{{11}},\dotsc ,t_{{1k}}),\dotsc ,(t_{{m1}},\dotsc ,t_{{mk}})$ from $R$, it holds that the tuple obtained from these
$m$ tuples by applying $f$ coordinate-wise, i.e. 
$(f(t_{{11}},\dotsc ,t_{{m1}}),\dotsc ,f(t_{{1k}},\dotsc ,t_{{mk}}))$ 
is in $R$. 
Equivalently, $f \colon D^{m}\to D$ is a polymorphism of $R$ if it is a
homomorphism from the $m$-fold Cartesian
product of $R$ to $R$ itself.
If $f$ is a polymorphism of $R$, then we sometimes say that $R$ is {\em invariant} under $f$.
A constraint language $\Gamma$ has the polymorphism $f$ if every relation in $\Gamma$ has $f$ as a polymorphism.
We let $\Pol(\Gamma)$ denote the set of polymorphisms of $\Gamma$,
respectively.

If $F$ is a set of functions
from $D^m$ to $D$, then $\Inv(F)$ denotes the set of relations over $D$ that are invariant under the
functions in $F$.
There are close algebraic connection between the operators $\cclone{\cdot}$, $\Pol(\cdot)$, and
$\Inv(\cdot)$. For instance, if $\Gamma$ has a finite domain, then we have a Galois connection and an induced closure operator which gives us that $\cclone{\Gamma} = \Inv(\Pol(\Gamma))$~\cite{Geiger:pjm68}.

The relations in $\cclone{\Gamma}_{\nexists}$ can also be captured algebraically.
By an $n$-ary partial function,
we mean a map $f$ from $X \subseteq D^n$ to $D$, and we say that $f(d_1, \ldots, d_n)$ is {\em defined} if $(d_1, \ldots, d_n) \in X$ and {\em undefined}
otherwise. We sometimes call the set of values where the partial function is defined for the {\em domain} of the partial function.
We say that a relation $R \subseteq D^n$ is closed under a partial function $f \colon D^n \rightarrow D$ if $f$ applied componentwise to the tuples of $R$ always
results in a tuple from $R$ or an undefined result. More formally,
for each sequence of tuples $t_1, \ldots, t_m \in R$, either $f(t_1,\dots,t_m) \in R$ or there exists an $1 \leq i \leq n$ such that $(t_1[i], \ldots,t_n[i])$ is not included in the domain of $f$.

Let $\Gamma$ be a constraint language with domain $D$.
The set of all partial functions
preserving the relations in $\Gamma$, i.e., the {\em partial polymorphisms} of $\Gamma$, 
is denoted by $\ppol(\Gamma)$ and is called a {\em strong partial clone}.
We have a Galois connection and an induced closure operator which implies that  $\cclone{\Gamma}_{\nexists} = \Inv(\ppol(\Gamma))$~\cite{Geiger:pjm68}.

To additionally handle the equality-free definitions we need a polymorphism notion which does not necessarily preserve the equality relation $\eq_D$ (which is trivially preserved by every total and partial function). Here, the answer is to consider (partial) {\em multifunctions} instead of functions, i.e., functions of the form $f \colon D^n \to 2^{D}$, where $f(d) = \emptyset$ means that $f$ is undefined for $d \in D^n$. A multifunction $f$ where $|f(d)| \leq 1$ for each $d \in D^n$ is said to be {\em elementary} since it is equivalent to a (partial) function and we often do not make a sharp distinction between elementary multifunctions and (partial) functions.

For a domain $D$ we then define a \emph{(partial) multivalued polymorphism} over $D$ of arity $r$
as a function $f \colon D^r \to 2^D$. With slight abuse of notation, we also define $f(t_1,\ldots,t_r)$ for tuples $t_1, \ldots, t_r \in D^s$ (for some $s$) as
\[
  f(t_1,\ldots,t_r) := f(t_1,\ldots,t_r)[1] \times \ldots \times f(t_1,\ldots,t_r)[s].
\]
For a relation $R$ with domain $D$, we say that $f$ \emph{preserves} $R$ if
for all $t_1, \ldots, t_r \in R$ we have $f(t_1,\ldots,t_r) \subseteq R$, and $f$ preserves a set of relations $\Gamma$ if $f$ preserves each relation in $\Gamma$.

\begin{definition}
  For a set of relations $\Gamma$ we write $\mpol(\Gamma)$ for the set of all partial multivalued polymorphisms that preserve $\Gamma$ and $\tmpol(\Gamma)$ for the set of total multivalued polymorphisms of $\Gamma$.
\end{definition}

If $\Gamma = \{R\}$ is singleton we write $\mpol(R)$ rather than $\mpol(\{R\})$.
Similarly, for a set of (total or partial) multifunctions $F$ we let $\inv(F)$ be the set of relations invariant under each function in $F$, and we again get a useful closure operator which implies that $\cclone{\Gamma}_{\neq} = \Inv(\tmpol(\Gamma))$ and $\cclone{\Gamma}_{\nexists, \neq} = \Inv(\mpol(\Gamma))$, as well as the following inverse relationships~\cite{Geiger:pjm68}.

\begin{theorem}
  Let $\Gamma$ and $\Delta$ be two sets of relations over the same domain. Then $\Gamma \subseteq \cclone{\Delta}_{\neq}$ if and only if $\tmpol(\Delta) \subseteq \tmpol(\Gamma)$.
\end{theorem}

\begin{theorem}
  Let $\Gamma$ and $\Delta$ be two sets of relations over the same domain. Then $\Gamma \subseteq \cclone{\Delta}_{\nexists, \neq}$ if and only if $\mpol(\Delta) \subseteq \mpol(\Gamma)$.
\end{theorem}

%%%%%%%%%%%%%%%%%%%%%%%%%%%%%%%%%%%%%%%%%%%%%%%%%%%%%%%%%%%%%%%%%%%%%%%%%%%%%%%%%%%%%%%%%%%%

\section{The $\udcsp$ problem and its Algebraic Underpinnings} \label{sec:algebra}

%\todo[inline]{PJ: The title of this section sounds a bit weak since we also
%introduce a brand new computational problem.}

In this section we formally introduce the unbounded domain CSP problem, and present an algebraic framework suitable for studying (both parameterized and classical) complexity of this problem. Our theory is based on three considerations. 
 
First, recall from Section~\ref{sec:prelims} that the universal algebraic approach to CSPs on the relational side is built on primitive positive definitions (pp-definitions) where equality is always allowed. This is generally unproblematic since if a pp-introduces equality constraints one can always simplify the resulting CSP instance by identifying one of the variables with the other throughout the instance. In the unbounded domain CSP problem, however, equality is a relatively powerful relation which sometimes (depending on the allowed maps) is enough to give hardness by itself. Hence, we do not allow equality (unless it is already included in the constraint language) and on the functional side this means that we go from functions to {\em multifunctions}, i.e. functions from $D^r$ to $2^D$.

Second, our definability notion need to be compatible with the fixed set of allowed maps and allow us to define unbounded-domain relations from the base  language $\Gamma$. This significantly increases the expressive strength of the definitions and we obtain a generalization of {\em functionally guarded primitive positive definitions} (fgpp-definitions) from Carbonnel~\cite{carbonnel2022Redundancy}.

Third, the underlying map family greatly affects the algebraic theory. If the set of is completely arbitrary then we (as expected) cannot make any significant algebraic simplifications, but, as we will see, if the set of maps e.g.\ is monotone then the algebraic theory is greatly simplified by using the map family together with homomorphic images. The effect is that we can concentrate exclusively on operations satisfying {\em order patterns}, and we call such functions {\em order polymorphisms}.
%\todo{If we have anything to say in the one hot case, then add it here.}

\subsection{Problem Statement}
\label{sec:statement}

Our setup is different from conventional CSPs as we do not have one global domain $D$ that is shared by all relations and instances. Instead, each instance has an input-defined domain size $n$, and we use unary maps to translate from this domain to the domains of the relations. Furthermore, in order to properly facilitate the creation of new relations with arbitrary domain sizes, we do not want to be restricted by a global domain shared by all relations. This motivates the following definition, where we redefine the notions of \emph{relation}, \emph{constraint language}, and \emph{formula} to allow different domains for each instance and relation, together with the generalization that an atom now may contain literals of the form $f(x)$ where $f$ is a map and $x$ a variable.

\begin{definition}\label{def:functionally-guarded} We redefine the following notions from conventional CSPs.
    \begin{itemize}
        \item A \emph{relation} of arity $r$ is a tuple $(D,R)$ consisting of a finite domain $D$ and a set $R \subseteq D^r$. Contrary to the conventional definition, the domain is explicitly specified for each relation since constraints on the relation needs to be aware of the domain in question. If the domain is clear from the context then we do not make a sharp distinction between $R$ and $(D, R)$.
        \item A \emph{constraint language} $\Gamma$ is a finite set of relations. Contrary to the conventional definition, there is no fixed domain and $\Gamma$ is assumed to be finite by default.
        \item A \emph{formula} over some constraint language $\Gamma$ is a tuple $(D,\phi)$ consisting of a finite domain $D$ and a logic formula $\phi(x_1, \ldots, x_m)$ where each variable has domain $D$ and where each atom is of the form $R(f_1(y_1), \ldots, f_r(y_r))$ with $(E,R) \in \Gamma$ some relation and $f_1, \ldots, f_r \colon D \to E$ maps. Contrary to the conventional definition, we include a domain that is not dependent on $\Gamma$ and we add maps to translate from the domain of the formula to the domain of the relations. We refer to these maps as \emph{guarding functions}.
    \end{itemize}
\end{definition}

As stated in the introduction, it is useful to add restrictions to these guarding functions. Our main tool to define such restrictions are map families.

\begin{definition}
    A \emph{map family} $\MM$ is a collection of sets of maps which, for each pair $(d,e) \in \N^2$, contains exactly one set of maps $\MM_d^e \subseteq \numdom{d}^{\numdom{e}}$. We define standard set operations on map families by applying them set-wise.
\end{definition}

We are now able to state our problem definition. Let $\Gamma$ be a constraint language and $\MM$ a map family.
\problemDef{$\udcsp(\Gamma,\MM)$}
{A tuple $(n,V,C)$ where \begin{itemize}
    \item $n \in \N$ is a domain size, given in unary.
    \item $V$ is a set of variables on the domain $\numdom{n}$.
    \item $C$ is a set of functionally guarded constraints, each consisting of a relation $R \in \Gamma$ with arity $r$ and domain size $d$, a tuple of variables $v_1, \ldots, v_r \in V$ and a tuple of maps $m_1, \ldots, m_r$ satisfying $m_i \in \MM_d^n$.
  \end{itemize}}
{Does there exist a function $f : V \to \numdom{n}$ mapping each variable to an element from its domain such that for each constraint $(R, (v_1, \ldots, v_r), (m_1, \ldots, m_r)) \in C$ we have that $(m_1(f(v_1)),\ldots, m_r(f(v_r))) \in R$.}

%We sometimes consider relations or instances on an arbitrary (ordered) finite domain $D$, which we then identify with the domain $\numdom{d}$ where $d = |D|$.

%The size of an instance is primarily determined by $n$\todo{PJ: Why?}. 

The choice of unary representation for the domain size is made
primarily to avoid complications with algorithm efficiency,
since these issues are not the focus of the paper.
However, for arbitrary unary maps we anyway need $\Theta(n)$
space to define each map. On the other hand,
for \mcsp and \ohcsp, an instance with $m$ constraints
will have only $O(m)$ equivalence classes of domain elements,
where two domain elements are equivalent if they are indistinguishable
in every unary map. Hence the unary representation is not a
significant restriction. 

We typically parameterize the problem with respect to the parameter $|V|$.
Note that this parameterization is not very useful for CSP$(\Gamma)$:
when $\Gamma$ has a finite domain $D$, then every instance can be solved in time
$|V|^{|D|} \cdot poly(||I||)$ and CSP$(\Gamma)$ is {\em always} in FPT.
Certainly, this is not true for the $\udcsp(\Gamma,\MM)$ problem.

We note that the arc consistency procedure works in the $\udcsp$ setting and it still runs in polynomial time.
Let $\Gamma$ be a finite constraint language with relations of maximum arity $r$.
Let $(n,V,C)$ be an instance of $\udcsp(\Gamma)$ where $|C| = m$ and $V = \{x_1,\dots,x_k\}$. 
  Since the sum of all domain sizes is initially $nk$ and every iteration reduces this sum by one, we need at most $nk$ iterations. Each iteration checks $m\cdot k \cdot n$ tuples of relation/variable/value and each check takes $O(n^{r})$ time. The runtime of this step is polynomial since $r$ is a fixed constant. The whole process takes
  polynomial time since $n$ is given in unary.

We now introduce the main map families for which we will study the (parameterized) complexity of the resulting $\udcsp$ problems. We begin with the family containing all maps.

\begin{definition}
    Let $\all$ be the map family where $\all_d^n$ consists of all maps from $\numdom{n}$ to $\numdom{d}$.
    For any constraint language $\Gamma$, we now define $\ucsp(\Gamma)$ as $\udcsp(\Gamma,\all)$.
\end{definition}

Dually, we define one of the simplest possible map families. 
%Here, we do not define the corresponding CSP since the definition is mainly of algebraic interest. \todo{JdV: if all relations have the same domain, then the corresponding CSP here is precisely the conventional CSP}

\begin{definition}
    Let $\id$ be the map family where $\id^d_d$ contains only the identity function $i(x) = x$ for each $x \in \numdom{d}$ and $\id^d_e = \emptyset$ for all distinct $e,d \in \mathbb{N}$.
\end{definition}

Note that if each relation in $\Gamma$ has the same domain then $\udcsp(\Gamma, \id)$ is simply an alternative formulation of the standard CSP problem since the identity maps add no expressive power at all. We continue by defining the two map families of main interest in this paper.

\begin{definition}
    A map $m : \numdom{n} \to \numdom{2}$ is called \emph{one-hot} if $m(x) = 1$ for exactly one $x \in \numdom{n}$. Let $\oh_n$ be the set of all one-hot maps from $\numdom{n} \to \numdom{2}$ and let $\oh$ be the map family where $\oh_d^n$ equals $\oh_n$ if $d=2$ and $\emptyset$ for all other $d$.
    For any constraint language $\Gamma$, we now define $\ohcsp(\Gamma)$ as $\udcsp(\Gamma,\oh)$.
\end{definition}

\begin{definition}
  We make the following definitions.
  \begin{enumerate}
  \item
    A map $m : \numdom{n} \to \numdom{d}$ is called \emph{(anti-)monotone} if $x \leq y$ implies that $m(x) \leq m(y)$ ($m(y) \leq m(x)$) according to the natural ordering on $\numdom{n}$ and $\numdom{d}$.
  \item
    Let $\mo$ ($\mo'$) be the map family where $\mo_d^n$ (${\mo'}_d^n$) consists of all (anti-)monotone maps from $\numdom{n}$ to $\numdom{d}$\footnote{The anti-monotone map family is not so interesting in its own right but becomes useful for constraint languages which, in a certain technical sense, can simulate anti-monotone maps and obtain more powerful definitions.}.
    \item
      For any constraint language $\Gamma$, we now define $\mcsp(\Gamma)$ as $\udcsp(\Gamma,\mo)$.
    \end{enumerate}
\end{definition}

%\begin{definition} \label{def:functionally-guarded}
%  A \emph{relation} $R$ of arity $r$ and domain size $d \in \N$ is a set $R \subseteq \numdom{d}^r$. A \emph{constraint language} $\Gamma$ is a finite collection of such relations.
  % Constraints may have different domains.
%  A \emph{formula} of domain $\numdom{n}$ over $\Gamma$ is a logic formula $\phi(x_1, \ldots, x_m)$ where each variable has domain $\numdom{n}$ and where each atom is of the form $R(f_1(y_1), \ldots, f_r(y_r))$ with $R \in \Gamma$ a relation of arity $r$, domain $\numdom{d}$ and $f_1,\ldots,f_r : \numdom{n} \to \numdom{d}$.
%\end{definition}\todo{PJ: Already defined in Section~2.2.}

When $S$ is a (mathematical) statement, then the {\em Iverson bracket} is defined by
$[S]=0$ if $S$ is false and $[S]=1$ if $S$ is true. 
We often use Iverson brackets for describing unary maps when the base language is Boolean. 
Let $x$ be a variable with domain $\numdom{d}$ and $a$ an element in $\numdom{d}$. 
Then,
$[x=a]$ is a one-hot map on $\numdom{d}$ while $[x \leq a]$ is a
monotone map from $\numdom{d}$ to $\numdom{2}$.

\subsection{The Relational Side: Functionally Guarded Primitive Positive Definitions}

Let $\Gamma$ be a constraint language, $\MM$ a map family, and $R$ an $r$-ary relation over \numdom{d}. We say that $R$ has a {\em functionally guarded equality-free quantifier-free primitive positive} definition over $\Gamma$ if
\[
    R(x_1, \ldots, x_r) \equiv \varphi(x_{i_1}, \ldots, x_{i_n})
\]
where $i_1, \ldots, i_n \in [r]$ and $\varphi(\cdot)$ is an equality-free, quantifier-free primitive positive formula over $\Gamma$ with domain $\numdom{d}$. Recall from Definition \ref{def:functionally-guarded} that each usage of a variable in an atom $R_0 \in \Gamma$ of domain $\numdom{n}$ is now associated with a map from $\numdom{n}$ to $\numdom{d}$. In order to be consistent with Carbonnel~\cite{carbonnel2022Redundancy} and to avoid an unnecessarily long abbreviation we simply say that $R$ is {\em fgpp-definable} rather than the technically more correct but arguably too cumbersome {\em fgefqfpp-definition}.

For a map family $\MM$ we then say that $R$ is {\em $\MM$-fgpp-definable} if each involved map is included in $\MM$. We let $\fgpp{\Gamma}{\MM}$ be the set of all $\MM$-fgpp-definable relations over $\Gamma$.

We sometimes relax the notion of fgpp-definability and allow existential quantification and thus consider definitions of the form
\[
    R(x_1, \ldots, x_r) \equiv \exists y_1, \ldots, y_s \colon \varphi(x_{i_1}, \ldots, x_{i_n}, y_{j_1}, \ldots, y_{j_m})
\]
where $i_1, \ldots, i_n \in [r]$, $j_1, \ldots, j_m \in [s]$, and $\varphi(\cdot)$ is an equality-free, quantifier-free primitive positive formula over $\Gamma$ with domain $D$.
Similarly to the above case, we name such relations \emph{$\MM$-efpp-definable} and let $\efpp{\Gamma}{\MM}$ be the set of all $\MM$-efpp-definable relations. 

\begin{remark}
  Recall from Section~\ref{sec:prelim_algebra} that we introduced the classical closure operators $\cclone{\Gamma}_{\neq}$ (closure under equality-free pp-definitions) and $\cclone{\Gamma}_{\nexists, \neq}$ (closure under equality-free, quantifier-free pp-definitions). We now express these as $\cclone{\Gamma}_{\neq} = \efpp{\Gamma}{}$ and $\cclone{\Gamma}_{\nexists, \neq} = \fgpp{\Gamma}{}$ with the motivation that they can be seen as special cases of the latter two operators with only identity maps.
\end{remark}  

Our notion of fgpp-definability generalizes and (arguably) simplifies Carbonnel~\cite{carbonnel2022Redundancy} in a few different directions. Concretely, we (1) do not allow equality, (2) consider (subsets of) maps from $D$ to $E$ rather than over a uniform domain, and (3) allow our definitions to use functional expressions of the form $g(x)$ rather than phrasing it via existential quantification $\exists y \colon g^\bullet(x,y)$.

We now look at the effect of fgpp-definitions on reductions between $\udcsp$ instances. We begin with reductions from richer constraint languages.

\begin{definition}
    Let $\MM$ and $\NN$ be map families. We define their composition $\NN \circ \MM$ as the map family with $(\NN \circ \MM)_d^e = \bigcup_{i\in\N} \{n \circ m \mid n \in \NN_d^i, m \in \MM_i^e\}$, where $(n \circ m)(x) = n(m(x))$.
    We say that an $\MM$ is \emph{closed under composition with $\NN$} if $(\NN \circ \MM) \subset \MM$.
\end{definition}

\begin{proposition}\label{prop:language_fpt_reduction}
    Let $\MM$ and $\NN$ be map families and $\Gamma$ a constraint language. Consider a relation $R \in \fgpp{\Gamma}{\NN}$. Then $\udcsp(\Gamma \cup \{R\}, \MM) \leq_{\rm FPT} \udcsp(\Gamma, \MM \cup (\NN \circ \MM))$.
\end{proposition}
\begin{proof}
    Let $r$ and $D_R$ be the arity and domain of $R$. Let an instance of $\udcsp(\Gamma \cup \{R\}, \MM)$ be given. We construct an instance of $\udcsp(\Gamma, \MM \cup (\NN \circ \MM))$ as follows. We first copy each variable and constraint not involving $R$. Now consider a constraint $(R, (v_1,\ldots,v_r),(m_1,\ldots,m_r))$. Since $R$ is $\MM$-fgpp-definable in $\Gamma$, it corresponds to conjunction of atoms of the form $S(f_1(y_1), \ldots, f_s(y_s))$ for some $S \in \Gamma$ of (say) arity $s$ and domain $D_S$, with function $f_1,\ldots,f_s \in \MM^{D_R}_{D_S}$. Each variable $y_i$ corresponds to some variable $v_{i'}$.
    For each atom, we now add the constraint $(S,(v_{1'},\ldots,v_{s'}),(f_1 \circ m_{1'},\ldots,f_s \circ m_{s'}))$. Since $S \in \Gamma$ and $f_i \circ m_{i'} \in \NN \circ \MM$, this constraint is valid.
    An assignment satisfies the new instance if and only if it satisfies the constraints from these atoms and the constraints that did not use $R$. This happens if and only if it satisfies the constraints involving $R$ and those that did not use $R$, which in turn happens if and only if it satisfies the old instance. This completes the reduction. Since the number of variables is not affected at all the reduction is FPT.
\end{proof}

We continue with reductions from richer map families and show what happens when the right-hand map family $\NN$ is definable over the base language $\Gamma$.

\begin{proposition}\label{prop:maps_fpt_reduction}
    Let $\MM$ and $\NN$ be map families and $\Gamma$ a constraint language. Suppose that the following conditions hold.
    \begin{itemize}
        \item The only nonempty map sets in $\NN$ are of the form $\N^n_n$ and satisfy $\left|\N^n_n\right| \leq 1$. That is, each map has the same domain and co-domain, and each domain occurs at most once.
        \item There is a polynomial $P$ such that for each function $g$ from $\NN$ we have $g^\bullet \in \fgpp{\Gamma}{\MM}$ using an $\MM$-fgpp-definition that can be constructed in $P(|g^{\bullet}|)$ time.
    \end{itemize}
    Then $\udcsp(\Gamma, \MM \cup (\MM \circ \NN)) \leq_{\rm FPT} \udcsp(\Gamma, \MM)$. This is an FPT-reduction that doubles the number of variables.
\end{proposition}
\begin{proof}
    Let an instance $(n,V,C)$ be given. Due to the first requirement, $\NN$ contains at most one function $g$ whose domain is $\numdom{n}$. If there is none, we are done. Otherwise the second requirement shows we can construct a formula that $\MM$-efpp-defines $g^\bullet$ in polynomial time. We now extend each variable $x \in V$ with a fresh variable $x_g$ corresponding to the image of $g$. We then implement the formula on $x$ and $x_g$ by replacing each atom with a constraint. Overall, the number of variables doubles, which is independent of the domain size, and thus FPT.
    Now consider some constraint containing a guarding function $m \circ g$ from $\MM \circ \NN$. We  can simply replace $m(g(x))$ with $m(x_g)$ to obtain a guarding function from $\MM$.
\end{proof}

The previous propositions considered FPT-reductions. For polynomial-time reductions, we instead obtain the following results.
\begin{proposition}\label{prop:language_poly_reduction}
    Let $\MM$ and $\NN$ be map families and $\Gamma$ a constraint language. Consider a relation $R \in \efpp{\Gamma}{\NN}$. Then $\udcsp(\Gamma \cup \{R\}, \MM) \leq_{\rm poly} \udcsp(\Gamma, \MM \cup (\NN \circ \MM))$. 
\end{proposition}
\begin{proof}
    Let $\phi$ be the formula that $\NN$-efpp-defines $R$. We modify $\phi$ into a (higher arity) fgpp formula $\phi'$ by replacing each quantifier with an input variable, and let $R'$ be the corresponding (higher arity) relation.
    
    Now let $(n,V,C)$ be an instance of $\udcsp(\Gamma \cup \{R\}, \MM)$. We replace each occurrence of $R$ with $R'$ by introducing a fresh variable for each existential quantifier. The number of new variables increases linearly in the number of constraints, hence this is a polynomial-time reduction to $\udcsp(\Gamma \cup \{R'\}, \MM)$. Since $R' \in \fgpp{\Gamma}{\NN}$, we can apply Proposition \ref{prop:language_fpt_reduction} and obtain a reduction
    \[
        \udcsp(\Gamma \cup \{R\}, \MM) \leq_{\rm poly} \udcsp(\Gamma \cup \{R'\}, \MM) \leq_{\rm FPT} \udcsp(\Gamma, \MM \cup (\NN \circ \MM))
    \]
    which completes the proof.
\end{proof}

\begin{proposition}\label{prop:maps_poly_reduction}
Let $\MM$ and $\NN$ be map families and $\Gamma$ a constraint language. Suppose that the following conditions hold.
    \begin{itemize}
        \item The only nonempty map sets in $\NN$ are of the form $\NN^n_n$. That is, each function has the same domain and co-domain.
         \item There is a polynomial $P$ such that for each function $g$ from $\NN$ we have $g^\bullet \in \efpp{\Gamma}{\MM}$ using an $\MM$-efpp-definition that can be constructed in $P(|g^{\bullet}|)$ time.
    \end{itemize}
    Then $\udcsp(\Gamma, \MM \cup (\MM \circ \NN)) \leq_{\rm poly} \udcsp(\Gamma, \MM)$.
\end{proposition}
\begin{proof}
    Let an instance $(n,V,C)$ be given. For each function $g \in \NN^n_n$, the second requirement shows that we can construct a formula that $\MM$-efpp-defines $g^\bullet$ in polynomial time. We now extend each variable $x \in V$ with a fresh variable $x_g$ corresponding to the image of $g$ and more fresh variables corresponding to the quantifiers from the formula. We then implement the formula on $x$ and these new variables by replacing each atom with a constraint. Overall, this reduction takes polynomial time. Since the number of variables increases greatly, it is no longer FPT.
    Now consider some constraint containing a guarding function $m \circ g$ from $\MM \circ \NN$. We  can simply replace $m(g(x))$ with $m(x_g)$ to obtain a guarding function from $\MM$.
\end{proof}

In most applications, we only work with map families that are closed under composition with themselves. This is the case in monotone maps, anti-monotone maps, and arbitrary maps. We conclude this section by looking at reducibility in this setting.
\begin{corollary}[of Proposition \ref{prop:language_fpt_reduction}] \label{corr:language_fpt_reduction}
    Let $\MM$ be a map family closed under composition with itself. Let $\Gamma$ be a constraint language. Consider a relation $R \in \fgpp{\Gamma}{\MM}$. Then $\udcsp(\Gamma \cup \{R\}, \MM) \leq_{\rm FPT} \udcsp(\Gamma, \MM)$.
\end{corollary}
\begin{corollary}[of Proposition \ref{prop:language_poly_reduction}]\label{corr:language_poly_reduction}
    Let $\MM$ be a map family closed under composition with itself. Let $\Gamma$ be a constraint language. Consider a relation $R \in \efpp{\Gamma}{\MM}$. Then $\udcsp(\Gamma \cup \{R\}, \MM) \leq_{\rm poly} \udcsp(\Gamma, \MM)$.
\end{corollary}

\begin{proposition} \label{prop:can_compose}
  Let $\MM$ be a map family closed under composition with itself. Let $\Gamma$ be a constraint language. If $R \in \fgpp{\Gamma}{\MM}$ and $S \in \fgpp{\{R\}}{\MM}$ then $S \in \fgpp{\Gamma}{\MM}$
\end{proposition}
\begin{proof} 
    The maps used to fgpp-define $S$ over $\{R\}$ and those used to fgpp-define $R$ over $\Gamma$ compose into maps that fgpp-define $S$ over $\Gamma$.
\end{proof}

\subsection{The Functional Side: Partial Multivalued Operations}
\label{sec:img_and_retraction}

%\todo[inline]{MW: The following is not fully checked for definitions,
%  typography, etc. and may even contain a TODO}

Let $D_1$ and $D_2$ be domains. Recall from Section~\ref{sec:algebra} that a partial multivalued operation from $D_1$ to $D_2$
is a function $f \colon D_1 \to 2^{D_2}$, where $f(d)=\emptyset$ means that $f(d)$ is undefined.
With slight abuse of notation, we also define $f(S)$ for sets $S \subseteq D_1^r$ as $f(S) := \bigcup_{s \in S} f(s)$. Furthermore, we define the inverse operation $f^{-1} : D_2 \to 2^{D_1}$ as $f^{-1}(x) := \{y \mid x \in m(y)\}$.
We often write the shorthand \emph{multifunction} for a partial multivalued operation or polymorphism when it is clear from context which one is meant.
Recall that for a constraint language $\Gamma$ we write $\mpol(\Gamma)$ be the set of all partial multivalued polymorphisms that preserve $\Gamma$.

In order to make it possible to relate multifunctions between different domains we generalize the well-known concept of \emph{homomorphic images} from universal algebra to the multifunction setting.

\begin{definition}
    Let $D_1$ and $D_2$ be domains. Let $m \colon D_1 \to 2^{D_2}$ be a partial multivalued operation between them. For an $r$-ary partial multivalued polymorphism $f$ on $D_1$ we define $f_m$ over $D_2$ as 
    $f_m(m(x_1), \ldots, m(x_r)) = m(f(x_1, \ldots, x_r))$ for every $(x_1, \ldots, x_r) \in D_1^r$.
    Since the result is a multivalued polymorphism, it is no problem to consider a set of non-surjective maps rather than a single surjective map. For a set of maps $M \subseteq D_1^{D_2}$, we define $f_M$ over $D_2$ as
    \[f_M(x) = \bigcup_{m \in M} f_m(x)\] and if $M = \emptyset$ we let 
    \[f_{\emptyset}(d) = \emptyset\] for each $d \in D^r_2$.
    Additionally, for a set of partial multivalued polymorphisms $F$ we define the \emph{concrete homomorphic image} $\img{F}{M}$ as $\{f_M \mid f \in F\}$.
\end{definition}

We now restrict this definition to two settings.
Firstly, let $M$ be a set of (not multivalued) maps from $D_1$ to $D_2$. We then define the \emph{homomorphic image} $\img{F}{M}$ by interpreting each map as an elementary multivalued operation. 
Secondly, let $M$ again be a set of (not multivalued) maps from $D_1$ to $D_2$. We define the \emph{inverse homomorphic image} $\img{F}{M}^{-1}$ as $\img{F}{\{m^{-1} \mid m \in M\}}$, where we again interpret $m$ as a multivalued operation to make $m^{-1}$ well-defined. We remark that inverse homomorphic images are a generalization of {\em reflections}~\cite{barto2018} to the partial multifunction setting.

With these operators we can, given maps from $D_1 \to D_2$, relate multifunctions over $D_1$ with multifunctions from $D_2$. We now use them to extend polymorphisms to arbitrary domains.

\begin{definition} Let $f$ be a (partial) multivalued polymorphism on some domain $D$ and let $\MM$ be a map family. We say that $f$ \emph{$\MM$-preserves} some relation $R$ on some domain $E$ if the concrete homomorphic image $f_{\MM^D_E}$ is a multivalued polymorphism of $R$. For a constraint language $\Gamma$, we define $\mpol(\Gamma,\MM)$ as the set of all partial multivalued polymorphisms that $\MM$-preserve all relations in $\Gamma$.
\end{definition}

It is important to note that $f$ is interpreted differently (via the homomorphic image $f_{\MM^D_E}$) depending on the domain of $R$, and can thus be related to operations definable from strong Maltsev conditions in the algebraic approach for CSPs~\cite{barto2018}.
In particular, we remark that $f$ is different from the interpretation over its own domain unless $\MM^D_D$ consists of all identity functions. This is why we usually refer to $f$ as a \emph{pattern polymorphism}, or just a {\em pattern}, to distinguish it from its interpretations over other domains. In the context of monotone maps $\mo$ we sometimes say that $f$ is an {\em order polymorphism}.

\begin{example}
    Consider a Boolean relation $R \subseteq \numdom{2}^r$. Then (recall Section~\ref{sec:prelim_algebra}) $\mpol(R)$ is the set of all Boolean partial multipolymorphisms of $R$. If we instead use the simple map family $\id$ of all identity functions then $\mpol(R, \id)$ is the slightly richer family of all multipolymorphisms that in addition to $\mpol(R)$ trivially contains all multifunctions defined on the ``wrong'' domain. To see this, pick e.g. a multifunction $f$ over $\numdom{3}$. Then $\id^3_2 = \emptyset$ and the interpretation $f_{\id^3_2} = f_{\emptyset}$ is the empty partial multifunction which trivially preserves $R$. This notion comes in handy when we relate relations and constraint languages defined over different domains. 
\end{example}

\begin{example}\label{example:min}
  Consider the case where $\MM = \mo$, the family of monotone maps. We define the \emph{minimum pattern} $\pmin$ as the partial multivalued polymorphism defined by $\pmin(0,1) = \pmin(1,0) = 0$. To see its effect on other domains, consider for example the domain $\numdom{5}$. The pattern $\pmin$ lifts to some polymorphism $f$ on this domain. This polymorphism satisfies, amongst others, $f(3,4) = 3$ (from $\pmin(0,1) = 0$ using the monotone map $01 \mapsto 34$), $f(2,0) = 0$ (from $\pmin(1,0) = 0$ using the map $01 \mapsto 02$), or $f(1,1) = 1$ (from $\pmin(0,1) = 0$ using the map $01 \mapsto 11$). Clearly, $f$ is exactly the conventional $\min$ operation on $\numdom{5}$ and is totally defined.
  
  We stress that $\pmin$ itself is partially defined and different from the interpretation of $\pmin$ on $\numdom{2}$. For example, $\min$ satisfies $\min(0,0) = 0$ while $\pmin(0,0)$ is undefined.
\end{example}

\subsection{A Galois Connection}
We now present the link between $\MM$-fgpp-definitions and $\NN$-preserving polymorphisms (where $\MM$ and $\NN$ are map families but not necessarily the same). Note that since we do not make any assumptions on $\MM$ and $\NN$ (e.g., that it is closed under composition, or contains any kind of non-trivial map) we should not expect the algebraic landscape to be simpler than partial multifunctions. We begin with some lemmas.

\begin{lemma}\label{lem:fgpp_to_mpol}
    Let $\MM$ and $\NN$ be map families and $\Gamma$ a constraint language. Let $R$ be a relation that is $\MM$-fgpp-definable over $\Gamma$. Let $f$ be a multivalued polymorphism that $(\MM \circ \NN)$-preserves all relations in $\Gamma$. Then $f$ $\NN$-preserves $R$.
\end{lemma}
\begin{proof}
    Let $D_R$ and $D_f$ be the domains and $r$ and $k$ the arities of $R$ and $f$, respectively. Assume to the contrary that $f$ does not $\NN$-preserve $R$. Then there exist tuples $t_1, \ldots, t_k \in R$, tuple $t_0 \notin R$ and maps $m_1,\ldots,m_k \in \NN^{D_f}_{D_R}$ such that, in each coordinate $i \in [r]$, $t_0[i] \in f_{m_i}(t_1[i], \ldots, t_k[i])$. Consequently, there exist $x_{i0}, \ldots, x_{ik} \in D_f$ satisfying $m_i(x_{ij}) = t_j[i]$ for each $j \in [k]$ and $x_{i0} \in f(x_{i1}, \ldots, x_{ik})$.

    Since $t_0 \notin R$, it fails some atom of the $\MM$-fgpp-formula defining $R$. This atom is of the form
    \[
        S(h_1(x_{1'}), \ldots, h_s(x_{s'}))
    \]
    for some $s$-ary relation $S \in \Gamma$ with domain $D_S$, functions $h_1, \ldots, h_s \in \MM^{D_R}_{D_S}$ and indices $1', \ldots, s' \in [r]$. We now define $u_j$ as the image of $t_j$ under these maps for $j \in [k]$. That is,
    \[
        u_j := (h_1(t_j[1']),\ldots,h_s(t_j[s']))
    \]
    Since $t_0$ does not satisfy $S$ while $t_1, \ldots, t_k$ all satisfy $R$, and hence also $S$, we find $u_j \in S$ if and only if $j \neq 0$.

    Now consider the tuple $(u_1, \ldots, u_k) \in S^k$. Since $\Gamma$ is $(\MM \circ \NN)$-preserved by $f$, we find that $f_M(u_1, \ldots, u_k) \in S$ where $M := (\MM \circ \NN)^{D_f}_{D_S}$. In the $i$-th coordinate, we find
    \[
        f_M(u_1[i], \ldots, u_k[i]) = f_M(h_i(t_1[i']),\ldots,h_i(t_k[i'])) = f_M(h_i(m_{i'}(x_{i'1})), \ldots, h_i(m_{i'}(x_{i'k})))
    \]
    Since $h_i \circ m_{i'} \in \MM \circ \NN$, we find that $f_M(u_1[i], \ldots, u_k[i])$ at least contains $h_i(m_{i'}(f(x_{i'1}, \ldots, x_{i'k}))) = h_i(m_{i'}(x_{i'0})) = h_i(t_0[i']) = u_0[i]$. Hence, $f_M(u_1, \ldots, u_k)$ contains $u_0$, which means that $u_0 \in S$; a contradiction.
\end{proof}

\begin{lemma}\label{lem:mpol_to_fgpp}
    Let $\MM$ be a map family and $\Gamma$ a constraint language. Let $R$ be a relation that is preserved by all multivalued polymorphisms that $\MM$-preserve $\Gamma$. Then $R$ is $\MM$-fgpp-definable over $\Gamma$.
\end{lemma}
%\todo[inline]{JdV: Does this generalize to the case where $R$ is $\NN$-preserved by all multivalued polymorphisms that $(\MM \circ \NN)$-preserve $\Gamma$ for some map family $\NN$? That would give a proper reverse statement to Lemma \ref{lem:fgpp_to_mpol}.}
%\todo{JdV: Not in general, it does not hold when e.g. $\NN$ is empty. But some variant may still work}
\begin{proof}
    Let $D_R$ and $r$ be the domain and arity of $R$. Let $r_1,\ldots, r_n$ be some enumeration of the tuples in $R$. For each $i \in [k]$, we let $c_i := (r_1[i],\ldots,r_n[i])$ correspond to a column of the tuples in $R$.

    We now construct an $\MM$-fgpp formula over $\Gamma$ using $r$ variables with domain $D_R$. For each constraint $S \in \Gamma$ of (say) arity $s$ and domain $D_S$, maps $h_1,\ldots, h_s \in \MM^{D_R}_{D_S}$, and indices $1',\ldots,s' \in [r]$, we add the atom
    \[
        S(h_1(x_{1'}),\ldots, h_s(x_{s'}))
    \]
    to the fgpp-definition if and only if $R$ does not fail this atom. That is, when $(h_1(r_j[1']), \ldots,  h_s(r_j[s'])) \in S$ for each $j \in [n]$. Let $R'$ be the relation that is $\MM$-fgpp-defined using this formula. Note that, by construction, $R \subseteq R'$. We now prove that $R = R'$. Suppose to the contrary that there exists some tuple $t \in R' \setminus R$.

    We now define a multivalued $n$-ary polymorphism $f$ with domain $D_R$. For each $i \in [r]$, we add $f(c_i[1],\ldots,c_i[n]) = t[i]$. Note that $R$ does not preserve $f$.

    Suppose that $f$ does not $\MM$-preserve some relation $S \in \Gamma$ of (say) arity $s$ and domain $D_S$. Then there exist tuples $u_1, \ldots, u_n \in S$, tuple $u_0 \notin S$ and maps $h_1,\ldots,h_s \in \MM^{D_R}_{D_S}$ such that, in each coordinate $i \in [s]$, $u_0[i] \in f_{h_i}(u_1[i], \ldots, u_n[i])$. Consequently, there exist $x_{i0}, \ldots, x_{in} \in D_f$ satisfying $h_i(x_{ij}) = u_j[i]$ for each $j \in [r]$ and $x_{i0} \in f(x_{i1}, \ldots, x_{in})$. By our restrictive definition of $f$, this means that there is some index $i'$ such that $x_{ij} = c_{i'}[j]$ and $x_{i0} = t[i']$.

    Now consider some tuple $u_i$ with $i>0$. We find that
    \[
        u_i = (u_i[1],\ldots,u_i[s]) = (h_1(x_{1i}),\ldots,h_s(x_{si})) = (h_1(c_{1'}[i]),\ldots,h_s(c_{s'}[i])) = (h_1(r_i[1']),\ldots,h_s(r_i[s']))
    \]
    Compare this with the atom $S(h_1(x_{1'}), \ldots, h_s(x_{s'}))$. For each $i \in [n]$, $r_i$ maps to $u_i$, which is contained in $S$. Hence, this atom is part of the $\MM$-fgpp-definition used to define $R'$. For the tuple $u_0$ we instead find
    \[
        u_0 = (u_0[1],\ldots,u_0[s]) = (h_1(x_{10}),\ldots,h_s(x_{s0})) = (h_1(t[1']),\ldots,h_s(t[s']))
    \]
    Hence, $t$ maps to $u_0$ in this atom. But since $u_0 \notin S$, $t$ does not satisfy this atom which contradicts $t \in R'$. It follows that $f$ $\MM$-preserves all relations from $\Gamma$. Furthermore, by our assumptions on $R$, this means that $f$ now also preserves $R$, which is again a contradiction. We conclude that $R' = R$, and that $R$ is $\MM$-fgpp-definable.
\end{proof}

We now combine the lemmas into a theorem relating $\MM$-fgpp definitions with $\MM$-preserving polymorphisms.

\begin{theorem} \label{thm:galois-wrapup}
    Let $\MM$ be a map family, $\Gamma$ a constraint language, and $R$ a relation. Then,
    \[
        R \in \fgpp{\Gamma}{\MM}
    \]
    if and only if
    \[
        \mpol(\Gamma,\MM) \subseteq \mpol(R, \id)
    \]
\end{theorem}
\begin{proof}
    The first direction follows from Lemma \ref{lem:fgpp_to_mpol} by letting $\NN = \id$. The second direction follows from Lemma \ref{lem:mpol_to_fgpp}.
\end{proof}

We conclude this section with a proposition regarding the absence of $\MM$-preserving polymorphisms, which will be very useful in hardness proofs.

\begin{proposition}\label{prop:no_preserve_f}
    Let $\MM$ be a map family and $\Gamma$ a constraint language. Let $f$ be a multivalued polymorphism that does not $\MM$-preserve $\Gamma$. Let $D_f$ and $k$ be the domain and arity of $f$, and $((x_{11},\ldots,x_{1k}),y_1),\ldots,((x_{ri},\ldots,x_{rk},y_r))$ an enumeration of all input/output pairs of $f$. That is, $y_i \in f(x_{i1},\ldots,x_{ik})$ for all $i$. Then, there exists a relation $R \in \fgpp{\Gamma}{\MM}$ with domain $D_f$ and arity $r$ satisfying $(x_{1j},\ldots,x_{rj}) \in R$ for each $j$ while $(y_1,\ldots,y_r) \notin R$.
\end{proposition}
\begin{proof}
    Since $f$ does not $\MM$-preserve $\Gamma$, it does not preserve some relation $S \in \Gamma$ of (say) arity $s$ and domain $D_S$. Then there exist tuples $u_1, \ldots, u_k \in S$, tuple $u_0 \notin S$ and maps $h_1,\ldots,h_s \in \MM^{D_f}_{D_S}$ such that, in each coordinate $i \in [s]$, $u_0[i] \in f_{h_i}(u_1[i], \ldots, u_k[i])$. Consequently, there exists some index $i'$ such that the input/output pair $((x_{i'1}, \ldots, x_{i'k}),y_{i'})$ of $f$ satisfies $h_i(x_{i'j}) = u_j[i]$ for each $j \in [k]$ and $h_i(y_{i'}) = u_0[i]$.

    In particular, each coordinate $i$ of $S$ corresponds with the $i'$-th input/output pair. We now construct an $r$-ary relation $R$ with domain $D_f$ as follows:
    \[
        R(x_1,\ldots,x_r) \equiv S(h_1(x_{1'}),\ldots,h_s(x_{s'}))
    \]
    We note that a variable does not occur on the right hand side if its input/output pair is never reached. Now consider the tuple $(x_{1j},\ldots,x_{rj})$ for some $j$. The right hand side of the relation becomes
    \[
        S(h_1(x_{1'j}),\ldots,h_s(x_{s'j})) = S(u_j[1],\ldots,u_j[s])
    \]
    which satisfies the formula as $u_j \in S$. Similarly, we find that $(y_1,\ldots,y_r) \notin R$ as $u_0 \notin S$. This completes the proof.
\end{proof}

\subsection{Application: Monotone Maps}
As an example, we now look at some results specific to the case of monotone maps. Here, $\mo$-preserving polymorphisms can be interpreted as \emph{order operations}, where the output of a polymorphism is only dependent on how its argument are ordered. Recall for example the polymorphism pattern $\pmin$ from Example~\ref{example:min}. We can define patterns for the \emph{maximum} and \emph{median} in a similar way.

\begin{definition} \label{def:minmaxmedian}
    We define the polymorphism patterns $\pmin$, $\pmax$ and $\pmedian$ as follows.
    \begin{itemize}
        \item $\pmin$ is the binary partial polymorphism pattern on $\numdom{2}$ satisfying $\pmin(0,1) = \pmin(1,0) = 0$.
        \item $\pmax$ is the binary partial polymorphism pattern on $\numdom{2}$ satisfying $\pmax(0,1) = \pmax(1,0) = 1$.
        \item $\pmedian$ is the ternary partial polymorphism pattern on $\numdom{3}$ satisfying $\pmedian(0,1,2) = \pmedian(0,2,1) = \ldots = \pmedian(2,1,0) = 1$.
    \end{itemize}
\end{definition}

    We stress that this is an abstract definition in the sense that these polymorphisms can be interpreted over any ordered domain. In contrast, neither min nor max can abstractly be defined via so-called {\em strong Maltsev conditions}. The typical workaround is to instead view them as special cases of {\em semilattice operations}, i.e., idempotent, associative, and commutative, operations.

To make the above point more precise, consider $\pmin$ for example. If we view $0$ and $1$ as variables $x,y$ then the two identities can be viewed as a simple type of strong Maltsev condition $\pmin(x,y) = \pmin(y,x) = x$ equipped with an order $x < y$. For every $\numdom{d}$ one can then define a unique (partial multi)function (by using monotone maps from $\{x,y\}$ to $\numdom{d}$) that satisfies these identities and which is undefined otherwise, and this function is precisely the concrete homomorphic image of $\pmin$ under monotone maps. 

A key property in the monotone maps setting is that of \emph{domain reversal}. This occurs when we, in addition to monotone maps, also have access to anti-monotone maps. The relation between monotone and anti-monotone maps is given by the family of order-reversing maps; maps of the form $r : \numdom{n} \to \numdom{n}$ with $r(i) = n-1-i$.

We now investigate when we have access to order-reversing maps.
\begin{lemma} \label{lemma:can_define_reverse}
    Let $\mo$ be the monotone map family. Let $\Gamma$ be a constraint language over $\numdom{d}$ that is not $\mo$-preserved by neither $\pmin$ nor $\pmax$. Then for each order-reversing map $r$, we have $r^{\bullet} \in \fgpp{\Gamma}{\mo}$.
\end{lemma}
\begin{proof}
    Since $\Gamma$ is not $\mo$-preserved by $\pmin/\pmax$, Proposition \ref{prop:no_preserve_f} shows that we can $\mo$-define relations $R_{\min}$ and $R_{\max}$ on the domain $\numdom{2}$ satisfying $\{01, 10\} \subseteq R_{\min} \subseteq \numdom{2}^2 \setminus \{00\}$ and $\{01, 10\} \subseteq R_{\max} \subseteq \numdom{2}^2 \setminus \{11\}$. Together, these form the binary relation $R = \{01, 10\}$.

    Now consider the domain $\numdom{d}$. For $1 \leq i \leq d-1$ we let $f_i \colon \numdom{d} \to \numdom{d}$ be the map $f_i(X) = \{0\}$ for $X = \{0, \ldots, i-1\}$ and $f_i(Y) = \{1\}$ for $Y = \{i, \ldots, d-1\}$. Then, observe that a constraint $R(f_i(x), f_{d-i}(y))$ implies that $x \in \{0 \ldots, i-1\}$ if and only if $y \in \{d-i, \ldots, d-1\}$, i.e., $x \leq i-1$ iff $y \geq d-i$. Putting together these constraints for $1 \leq i \leq d - 1$ correctly defines $r^{\bullet}$.
\end{proof}

\begin{proposition} \label{prop:anti_monotone_maps}
    Let $\mo$ and $\mo'$ respectively be the monotone and anti-monotone map families. Let $\Gamma$ be a constraint language is not $\mo$-preserved by neither $\pmin$ nor $\pmax$. Then, $\udcsp(\Gamma,\mo \cup \mo') \leq_{\rm FPT} \mcsp(\Gamma)$.
\end{proposition}
\begin{proof}
    Because of Lemma \ref{lemma:can_define_reverse}, we can define all order-reversing maps. We now apply Proposition \ref{prop:maps_fpt_reduction} with $\MM = \mo$ and $\NN$ the family of all order-reversing maps, noting that $\mo' = \mo \circ \NN$.
\end{proof}

%%%%%%%%%%%%%%%%%%%%%%%%%%%%%%%%%%%%%%%%%%%%%%%%%%%%%%%%%%%%%%%%%%%%%%%%%%%%%%%%%%%%%%%%%%%%

\section{Unrestricted Maps}
\label{sec:always_hard}

We begin our study by considering the $\udcsp$ problem in the unrestricted map setting, i.e., the $\ucsp$ problem. We show in Section~\ref{sec:unrestr-hard} that $\ucsp(\Gamma)$ is essentially
always W[1]-hard, except for very simple settings. 
We use this result for discussing and analyzing various sources of hardness in Section~\ref{sec:unrestr-discussion}.

\subsection{Dichotomy Result}
\label{sec:unrestr-hard}

%Assume $\Gamma$ is a constraint language with domain $D$.
%Recall that $\all$ is the map family containing all possible maps.
%We show that $\ucsp(\Gamma)$ is essentially
%always W[1]-hard, except for very simple settings. 
We begin by showing that Boolean equality relation in itself is sufficient to yield hardness.
The proof uses a standard hardness reduction in the area~\cite{MarxR09,MarxR14},
which has become a standard method for proving W[1]-hardness, especially
for \textsc{MinCSP}-type problems~\cite{Dabrowski:etal:ipec2023,DabrowskiJOOW23soda,KimKPW23fa3,KimMPSW24weighted}.

\begin{lemma} \label{lm:equality-hard}
    The problem $\ucsp(\{\mathsf{Eq}\})$
  is W[1]-hard when parameterized by the number of variables. 
\end{lemma}
\begin{proof}
  We show a reduction from the W[1]-hard problem 
   {\sc Multicoloured Clique} that we introduced in Section~\ref{sec:param-comp}.
  Let $I=(G,k,\chi)$ be an arbitrary instance of this problem.
  For $i \in [k]$, let $V_i=\chi^{-1}(i)$
  and order $V_i=\{v_{i,1}, \ldots, v_{i,n_i}\}$ for each $i$ arbitrarily. 
  We create an instance of $\ucsp(\{\mathsf{Eq}\})$ as follows. Let the variable set be $X=\{x_{i,j} \mid i ,j \in [k], i \neq j\}$,
  where the domain of a variable $x_{i,j} \in X$ is $\{(u,v) \mid u \in V_i, v \in V_j, uv \in E(G)\}$. 
  For every pair of variables $x_{i,j}$ and $x_{i,j'}$ agreeing in the index $i$, impose the constraints
  \[
    \bigwedge_{t=1}^{n_1} \mathsf{Eq}(f_t(x_{i,j}), g_t(x_{i,j'}))
    \quad \text{where} \quad
    f_t((u,v))=g_t((u,w)) = [u=v_{i,t}],
  \]
  thereby encoding that $x_{i,j}[1] = x_{i,j'}[1]$. 
  Furthermore, for every pair $i,j \in [k]$, $i \neq j$, and every value $(u,v)$ in the domain of $x_{i,j}$,
  impose
  \[
    \mathsf{Eq}([x_{i,j} = (u,v)],[x_{j,i} = (v,u)]),
  \]
  signifying $x_{i,j}=(u,v)$ if and only if $x_{j,i}=(v,u)$.
  Proving the correctness of the reduction is straightforward.
\end{proof}

Call a language $\Gamma$ \emph{essentially unary} if every relation $R
\in \Gamma$ is efpp-definable over a set of unary relations.
%In the literature this property is usually defined via pp-definability rather than efpp-definability but the former is not suitable since we by Lemma~\ref{lm:equality-hard} would then have trivially hard cases.
%\todo{PJ: Equality-free definitions are needed here.} over a set of unary relations.
%\todo[inline]{I suppose that we assume that each relation is
%pp-definable over a set of unary relations.}  
Clearly, the property of being essentially unary is preserved by arbitrary unary maps, and every such language is
trivially tractable (see Theorem~\ref{thm:arbitrary-dichotomy}).
Using Lemma~\ref{lm:equality-hard}, we show that this is the only
tractable class.

\begin{lemma} \label{lm:f-essentially-unary}
    Consider the binary multivalued operation $f(0,1)=\{0,1\}$ on the Boolean domain $\numdom{2}$.
  Then a constraint language $\Gamma$ is $\all$-preserved by $f$ if and only if it is essentially unary. 
\end{lemma}
\begin{proof}
  Let $M = \all^{\numdom{2}}_D$ be the set of maps from $\numdom{2}$ to $D$. We prove each direction in turn.
  
  On the one hand, assume that $\Gamma$ is essentially unary and pick $R \in \Gamma$ with domain $D$ together with tuples $\alpha, \beta \in R$. 
    Let $\gamma \in f_{M}(\alpha,\beta)$. Then for every coordinate $i$ of $\gamma$, there exists an arbitrary map $m \colon \numdom{2} \to D$ with $m(0) = \alpha[i], m(1) = \beta[i]$, $\gamma[i] \in \{m(0),m(1)\}$, hence $\gamma[i] \in \{\alpha[i],\beta[i]\}$, which gives $\gamma \in R$ since $R$ is essentially unary.

    On the other hand, assume that $\Gamma$ is not essentially unary and assume that it is $\all$-preserved by $f$. Then there is is some $r$-ary $R \in \Gamma$ and indices $i_1, \ldots, i_k \in [r]$, $i_1 < i_2 < \ldots < i_k$, and a tuple $t \notin R$ such that we for $I_1 = \{t[i_1] \mid t \in R\}, \ldots, I_k = \{t[i_k] \mid t \in R\}$ have $(t[i_1], \ldots, t[i_k]) = (x_1, \ldots, x_k) \in I_1 \times \ldots \times I_k$. Next, we observe that if we have a tuple $u,v \in R$ where e.g. $u[i_1] = x_1$ and $v[i_2] = x_2$, then we can construct a tuple $w \in f_{M}(u,v) \subseteq R$ where $w[i_1] = x_1$ and $w[i_2] = x_2$. In fact, we can produce all such combinations, and
    by repeatedly applying $f_M$ it is easy to show that we can construct two tuples $u,v \in R$ such that $t \in f_{M}(u, v)$, contradicting the assumption that $f \in \mpol(\Gamma, \all)$.
    %not $\all$-preserved by $f$. Then there is $R \in \Gamma$ and $\alpha, \beta \in R$ such that there is a tuple $\gamma \in f_{\all^{\numdom{2}}_D}(\alpha,\beta) \setminus R$.
%    But for every coordinate $i$ there is a map $m : \numdom{2} \to D$ with $m(0) = \alpha[i], m(1) = \beta[i]$, $\gamma[i] \in \{m(0),m(1)\}$, hence $\gamma[i] \in \{\alpha[i],\beta[i]\}$, showing that $\Gamma$ is not essentially unary.
\end{proof}

\begin{lemma} \label{lm:express-equality}
  Let $\Gamma$ be a constraint language which is not
  essentially unary. Then $\mathsf{Eq} \in \fgpp{\Gamma}{\all}$.
\end{lemma}
\begin{proof}
    Consider the binary multivalued operation $f(0,1) = \{0,1\}$. By Lemma~\ref{lm:f-essentially-unary}, $\Gamma$ is not $\all$-preserved by $f$, hence by Proposition \ref{prop:no_preserve_f}, we can $\all$-fgpp-define a relation $R$ on domain $\numdom{2}$ satisfying $\{00,11\} \in R \subseteq \numdom{2}^2 \setminus \{01\}$. Now, we can define Boolean equality by $\mathsf{Eq}(x,y) \Leftrightarrow R(x,y) \land R(y,x)$.
\end{proof}

We can now present the dichotomy for $\ucsp(\Gamma)$.

\begin{theorem} \label{thm:arbitrary-dichotomy}
    Let $\Gamma$ be a constraint language. If $\Gamma$ is essentially unary then $\ucsp(\Gamma)$ is in P, and otherwise $\ucsp(\Gamma)$ is NP-hard and W[1]-hard parameterized by the number of variables. 
\end{theorem}
\begin{proof}
    If $\Gamma$ is essentially unary, then any instance of $\ucsp(\Gamma)$ can be solved in polynomial time by the arc-consistency algorithm outlined in Section~\ref{sec:prelims}.
    Let $I=(\{x_1,\dots,x_k\},C)$ denote an arbitrary instance of $\ucsp(\Gamma)$ and
    assume $(D_1,\dots,D_k)$ is the result from applying the arc-consistency algorithm. 
    If $D_i=\emptyset$ for some $1 \leq i \leq k$, then $I$ is not satisfiable. Otherwise,
    every assignment to the variables agreeing with $(D_1,\dots,D_k)$ is a solution to $I$. On the other hand, if $\Gamma$ is not essentially 
    unary, then by Lemma~\ref{lm:express-equality} $\Gamma$ can $\all$-fgpp-define Boolean equality,
    and W[1]-hardness and NP-hardness follow from Lemma~\ref{lm:equality-hard} via Corollaries~\ref{corr:language_fpt_reduction} and~\ref{corr:language_poly_reduction} (since $\all$ is trivially closed under composition with itself).
\end{proof}

\subsection{A Generic Lower-bound Perspective}
\label{sec:unrestr-discussion}

Let us now abstract out the ingredients of the hardness proof of
Lemma~\ref{lm:equality-hard} to paint a general image of what
makes a $\udcsp$ problem class intractable (at least, up to the
investigations in this paper). We take as starting point the set of
variables $x_{i,j}$, $i, j \in [k]$, where each variable $x_{i,j}$
has a domain $D_{i,j} \subseteq D_i \times D_j$ interpreted
as a subset of the product of two sets from a family of sets
$D_1, \ldots, D_k$.
The hardness proof then needs two types of constraints:
\begin{enumerate}
\item The \emph{projection constraints} $x_{i,j}[1] = x_{i,j'}[1]$ 
  for all $i, j, j' \in [k]$ (which might be called a \emph{vertex
    choice gadget} in the FPT literature).
\item The \emph{flip constraint} that says that $x_{i,j}=(u,v)$
  if and only if $x_{j,i}=(v,u)$ for all $i, j \in [k]$,
  $u \in D_i$, $v \in D_j$ (which might be called an \emph{edge choice
    gadget} or \emph{coordination gadget}).
\end{enumerate}
Visualised as bipartite graphs, the projection constraints look like
a collection of non-crossing ``thick edges'', i.e., a disjoint union
of bicliques, while the flip constraints are complex permutations.
Clearly, any problem $\udcsp(\Gamma,\MM)$ which can implement both
classes of constraints is W[1]-hard.

Now consider these constraints from the special perspectives of \ohcsp
and \mcsp. In \ohcsp, there is no relevant domain order and 
the basic \textsf{Eq} relation implements the flip constraint.
On the other hand, there is no obvious way to use one-hot maps to
directly build projection constraints; the two canonical hardness
constructions use \textsf{Nand}$_2$ to build arbitrary binary
constraints, and \textsf{Impl} to implement a vertex choice gadget via
the introduction of an additional variable $x_i$ for each $i \in [k]$.
See Section~\ref{sec:one_hot}.

Contrary to this, in \mcsp with the domains ordered lexicographically, 
the basic \textsf{Eq} relation can define projection constraints
via constraints $(x_{i,j} \leq (a,\cdot) \text{ iff } x_{i,j'} \leq (a,\cdot))$
(e.g., $\mathsf{Eq}([x_{i,j} \leq (a,n)], [x_{i,j'} \leq (a,n)])$
in Iverson bracket notation).
On the other hand, with a domain order in place,
the permutation of a flip constraint is harder to implement -- in
fact, the flip constraints form a class of \emph{universal permutations},
and their intractability is the basis for the rich theory of
\emph{twin-width}, capturing notions of ordered (binary) structures
of bounded complexity; see~\cite{bonnet2024twin} and Section~\ref{sec:monotone}.

It may be interesting to compare this phenomenon to the discussion of
signed logic from the introduction -- recall that \ohcsp corresponds to
regular signed formulas over the empty partial order, and \mcsp
corresponds to regular signed formulas over the total order.
That said, in the rest of the paper we treat \ohcsp and \mcsp as two
separate problems, and leave any questions of unification of results
for future work.

%%%%%%%%%%%%%%%%%%%%%%%%%%%%%%%%%%%%%%%%%%%%%%%%%%%%%%%%%%%%%%%%%%%%%%%%%%%%%%%%%%%%%%%%%%%%

\section{One-hot Maps} \label{sec:one_hot}

We continue by studying one-hot map families and the problem $\ohcsp(\Gamma)$. 
One-hot maps $m:\numdom{n} \rightarrow \{0,1\}$ satisfy $m(x) = 1$ for exactly one $x \in \numdom{n}$, and can obviously be represented by Iverson brackets $[x=a]$.
Since one-hot maps has the range $\{0,1\}$, we can always assume that the base language
$\Gamma$ is Boolean. The central result for obtaining
a P versus NP dichotomy is that $\ohcsp(\{\mathsf{Eq}\})$ is NP-hard (Lemma~\ref{lm:one-hot-one-in-three}). This severely limits the possibilities for polynomial-time
algorithms and NP-hardness is the norm for $\ohcsp(\Gamma)$. 

The situation is quite different in the parameterized setting: there exists
an interesting class of FPT problems
defined by relations being {\em weakly separable} in the sense of
Marx~\cite{Marx05CSP}.
This property has its roots in work on CSPs with global cardinality constraints
by Marx~\cite{Marx05CSP} and Bulatov and Marx~\cite{Bulatov:Marx:sicomp2014}, and we will recycle
some of their algorithmic ideas (for instance, Lemmas~\ref{lem:enumerate_minimal} and~\ref{lem:separability} are adaptations of
Lemma~3.3 and 3.4 in~\cite{Bulatov:Marx:sicomp2014}, respectively).
The hardness proof is guided by a characterization of relations that are
not weakly separable by Kratsch et al.~\cite{Kratsch:etal:toct2016}. They proved
that weak separability can be captured by two partial polymorphisms. 
The absence of (at least) one of these allows us to define relations with adverse
computational properties: these
are $\mathsf{Nand}_2$ and $\mathsf{Impl}$.

\subsection{P versus NP Dichotomy} \label{sec:one-hot-p-np}

We begin by proving that the Boolean equality relation in itself leads
to NP-hardness. The reduction starts from the following NP-hard problem.

 \problemDef{Exact 3-Hitting Set}
{ A set $S$ and a set ${\cal C}=\{C_1,\dots,C_m\}$
of subsets of $S$ such that $|C_1|=|C_2|=\dots=|C_m|=3$.}
{Is there a subset $S' \subseteq S$ such that $|S' \cap C_i|=1$,
$1 \leq i \leq n$?}

\begin{lemma}\label{lm:one-hot-one-in-three}
  \ohcsp$(\{\mathsf{Eq}\})$ is NP-hard. 
\end{lemma}
\begin{proof}
   Let
  $I=(S,\cC)$ be an instance of {\sc Exact 3-Hitting Set}. 
  For convenience, let $\cC=\{C_1,\ldots,C_m\}$
  where for every $C_i \in \cC$ we have $C_i=\{s_{i,0},s_{i,1},s_{i,2}\}$;
  clearly $s_{i,j} \in S$ for every $i, j$ and distinct indices $(i,j)$
  do not imply distinct elements in $S$. Create an instance $I'$ of
  $\ohcsp(\{\mathsf{Eq}\})$ by creating a variable $x_i$ of domain $\numdom{3}$ for
  every $C_i \in \cC$, and for $j \in \numdom{3}$ let $x_i=j$ represent
  that $v_{i,j}$ is chosen in the solution. For every distinct pair of sets
  $C_i, C_j \in \cC$ such that $s_{i,a}=s_{j,b}$ for some $a, b \in \numdom{3}$,
  add a constraint $\mathsf{Eq}([x_i=a],[x_j = b])$ to $I'$. This completes the
  construction.

  Showing correctness of the reduction is straightforward. 
  Assume $S' \subseteq S$ is a solution, i.e. there is an
  assignment
  $\varphi \colon S \to \{0,1\}$
  such that precisely one member $s_{i,j}$ of every set $C_i \in \cC$
  is mapped to $1$. Then we create an assignment $\varphi'$ for $I'$
  where $\varphi'(x_i)=j$. This satisfies all equality constraints and assigns a
  well-defined value to every variable in $I'$. The reverse is similar.
  For every satisfying assignment $\varphi'$ to $I'$, due to the
  equality constraints, every element $s_i \in S$ either has all its
  copies selected by $\varphi'$, in which case we choose $s_i$ to be a member of the solution,
  or none of its copies selected, in which case we do not put $s_i$ into the solution. 
\end{proof}

Lemma~\ref{lm:one-hot-one-in-three} shows that the computational complexity of
$\ohcsp(\Gamma)$
does not behave well under pp-definitions since they can freely introduce the equality relation.
However, the next proposition shows that equality-free pp-definitions
are useful under additional assumptions.

\begin{proposition} \label{prop:onehot-ppdef}
Let $\Gamma$ be a constraint language with domain $D$
and $\MM$ a map family such that for every $n$, there is
an $m^*_n \in \MM$ such that $D = \{m^*_n(d) \; | \; d \in D\}$.
If $R$ is a relation that
is efpp-definable in $\Gamma$, then \udcsp$(\Gamma \cup \{R\},\MM)$
is polynomial-time reducible to $\udcsp(\Gamma,\MM)$
\end{proposition}
\begin{proof}
Assume the relation $R(x_1,\dots,x_k)$ has the efpp-definition
\[R(x_1,\dots,x_k) \Leftrightarrow \exists y_1,\dots,y_l . \phi(x_1,\dots,x_k,y_1,\dots,y_l)\]
where $\phi$ is quantifier- and equality-free.
Let $(n,V,C)$ denote an instance of $\udcsp(\Gamma \cup \{R\},\MM)$.
Arbitrarily choose a constraint $c=R(m_1(x_1),\dots,m_k(x_k))$ in $C$.
Introduce $l$ fresh variables $y_1,\dots,y_l$ and replace
$c$ with the constraints in $\phi(x_1,\dots,x_k,y_1,\dots,y_m)$ where each variable $x_i$, $1 \leq i \leq k$, is replaced by $m_i(x_i)$ and
each variable $y_j$, $1 \leq j \leq l$, is replaced by $m^*_n(y_j)$. 
Since the domain of $\Gamma$ is $D$, the resulting gadget
accepts the same assignments to $x_1,\dots,x_k$ as the original
constraint $c$.
\end{proof}

We always have access to a suitable $m^*_n$ in the Boolean case:
the map $m:\numdom{n} \rightarrow \numdom{2}$ defined by $m(x)=1$ if and only if $x=1$ is a member of $\oh$ for arbitrary $n$.
 We continue with an algorithmic result.

\begin{lemma} \label{lem:one-hot-tract}
If $\Gamma \subseteq \efpp{\{\mathsf{Or}_2,0,1\}}{}$, then
$\ohcsp(\Gamma)$ is in P.
\end{lemma}
\begin{proof}
The result follows from combining Proposition~\ref{prop:onehot-ppdef}
with a polynomial-time algorithm for $\ohcsp(\{\mathsf{Or}_2,0,1\})$.
To this end,
we present a polynomial-time reduction to 2-SAT. 
Arbitrarily choose an instance $I=(n,V,C)$ of $\ohcsp(\{\mathsf{Or}_2,0,1\})$.
For each Iverson bracket $[x=a]$ appearing in $C$,
introduce a variable $x_a$. Add the clause $(\neg x_a \vee \neg x_b)$
for all $x_a,x_b$ with $a \neq b$. Do the following:

\begin{itemize}
\item
Add the clause $(x_a \vee y_b)$
for all constraints $\mathsf{Or}_2([x=a],[y=b])$.

\item
Add the clause $(x_a)$
for all constraints $1([x=a])$.

\item
Add the clause $(\neg x_a)$
for all constraints $0([x=a])$.
\end{itemize}

We claim that the resulting 2-SAT instance $I'$ is satisfiable if and only
if $I$ is satisfiable.

Assume $f$ is a solution to $I$. Assign 1 to $x_a$ if and only if
$f(x)=a$. This assignment satisfies all clauses of the type
$(\neg x_a \vee \neg x_b)$ (since  a variable cannot be assigned two distinct values)
and all clauses of the other types since $I$ is satisfiable.

Assume $f'$ is a solution to $I'$. Construct an assignment $f$
by setting $f(x)=a$ if and only if $f'(x_a)=1$. The clauses
of type $(\neg x_a \vee \neg x_b)$ guarantees that no assignment
conflict can occur; however, certain variables can be left unassigned
and we assign them arbitrary values.
Constraints of types $1([x=a)$ and $0([x=a])$ are obviously
satisfied by this assignment.
Constraints of type $\mathsf{Or}_2([x=a],[y=b])$ are satisfied
since it is clear that $f$ satisfies at least one of the Iverson brackets
in the constraint. We conclude that $I$ is satisfied by $f$.
\end{proof}

We can now prove the dichotomy result. The proof utilizes a few basic
facts concerning 
{\em Post's lattice}~\cite{pos41} of Boolean (relational) clones.
The interested reader can find an enjoyable exposition in~\cite{Bohler:etal:sigactnews2003,Bohler:etal:sigactnews2004}.

\begin{theorem} \label{thm:p-vs-np-oh}
  Let $\Gamma$ be a Boolean constraint language. Then
      $\ohcsp(\Gamma)$ is in P if and only if $\Gamma \subseteq \efpp{\{\mathsf{Or}_2, 0, 1\}}{}$. Otherwise, the problem is NP-hard.
\end{theorem}

\begin{proof}
If $\Gamma \subseteq \efpp{\{\mathsf{Or}_2,0,1\}}{}$, then 
$\ohcsp(\Gamma)$ is in P by Lemma~\ref{lem:one-hot-tract}.
 Assume instead that $\Gamma \not \subseteq \efpp{\{\mathsf{Or}_2, 0, 1\}}{}$ and let $R \subseteq \{0,1\}^{r}$ be an $r$-ary witnessing relation, i.e., $R \in \Gamma$ but $R \notin \efpp{\{\mathsf{Or}_2, 0, 1\}}{}$. 

  If $\mathsf{Eq} \in \efpp{R}{}$ then we get NP-hardness from Lemma~\ref{lm:one-hot-one-in-three}. It is then known~\cite[Lemma 7]{10.1093/logcom/exaa079} that if $\eq \notin \efpp{R}{}$, then $R$ is either (1) {\em essentially negative} and preserved by the operation $f(x,y,z) = 
  \min(x,\max(y,1-z))$ or (2) {\em essentially positive} and preserved by the operation $g(x,y,z) = \max(x,\min(y,1-z))$. Hence, the only two possible cases remaining for a complete classification is when $R$ is essentially negative or essentially positive. We make some observations.
First, $\efpp{\Gamma}{} \subseteq \efpp{\Gamma \cup \{\mathsf{Eq}\}}{}$ for any Boolean language $\Gamma$, the inclusion is strict if and only if $\mathsf{Eq} \notin \efpp{\Gamma}{}$, and there cannot exist any Boolean $\Delta$ such that $\efpp{\Gamma}{} \subsetneq \efpp{\Delta}{} \subsetneq \efpp{\Gamma \cup \{\mathsf{Eq}\}}{}$~\cite{lagerkvist2020a}. Second, if $\Gamma$ and $\Delta$ are two Boolean constraint languages where each relation in $\Gamma$ is {\em irredundant}, i.e., for each $s$-ary $S \in \Gamma$ and distinct $i,j \in [s]$ it is not the case that $t[i] = t[j]$ for every $t \in S$, then $\Gamma \subseteq \efpp{\Delta \cup \{\mathsf{Eq}\}}{}$ if and only if $\Gamma \subseteq \efpp{\Delta}{}$~\cite[Lemma 5]{lagerkvist2020a}. This, in particular, implies that if $\Gamma$ is irredundant and $\ohcsp(\Gamma)$ is in P then $\ohcsp(\Delta)$ is in P for every $\Delta \subseteq \efpp{\Gamma}{}$, and if $\ohcsp(\Gamma)$ is NP-hard then $\ohcsp(\Delta)$ is NP-hard for every $\Delta$ such that $\Gamma \subseteq \efpp{\Delta \cup \{\mathsf{Eq}\}}{}$. 
  
  We handle the remaining two cases in turn.

  \begin{enumerate}
    \item Essentially negative. We show that
     $\ohcsp(\{\textsf{Nand}_2\})$ is NP-hard. This follows by a reduction
from 3-{\sc Colourability}: consider the variable domain $\numdom{3}$ and note that
for $x,y \in \numdom{3}$, it holds that
\[x \neq y \Leftrightarrow \textsf{Nand}_2([x=0],[y=0]) \wedge \textsf{Nand}_2([x=1],
[y=1]) \wedge \textsf{Nand}_2([x=2],[y=2]).
\]
Thus, we assume that $\textsf{Nand}_2 \not\in \efpp{\Gamma}{}$.
    By using the above observations and Post's lattice, it follows that the only remaining case is when $R$ is essentially unary, but then $\ohcsp(\{R\})$ is in P by Theorem~\ref{thm:arbitrary-dichotomy}.  
    
    \item Essentially positive: if $R \in \efpp{\{\mathsf{Or}_2, 0, 1\}}{}$ then $\ohcsp(\{R\})$ is in P by Lemma~\ref{lem:one-hot-tract}. Otherwise, via Post's lattice, $\mathsf{Or}_3 \in \efpp{R}{}$. The 3-SAT-relations are efpp-definable in $\mathsf{Or}_3$ with domain $\numdom{2}$. For instance,
$(x \vee \neg y \vee z) \Leftrightarrow \mathsf{Or}_3([x=1],[y=0],[z=1])$. We conclude that
$\ohcsp(\{R\})$ is NP-hard.
  \end{enumerate}
\end{proof}

\subsection{FPT versus W[1] Dichotomy}
\label{sec:one-hot-fpt-w1}

We continue by presenting our FPT versus W[1] dichotomy for One-Hot-CSP$(\Gamma)$ when $\Gamma$ is finite and Boolean. 
We recall some terminology.  For streamlining the presentation, we introduce a special null-value $\bot$.
%Given a constraint language $\Gamma$, we let $\Gamma_{\bot}$ denote $\Gamma$
%with the element $\bot$ added to its domain.
Given a  map $f:\numdom{n} \rightarrow \numdom{d}$, we define $f_{\bot}:\numdom{n,\bot}  \rightarrow \numdom{d}$ such that
$f_{\bot}(a)=f(a)$ if $a \in \numdom{d}$ and $f_{\bot}(\bot)=0$.
We define $\oh_{\bot}=\{f_{\bot} \; | \; f \in \oh\}$. Let $f : V \to [n]_\bot$ and $g : V \to [n]_\bot$ be two assignments. We say that $f$ is an \emph{extension} of $g$ if for any $v \in V$ we have $f(v) = g(v)$ if $g(v) \neq \bot$. We say that $f$ is a \emph{minimal satisfying extension} of $g$ if $f$ is satisfying and there is no other satisfying extension $h$ of $g$ such that $f$ is an extension of $h$. 
Finally, we say that an assignment $f : V \to \numdom{n,\bot}$ is \emph{trivial} if $f(v) = \bot$ for all $v \in V$ and a \emph{minimal satisfying assignment} if $f$ is non-trivial and not an extension of any other non-trivial satisfying assignment.

Consider a finite domain $D$ and an arbitrarily chosen element $d$ in $D$.
We say that a relation $R \subseteq D^a$ is $d$-\emph{valid} if $(d,\dots,d) \in D$. Two tuples $s, t \in R$ are $d$-\emph{disjoint} if $s_i=d$ or $t_i=d$ for every $1 \leq i \leq a$. The $d$-\emph{union} $s \cup_d t$ is the tuple $(u_1,\dots,u_a)$ where 
$u_i=s_i$ if $t_i=d$ and $u_i=t_i$, otherwise.

 A relation $R$ is \emph{weakly $d$-separable} if the following two conditions hold:
\begin{enumerate}
    \item For every pair of $d$-disjoint tuples $s, t \in R$, we have $s \cup_d t \in R$.
    \item For every pair of $d$-disjoint tuples $s, t$ with $s \in R$ and $s \cup_d t \in R$, we have $t \in R$.
\end{enumerate}

Let us illustrate with
{\em affine} relations.
A relation $R \subseteq \{0,1\}^r$ is called affine if it is an affine subspace of the $r$-dimensional vector space over GF[2] or, equivalently, $R$ is invariant under the
operation $a(x,y,z)=x-y+z \; ({\rm mod} \; 2)$.
Example~2.4 in \cite{Marx05CSP} shows that affine relations are weakly 0-separable, and
our forthcoming algorithmic result (Lemma~\ref{lem:one-hot-fpt-alg}) thus implies
that
$\ohcsp(\Gamma)$ is in FPT if $\Gamma$ is a set of affine relations.
A non-example is the relation $\textsf{Nand}_2=\{00,10,01\}$: pick the 0-disjoint
tuples $01$ and $10$ and note that their 0-union is the tuple $11$
which is not a member of $R$.

One may alternatively interpret weak separability algebraically.
Kratsch et al.~\cite[Section~4]{Kratsch:etal:toct2016} point out that a Boolean language $\Gamma$ is weakly 0-separable if and only if it is invariant
under the following two partial operations on $\numdom{2}$:

\begin{enumerate}
  \item $f_{\mathsf{Nand}}(0,0,0) = 0$, $f_{\mathsf{Nand}}(0,0,1) = 1$, $f_{\mathsf{Nand}}(0,1,0) = 1$, $f_{\mathsf{Nand}}(1,1,1) = 1$, and
  \item $g_{\mathsf{Impl}}(0,0,0) = 0$, $g_{\mathsf{Impl}}(0,1,1) = 0$, $g_{\mathsf{Impl}}(0,0,1) = 1$, $g_{\mathsf{Impl}}(1,1,1) = 1$.
\end{enumerate}

We will use the relational definition of weak 0-separability for algorithmic
purposes (Lemmas \ref{lem:zero_valid}--\ref{lem:one-hot-fpt-alg})
and switch to the algebraic definition when proving hardness (Theorem~\ref{thm:one-hot-fpt-w1}).

%\medskip
%
%\begin{tabular}{cccc}
%0 & 0 & 0 & 1 \\
%0 & 0 & 1 & 1 \\
%0 & 1 & 0 & 1 \\ \hline
%0 & 1 & 1 & 1 
%\end{tabular}
%
%\medskip
%
%\begin{tabular}{cccc}
%0 & 0 & 0 & 1 \\
%0 & 1 & 0 & 1 \\
%0 & 1 & 1 & 1 \\ \hline
%0 & 0 & 1 & 1 
%\end{tabular}

We first show that combining a 0-valid and weakly 0-separable constraint language
with one-hot maps is trouble-free.

\begin{lemma}\label{lem:zero_valid}
Let $R \subseteq \numdom{d}^a$ denote a relation that is 0-valid and weakly 0-separable. Then, the relation $R'(x_1,\dots,x_a) \equiv R(m_1(x_1),\dots,m_a(x_a))$ where $m_1,\dots,m_a \in \oh_{\bot}$
is $\bot$-valid and weakly $\bot$-separable.
\end{lemma}
\begin{proof}
Recall that each map in $\oh_{\bot}$ maps $\bot$ to 0. This implies that $R'$ is $\bot$-valid
since $R$ is assumed to be 0-valid.
Arbitrarily choose two $\bot$-disjoint tuples $s'=(s'_1,\dots,s'_a)$ and $t'=(t'_1,\dots,t'_a)$ in $\numdom{d,\bot}^a$. To verify the first condition in the definition of weak separability,
we assume that $s' \in R$ and $s' \cup_{\bot} r' \in R$.
Let $s=(m_1(s'_1),\dots,m_a(s'_a))$ and  $t=(m_1(t'_1),\dots,m_a(t'_a))$. 
Note that $s,t$ are 0-disjoint tuples in $R$ and that $s' \cup_{\bot} t' = s \cup_0 t$.
The relation $R$ is weakly $0$-separable so $s \cup_0 t$ is in $R$. Hence,
the tuple $s' \cup_{\bot} t'$ is in $R'$. Verifying the second condition is similar.
\end{proof}

In the next two lemmas, we verify that solutions to $\udcsp(\Gamma,\oh_{\bot})$
can be constructed from minimal satisfying assignments.

\begin{lemma} \label{lem:enumerate_minimal}
Let $\Gamma$ be a finite and Boolean constraint language. There are functions $d_{\Gamma}(k)$ and $e_{\Gamma}(k)$ such that for any instance of $\udcsp(\Gamma,\oh_{\bot})$ with domain size $n$ and $k$ variables, every variable $v$ is non-$\bot$ in at most $d_{\Gamma}(k)$ minimal satisfying assignments and all these minimal satisfying assignments can be enumerated in time $e_{\Gamma}(k) \cdot n^{O(1)}$.
\end{lemma}
\begin{proof}
    We use a branching algorithm. We start with setting any variable to any non-$\bot$ value ($nk$ branches). While there are unsatisfied relations, branch on assigning any not-yet-assigned variable in this relation to its `hot' values (at most $k$ branches). If this is not possible because all variables have been assigned a value, we have reached a dead end. If all relations are satisfied, we have reached a satisfying assignment. Since there are only $k$ variables, we can only go up to $k$ branching levels deep, so there are at most $k^k n$ branches in total.

    Let $t$ be any minimal satisfying assignment. We will now explain why $t$ is reached by the above procedure. We prove this by showing that there is a path down the branching tree that never results in a dead end and where $t$ is an extension of the current assignment throughout this branching path. Since $t$ is minimal and an extension of the satisfying assignment at the end of this path, we conclude that this path must end at $t$.
    
    As a base case, observe that $t$ is non-trivial, so it has at least one non-$\bot$ value. So, for the first branching step, we follow the path where this variable is set to this value.
    
    Now suppose that $t$ is an extension of the current assignment (say $s$) in one of the later branching steps. Let $R$ be the unsatisfied relation from this step. Since $t$ satisfies $R$ while $s$ does not, there is some variable $x_i$ in $R$ where $t$ and $s$ are mapped to different values in $\numdom{2}$ by its map $f_i$. Since $t$ is an extension of $s$, we must have $s(x_i) = \bot$ and $t(x_i) = c_i$ for some $c_i \neq \bot$. Since $0 = f_i(\bot) \neq f_i(c_i)$, we conclude that $f_i(c_i)$ must be 1 and thus $c_i$ is the hot value of $f_i$. This shows that we can follow the path where $x_i$ is set to $c_i$. This completes the induction.
\end{proof}

Let $f$ and $g$ denote functions from $A$ to $B$.
The {\em difference} $f \setminus g$ is the set difference of
$f$ and $g$ when viewed as subsets of $A \times B$.

\begin{lemma} \label{lem:separability}
    Let $\Gamma$ be a weakly 0-separable and $0$-valid Boolean constraint language. The following hold for any instance $I = (n,V,C)$ of $\udcsp(\Gamma,\oh_{\bot})$:
    \begin{enumerate}
        \item 
        Every satisfying assignment $f \colon V \to \numdom{n,\bot}$ is the $\bot$-union of pairwise disjoint minimal satisfying assignments, and
        \item 
        If there is a satisfying assignment $f$ with $f(v) = d$ for some variable $v \in V$ and $d \in \numdom{n}$, then there is a minimal satisfying assignment $f'$ with $f'(v) = d$.
    \end{enumerate}
\end{lemma}
\begin{proof}
    Since $\Gamma$ is 0-valid and weakly 0-separable, Lemma~\ref{lem:zero_valid} shows that 
    every relation $R'(x_1,\dots,x_a) \equiv R(m_1(x_1),\dots,m_a(x_a))$ where $m_1,\dots,m_a \in \oh_{\bot}$
    is $\bot$-valid and weakly $\bot$-separable. We now prove the parts in order.
    \begin{enumerate}
        \item We use induction on the number of non-$\bot$ values in $f$. If there is only one non-$\bot$ value, then $f$ is minimal by definition since we explicitly exclude the trivial assignment. Otherwise, there are two options. Either $f$ is minimal and we are done, or there exists a minimal satisfying assignment $g$ such that $f$ is an extension of $g$. Now consider $f' := f \setminus g$ and observe that $f'$ and $g$ are $\bot$-disjoint. Since $g$ contains at least one non-$\bot$ value, $f'$ contains fewer $\bot$-values than $f$, so we can apply the induction hypothesis. This shows that $f'$ is the $\bot$-union of some set of pairwise disjoint minimal satisfying assignments. Because of weak $\bot$-separability, adding $g$ to this set results in a solution for $f$.
        
        \item Since $f$ is the disjoint $\bot$-union of minimal satisfying assignments $f_1, \ldots, f_\ell$, there must be an $f_i$ with $f_i(v) = d$. 
    \end{enumerate}
\end{proof}

\noindent
Our final ingredient is a way of handling relations that
are not 0-valid. We will do this by branching over partial assignments with the aid
of the following.
If $R \subseteq \numdom{2}^r$, then
define \[R_{|(i,1)} = \{(a_1,\dots,a_{i-1},1,a_i,\dots,a_r) \; | \; (a_1,\dots,a_r) \in R\}.\]
That is, $R_{|(i,1)}$ is obtained from $R$ by assigning the $i$th position of every tuple the value 1. The 
relation $R_{|(i,0)}$ is defined analogously. 
By applying these two operations repeatedly, we may obtain $3^r$
(not necessarily distinct)
relations.
Given a finite constraint language $\Gamma$, we let $\Gamma^*$
denote $\Gamma$ extended with these restricted relations.
Marx has proven the following result.

\begin{lemma}[Lemma 2.6 in~\cite{Marx05CSP}] \label{lem:weakly-separable-restrictions}
If the Boolean constraint language $\Gamma$ is weakly 0-separable, then 
$\Gamma^*$ is weakly 0-separable, too.
\end{lemma}

Our FPT algorithm proceeds in three stages: (1) it removes constraints that are
not 0-valid, (2) it constructs a set of tentative solutions by enumerating minimal satisfying assignments, and (3) it searches for a solution that does not assign $\bot$ to any variable.

\begin{lemma} \label{lem:one-hot-fpt-alg}
Let $\Gamma$ be a finite Boolean constraint language that is weakly 0-separable.
Then, $\ohcsp(\Gamma)$ is in FPT.
\end{lemma}
\begin{proof}
We show the result for the constraint language $\Gamma^*$. Lemma~\ref{lem:weakly-separable-restrictions} guarantees that this language
is weakly 0-separable.
Assume $r$ is the maximum arity of any relation in $\Gamma^*$.
    Let $I=(n,V,C)$ be an arbitrary instance of $\ohcsp(\Gamma^*)$.
    We first verify that this instance can be reduced to at most $p \leq r^{|V|}$ $\ohcsp(\Gamma^*)$
    instances $I_1,\dots,I_p$ such that $I_i$, $1 \leq i \leq p$, only contains constraints $R(\bar{x})$ such that $R$ is 0-valid.
    Moreover, the procedure does not increase the
number of variables.
    
    We use a branching algorithm. If
$C$ contains a constraint $R([x_1=a_1],\dots,[x_m=a_m])$ such that $R$ is not 0-valid, then
at least one $x_i \in \{x_1,\dots,x_m\}$ must be assigned the corresponding value $a_i$. 
Thus, we branch in $m \leq r$ directions by
setting $x_1=a_1$, $x_2=a_2$, $\dots$, and $x_m=a_m$.
If we consider, for example, the assignment $x_1=a_1$, then it yields the
constraint $R_{|(1,1)}([x_2=a_1],\dots,[x_m=a_m])$ and we know that the relation
$R_{|(1,1)}$ is in $\Gamma^*$. 
If the assignment becomes inconsistent in a branch --- a variable is assigned two distinct values ---
then we deem this branch unsuccessful and abort it.
Otherwise, we repeat the procedure
until every constraint is 0-valid and we output the corresponding instance of $\ohcsp(\Gamma^*)$.
The search tree of the algorithm has height at most $|V|$ so it has at most $r^{|V|}$
leaves and, thus, we have at most $r^{|V|}$
instances built over 0-valid relations to consider. It is clear that at least one
of these instances are satisfiable if and only if the original instance is satisfiable.
    
    Now, let $(n,V,C)$ be one of these preprocessed instances.
    We view it as an instance $I_{\bot}$ of $\udcsp(\Gamma^*,\oh_{\bot})$, apply Lemma~~\ref{lem:zero_valid}, and conclude that $I_{\bot}$ is $\bot$-valid and
    weakly $\bot$-separable. Next,
    we enumerate all minimal satisfying assignments using Lemma~\ref{lem:enumerate_minimal}. By Lemma~\ref{lem:separability}, every satisfying assignment is a $\bot$-union of pairwise disjoint minimal satisfying assignments, and any such $\bot$-union is satisfying.
    
    We finally want to find a satisfying assignment that does not contain the value $\bot$, or equivalently, a set of minimal satisfying assignments such that every $v \in V$ is not $\bot$ in exactly one of these assignments. To do this, we associate every minimal satisfying assignment $f$ with a subset $S_f \subset V$, defined as the set of all $v \in V$ where $f(v) \neq \bot$. There are at most $2^{|V|}$ such subsets. By brute-force enumeration, we can determine whether there exists a set of disjoint subsets whose union is $V$ in FPT time.
\end{proof}

We combine the previous algorithmic results with two hardness results 
for obtaining a full parameterized complexity classification.
The 
case division is based on the two partial operations $f_{\mathsf{Nand}}$ and $g_{\mathsf{Impl}}$ that we presented
in the beginning of this section. Thus,
violating $f_{\mathsf{Nand}}$ allows us to efpp-define $\textsf{Nand}_2$-like relations and obtain reductions from {\sc Independent Set} while
violation of $g_{\mathsf{Impl}}$ allows us to efpp-define $\textsf{Impl}$-like relations and obtain reductions from {\sc Clique}.

\begin{theorem} \label{thm:one-hot-fpt-w1}
Let $\Gamma$ be a finite Boolean constraint language. If $\Gamma$ is weakly 0-separable, then $\ohcsp(\Gamma)$ is in FPT. Otherwise, $\ohcsp(\Gamma)$ is W[1]-hard.
\end{theorem}
\begin{proof}
If $\Gamma$ is weakly 0-separable, then the result follows from
Lemma~\ref{lem:one-hot-fpt-alg}. Assume that $\Gamma$ is not weakly 0-separable.
This implies, as noted earlier, that
there exists an $r$-ary relation $R \in \Gamma$ such that $R$ is not preserved by 
\begin{enumerate}
   \item $f_{\mathsf{Nand}}(0,0,0) = 0$, $f_{\mathsf{Nand}}(0,0,1) = 1$, $f_{\mathsf{Nand}}(0,1,0) = 1$, $f_{\mathsf{Nand}}(1,1,1) = 1$, or
  \item $g_{\mathsf{Impl}}(0,0,0) = 0$, $g_{\mathsf{Impl}}(0,1,1) = 0$, $g_{\mathsf{Impl}}(0,0,1) = 1$, $g_{\mathsf{Impl}}(1,1,1) = 1$.
\end{enumerate}
We handle each case case in turn and we use the parameterized versions of {\sc Independent Set} and {\sc Clique} (see Section~\ref{sec:param-comp}) as the
 starting points for our reductions.

First, assume that $R$ is not invariant under $f_{\mathsf{Nand}}$. Then there exists $t_1, t_2, t_3 \in R$ such that 
$f_{\mathsf{Nand}}(t_1, t_2, t_3) = t \notin R.$ 
Since the application of $f$ is not projective (in the sense that $f_{\mathsf{Nand}}(t_1, t_2, t_3)  \in \{t_1, t_2, t_3\}$)  there exist distinct $i,j \in [r]$ such that 
$$\{(t_1[i], t_2[i], t_3[i]), (t_1[j], t_2[j], t_3[j])\} = \{(0,0,1), (0,1,0)\}.$$ Furthermore, for every other $k \in [r]$ it holds that 
$$(t_1[k], t_2[k], t_3[k]) \in \{(0,0,1), (0,1,0), (0,0,0), (1,1,1)\}$$ 
since otherwise the application of $f_{\mathsf{Nand}}$ would not have been defined. By identifying arguments it follows that $S \in \efpp{R}{}$ for $S \in \{S_1, S_2, S_3, S_4\}$ where
\[
\begin{aligned}
S_1 &= \{00,\ 01,\ 10\} = \mathsf{Nand}_2, \\
S_2 &= \{000,\ 010,\ 100\}, \\
S_3 &= \{001,\ 011,\ 101\}, \: {\rm and} \\
S_4 &= \{0001,\ 0101,\ 1001\}.
\end{aligned}
\]
 %$S_1 = \{(0,0), (0,1), (1,0)\}$, $S_2 = \{(0,0,0), (0,1,0), (1,0,0)\}, S_3 = \{(0,0,1), (0,1,1), (1,0,1)\}$, $S_4 = \{(0,0,0,1), (0,1,0,1), (1,0,0,1)\}$, 
 That is, we can define $\mathsf{Nand}_2$ augmented with constant arguments. 
 
 We begin by showing that $\ohcsp(\mathsf{Nand}_2)$ is W[1]-hard by a reduction from \textsc{Independent Set} and the other cases then follow by simple reductions. Hence, let $(G,k)$ be an instance of \textsc{Independent Set} and with $V(G) = \{v_0, \ldots, v_{n-1}\}$ for $n \geq 1$. We introduce $k$ variables $x_1, \ldots, x_k$ and construct an instance of $\ohcsp(\mathsf{Nand}_2)$ over domain $\numdom{n}$. For every $i \in \numdom{n}$ we introduce the constraints 
\[\bigwedge_{k_1, k_2 \in [k], k_1 \neq k_2} \mathsf{Nand_2}([x_{k_1} = i], [x_{k_2} = i])\] and for every edge $\{v_i, v_j\} \in E(G)$ the constraints 
\[\bigwedge_{k_1, k_2 \in [k], k_1 \neq k_2} \mathsf{Nand_2}([x_{k_1} = i], [x_{k_1} = j]).\] 
The correctness is clear from the construction and we turn to the relations $S_2, S_3, S_4$:
\begin{enumerate}
  \item $\ohcsp(\mathsf{Nand}_2) \leq_{{\rm FPT}} \ohcsp(S_2)$: given an instance with variables $x_1, \ldots, x_k$ over domain $\numdom{n}$ (where we without loss of generality assume that $n \geq 2$ since $n = 1$ is trivial) we introduce one fresh variable $y_0$, arbitrarily choose $c \in \numdom{n}$, and replace each constraint $\mathsf{Nand}([x_i = d], [x_j = d'])$ with $S_2([x_i = d], [x_j = d'], [y_0 = c])$.
  \item $\ohcsp(\mathsf{Nand}_2) \leq _{{\rm FPT}}\ohcsp(S_3)$: entirely analogous to the above case.
  \item $\ohcsp(\mathsf{Nand}_2) \leq_{{\rm FPT}} \ohcsp(S_4)$: similar to the above cases but we introduce two variables $y_0$ and $y_1$ instead.
\end{enumerate}
We now continue with the second case where $R$ is not preserved by $g_{\mathsf{Impl}}$. By following the above arguments one can prove that $T_1, T_2, T_3, T_4 \in \efpp{R}{}$ for 
\[
\begin{aligned}
T_1 &= \{00,\ 01,\ 11\} = \mathsf{Impl}, \\
T_2 &= \{000,\ 010,\ 110\}, \\
T_3 &= \{001,\ 011,\ 111\}, \: {\rm and}\\
T_4 &= \{0001,\ 0101,\ 1101\}.
\end{aligned}
\]
%\[T_1 = \{(0,0), (0,1), (1,1)\}, T_2 = \{(0,0,0), (0,1,0), (1,1,0)\}, T_3 = \{(0,0,1), (0,1,1), (1,1,1)\}, T_4 = \{(0,0,0,1), (0,1,0,1), (1,1,0,1)\},\] 
Thus, it is possible to define $\mathsf{Impl}$ augmented with constant arguments. We can then prove W[1]-hardness of $\ohcsp(\mathsf{Impl})$ as follows. Let $(G,k)$ denote an arbitrary instance of parameterized \textsc{Clique} where we without loss of generality assume that $V(G) = \{v_0,\dots,v_{n-1}\}$ for some $n \geq 1$. We introduce $k$ variables $x_0, \ldots, x_{k-1}$ and $k \choose 2$ variables $\{e_{\{i,j\}} \mid i,j \in \numdom{k}, i \neq j\}$, all over domain $\numdom{n}^2$. 
%\todo[inline]{Minor encoding issues, fix later.}
Here, the intention is that $x_1, \ldots, x_k$ encode the selection of the elements in the $k$-clique, and should thus have domain $V(G)$ which we encode by $(0, 0), \ldots, (0, n-1)$. To force the edge variables to not take values outside $E(G)$ we for each $(a,b) \in \numdom{n}^2 \setminus \{(c,d) \; | \; \{v_c,v_d\} \in E(G)\}$ and every $e_{\{i,j\}}$ variable introduce the constraint $\mathsf{Imp}([e_{\{i,j\}} = (a,b)], [e_{\{i,j\}} = (b,a)])$.
%If this is not possible, i.e., if $|V| > |E|$, then we use $[|V|]$ as the domain instead and represent edges by elements from $V$. 
%Then we simply add constraints $\mathsf{Imp}([x_i = (x,y)], [x_i = e(x)])$ and $\mathsf{Imp}([x_i = (y,x)], [x_i = e(x)])$ for every edge $(x,y) \in E$ where $(x,y) \neq e(x)$.
Next, for each edge $\{v_a, v_b\} \in E$ introduce constraints $$\bigwedge_{i,j \in [k], i \neq k} \mathsf{Impl}([e_{\{i,j\}} = (a,b)], [x_i = (0,a)])$$ and  
$$\bigwedge_{i,j \in [k], i \neq k} \mathsf{Impl}([e_{\{i,j\}} = (a,b)], [x_j = (0,b)]).$$ Similarly to the $\mathsf{Nand}_2$ case above it is easy to show that hardness transfers to $\ohcsp(T_2)$, $\ohcsp(T_3)$, and $\ohcsp(T_4)$, and we conclude that $\ohcsp(R)$ is W[1]-hard.
\end{proof}

The algorithm underlying Theorem~\ref{thm:one-hot-fpt-w1} is quite flexible and can be
extended in various directions. One may, for instance, use it for solving the
$t$-Hot-CSP$(\Gamma)$ problem where we have maps $m : \numdom{n} \to \numdom{2}$ such that $m(x) = 1$ for at most $t$ values $x \in \numdom{n}$. This problem is in FPT when $\Gamma$ is weakly 0-separable: we only need
to slightly generalize the brute-force steps in the algorithm from Theorem~\ref{thm:one-hot-fpt-w1}.
Since the hardness result is applicable also for $t$-Hot-CSP$(\Gamma)$, we arrive at
a full FPT/W[1] dichotomy. The P/NP-dichotomy for One-Hot-CSP$(\Gamma)$ can also
be generalized to $t$-Hot-CSP$(\Gamma)$.
It is not difficult to prove that $2$-Hot-CSP$(\{\mathsf{Or}_2\})$ is NP-hard by a reduction from 3-{\sc Colourability}:
consider the variable domain $[3]$ and note that
\[x \neq y \Leftrightarrow \mathsf{Or}_2([x \in \{0,1\}],[y \in \{0,1\}])
\wedge
 \mathsf{Or}_2([x \in \{0,2\}],[y \in \{0,2\}])
\wedge
 \mathsf{Or}_2([x \in \{1,2\}],[y \in \{1,2\}]).
\]
Additionally, there is a simple polynomial-time algorithm for $t$-Hot-CSP$(\{0,1\})$.
We thus get a P/NP-dichotomy
for $t$-Hot-CSP$(\Gamma)$, $t \in \Nat$, by
combining these observations with the classification in Section~\ref{sec:one-hot-p-np}.
%\todo[inline]{MW: Is there a notion of maps $m \colon [n] \to [d]$ for
%$d>2$ that generalise one-hot into the CCSP or OCSP considered by
%Bulatov and Marx? Does that make sense?}

%%%%%%%%%%%%%%%%%%%%%%%%%%%%%%%%%%%%%%%%%%%%%%%%%%%%%%%%%%%%%%%%%%%%%%%%%%%%%%%%%%%%%%%%%%%%

\section{Monotone Maps} \label{sec:monotone}

We now turn our attention towards monotone map families and the problem $\mcsp(\Gamma)$. Determining the classical complexity of this problem turns out to be relatively straightforward and we obtain a complete dichotomy in Section~\ref{sec:monotone_classical}. The parameterized complexity is, as expected, harder, but in the light of Theorem~\ref{thm:galois-wrapup} one should expect this question to be addressable by $\mo$-preserving polymorphisms. We may also assume that the language in question is not covered by the tractable cases in Section~\ref{sec:monotone_classical} since the parameterized complexity is not interesting in these cases.

We tackle this by introducing the {\em connector polymorphism} in Section~\ref{sec:line} whose presence/absence is intrinsically linked to whether the source language $\Gamma$ can express all permutations. Indeed, we manage to show that if $\Gamma$ is not preserved by the connector polymorphism then $\Gamma$ can express all permutations in which case $\mcsp(\Gamma)$ is W[1]-hard. Whether the presence of the connector polymorphism is sufficient for an FPT algorithm is an intriguing question
and we successfully resolve this (in Section~\ref{sec:line_fpt}) for binary source languages by a twin-width style dynamic programming algorithm. 
We continue in Section~\ref{sec:non_binary} by investigating other avenues for proving a higher-arity FPT/W[1] dichotomy by studying additional consequences of the connector polymorphism.
%This result tempt us to formulate a working hypothesis.
%
%\begin{conjecture}
%Let $\Gamma$ denote a finite constraint language.
%If $\Gamma$ is invariant under the connector polymorphism, then $\mcsp(\Gamma)$ 
%can be solved in FPT time.
%Otherwise, $\mcsp(\Gamma)$ is W[1]-hard.
%\end{conjecture}
Finally, we want to
demonstrate that $\mcsp(\Gamma)$ restricted to binary constraints is a powerful formalism in its own right. We do so by  showcasing that parts of the parameterized dichotomy for \textsc{MinCSP}~\cite{KimKPW23fa3} can be greatly simplified by the $\mcsp(\Gamma)$ framework (Section~\ref{sec:boolean_mincsp}).

%********************
%
%In Section~\ref{sec:monotone-binary-output} we consider higher-arity relations and provide a natural construction which shows that the subproblem of $\mcsp(\Gamma)$ restricted to binary relations is in FPT if $\Gamma$ has the line polymorphism (but where $\Gamma$ itself may have higher arity).
%We continue (Section~\ref{sec:unbounded-rankwidth}) by proving that neither rank-width nor functionally equivalent parameters such as clique-width or boolean width captures our problem since there exists a family of relations invariant under the line polymorphism that has unbounded rank width. Last, in Section~\ref{sec:projected-gridrank} we consider other parameters than twin-width which seem to be easier to generalize to higher arity, and present the new parameter {\em projected grid-rank}.

\subsection{P versus NP Dichotomy} \label{sec:monotone_classical}

The classical complexity of $\mcsp(\Gamma)$ can be completely captured by the three pattern polymorphisms $\pmin$, $\pmax$ and $\pmedian$ from Definition \ref{def:minmaxmedian}. We will show that the presence of any of these polymorphisms results in a polynomial-time algorithm based on arc consistency, while the absence of all three results in NP-hardness via a reduction from {\sc Exact 3-Hitting Set}. We begin with the presence of $\pmin$ or $\pmax$.

\begin{lemma} \label{lem:minandmaxisinP}
    Let $\Gamma$ be a constraint language. Suppose that $\pmin$ or $\pmax$ $\mo$-preserve $\Gamma$. Then, $\mcsp(\Gamma)$ is in $P$.
\end{lemma}
\begin{proof}
    We prove the statement for $\pmin$. The proof for $\pmax$ is analogous but with the domain ordering reversed. Let $I = (n, V, C)$ be an instance of $\mcsp(\Gamma)$ where $|C| = m$, $n$ domain size, and $V = \{x_1,\dots,x_k\}$. 
    We apply the polynomial-time arc-consistency algorithm from Section~\ref{sec:csp} to $I$
    and let $(D_1,\dots,D_k)$ denote the vector of variable domains it returns.
    If $D_i=\emptyset$ for some $1 \leq i \leq k$, then $I$ has no solution.

    Otherwise, we claim that we can set every variable to the smallest remaining value in its domain. Suppose to the contrary that some constraint $R(x_1, \ldots, x_r) \in C$ is not satisfied, and let $c_1, \ldots, c_r$ be their values. For each $i$ with $1 \leq i \leq r$, there exists, by arc consistency, a tuple $(c'_1, \ldots, c'_r)$ satisfying $R$ with $c'_i = c_i$ and where $c'_j \in D_j$ ($1 \leq j \leq r$). Since each $c_j$ is the smallest remaining value in the domain of $x_j$, we have $c'_j \geq c_j$. Now, consider applying the $\pmin$ pattern to all $r$ tuples achieved this way. For each $i$, this is allowed and results in $c_i$ since this value is reached, all other values are larger, and $\mo$ preserves domain ordering. Since $\pmin$ $\mo$-preserves $\Gamma$, the resulting tuple $(c_1, \ldots c_n)$ is in $R$, which is a contradiction and proves the claim. This completes the proof.
\end{proof}

For $\pmedian$, we observe that any interpretation of $\median$ satisfies \[\median(x,x,y) = \median(x,y,x) = \median(y,x,x) = x\] and is thus a \emph{majority polymorphism}. We use the following result.

%\begin{theorem}[Theorem 56 in \cite{DBLP:conf/dagstuhl/BartoKW17}]
%If a constraint language $\Gamma$ has a majority polymorphism, then $\Gamma$ has width
%$(2, 3)$.
%\end{theorem}

\begin{theorem}[Theorems 61 and 68 in \cite{DBLP:conf/dagstuhl/BartoKW17}]
If a constraint language $\Gamma$ has a majority polymorphism, then CSP$(\Gamma)$
can be solved by the singleton arc-consistency procedure.
\end{theorem}
\begin{corollary} \label{corr:medianisinP}
    Let $\Gamma$ be a constraint language. Suppose that $\pmedian$ $\mo$-preserves $\Gamma$. Then, $\mcsp(\Gamma)$ is in $P$.
\end{corollary}
\begin{proof}
    Let $I = (n, V, C)$ be an instance of $\mcsp(\Gamma)$. We now convert this into a conventional CSP on the domain $\numdom{n}$ by replacing every constraint (of, say, arity $r$) with its set of satisfying assignments in $\numdom{n}^r$. Since $\pmedian$ $\mo$-preserves $\Gamma$, each of these constraints is preserved by the conventional $\median$ polymorphism on domain $\numdom{n}$. This is a majority polymorphism, hence the new CSP can be solved by the singleton arc-consistency procedure.
\end{proof}

We continue with NP-hardness. The first step is to use Proposition \ref{prop:anti_monotone_maps} which allows us to include anti-monotone maps as guarding functions. The next key step in our proof is the presence of certain permutations. Recall that we write $\sigma_{abc}$ for the permutation $\sigma(0) = a$, $\sigma(1) = b$, $\sigma(2) = c$.

\begin{lemma} \label{lem:no_median_perm_2}
    Let $\Gamma$ be a language that is not $(\mo\cup\mo')$-preserved by $\pmedian$.
    Then $\fgpp{\Gamma}{\mo\cup\mo'}$ contains all reversal permutations.
\end{lemma}
\begin{proof}
    The conventional $\median$ can be written in terms of $\min$ and $\max$ using
    \[
        \median(x,y,z) = \min(\max(x,y),\max(y,z),\max(z,x))
    \]
    Now, if $\Gamma$ were $(\mo\cup\mo')$-preserved by $\pmin$, then because of guarding function symmetry it would be $(\mo\cup\mo')$-preserved by $\pmax$ as well, and thus by $\pmedian$, which is a contradiction. Hence, $\Gamma$ is not $(\mo\cup\mo')$-preserved by $\pmin$.
    
    Proposition \ref{prop:no_preserve_f} now shows that we can $(\mo\cup\mo')$-fgpp-define a relation $R$ on the domain $\numdom{2}$ satisfying $\{01, 10\} \in R \subseteq \numdom{2}^2 \setminus \{00\}$. By reversing the order of all guarding functions, we additionally find a relation $R'$ satisfying $\{01, 10\} \in R \subseteq \numdom{2}^2 \setminus \{11\}$. Together, $R \land R'$ exactly defines the binary reversal $\sigma_{10}$. All other reversal permutations follow analogously to Lemma \ref{lemma:can_define_reverse}.
\end{proof}

\begin{lemma} \label{lem:no_median_perm_3}
    Let $\Gamma$ be a language that is not $(\mo\cup\mo')$-preserved by $\pmedian$.
    Then $\efpp{\Gamma}{\mo\cup\mo'}$ contains all permutations on the domain $\numdom{3}$.
\end{lemma}
\begin{proof}
    Since $\Gamma$ is not closed under $\pmedian$, Lemma \ref{lem:no_median_perm_2} shows that we can $\mo$-fgpp-define the permutations $\sigma_{10}$ and $\sigma_{210}$ on domains $\numdom{2}$ and $\numdom{3}$. Meanwhile, Proposition \ref{prop:no_preserve_f} shows that we can $\mo$-fgpp-define a 6-ary relation $R$ on the domain $\numdom{3}$ satisfying $\{001122, 120201, 212010\} \in R \subseteq \numdom{3}^6 \setminus \{111111\}$.

    We now construct a new ternary relation $R_\textrm{cycle}$ on $\numdom{3}^3$. Let $(x,y,z)$ be the coordinates. We define $R_\textrm{cycle}$ as follows:
    \[
        R_\textrm{cycle}(x,y,z) \equiv R(i(x),r(z),i(y),r(y),i(z),r(x)) \land \sigma_{10}(f(x),g(y)) \land \sigma_{10}(f(y),g(z)) \land \sigma_{10}(f(z),g(x))
    \]

    where $r,i \colon \numdom{3} \to \numdom{3}$ are, respectively, the order-reversing permutation and identity function over $\numdom{3}$, $f$ is the monotone guarding function satisfying $0 \mapsto 0, 1 \mapsto 0, 2 \mapsto 1$, and $g$ instead satisfies $0 \mapsto 0, 1 \mapsto 1, 2 \mapsto 1$. We now claim that $R_\textrm{cycle}$ is satisfied by precisely the relations 012, 120, 201.

    Because of the chosen guarding functions, each relation $\sigma_{10}$ is satisfied by precisely five tuples: $01, 02, 11, 12, 20$. Hence, if a solution to $R_\textrm{cycle}$ contains a 2, the next coordinate must be a 0, and the remaining coordinate can then only be a 1. Similarly, if it contains a 0, then the previous coordinate must be a 2 and the final coordinate a 1. This leaves only four tuples that potentially satisfy $R_\textrm{cycle}$: 012, 120, 201, 111. Of these, 111 is excluded because of $R'$. Finally, the tuples 012, 120, 201 are valid; they satisfy all relations. This completes the claim.

    Overall, we note that variables $x$ and $y$ together precisely form the graph of the permutation 120. Hence, using $z$ as a quantified variable, we obtain a $\mo\cup\mo'$-efpp-definition of $\sigma_{120}$. Together with the reversal permutation 210, this generates all permutations on \numdom{3}, completing the proof.
\end{proof}

We now show that the presence of these permutations implies hardness.

\begin{lemma} \label{lem:perm_3_hardness}
    Let $\Gamma$ be a constraint language that $(\mo\cup\mo')$-efpp-defines all permutations on the domain $\numdom{3}$ and the binary reversal relation $\sigma_{10}$. Then $\mcsp(\Gamma)$ is NP-hard.
\end{lemma}
\begin{proof}
    Since $\Gamma$ $(\mo\cup\mo')$-efpp-defines $\sigma_{10}$ we use Corollary \ref{corr:language_poly_reduction} and Proposition \ref{prop:can_compose} to assume without loss of generality that $\sigma_{10} \in \Gamma$.
    Let $\NN$ be the map family containing all permutations on the domain $\numdom{3}$. Since $\Gamma$ $\mo$-efpp-defines all these functions, Proposition \ref{prop:maps_poly_reduction} shows that, without loss of generality, we have access to all guarding functions from $\mo \circ \NN$, which consists of all functions whose domain is $\numdom{3}$.

     We show a reduction from the NP-hard problem {\sc Exact 3-Hitting Set} (that was defined in Section~\ref{sec:one-hot-p-np}).
    Let $I = (S,\cC)$ be an arbitrary instance of 
      this problem with $S=\{s_1,\dots,s_m\}$. For every set
    $C \in \cC$, create a variable $v_C$ with domain $\numdom{3}$. For a set $C = \{s_i,s_j,s_k\} \in \cC$ with $i < j < k$, we interpret the corresponding variable $v_C$ as follows: 
    $v_C$ is assigned $0$
    if $s_i$ is chosen, it is assigned $1$ if $s_j$ is chosen, and it is assigned $2$
    if $s_k$ is chosen.
    For every pair of sets $C,D \in \cC$ that shares some element, we now need to coordinate  $v_C$ and $v_D$. Assume for instance, that $C=\{s_1,s_2,s_3\}$
    and $D=\{s_2,s_4,s_5\}$. Then, we need to enforce that $v_C=1$ if and only if
    $v_D=0$.
To this end, we add the constraint $\sigma_{10}(f(v_c),g(v_d))$, with $f$ the function mapping $1$ to 1 and everything else to 0 and with $g$ the function mapping $0$ to 0 and everything else to 1. This idea can clearly be generalized
to all sets appearing in $\cC$ and
this completes the reduction.

    It follows from the construction that every solution to $I$ can be used for obtaining a solution to the new instance; simply assign the variables 
    values
    according to the interpretation of domain values given above.   
    Conversely, 
    let $f$ be a satisfying assignment to the new instance. Cycle through all
    variables $v_C$ and add the corresponding element to the solution $S'$ of $I$.
    This implies that $S'$ intersects exactly one element in every set $C \in \cC$
    due to the coordination constraints.
    Finally, the reduction is clearly polynomial, which completes the proof.
\end{proof}

By combining the results in this section we obtain the following dichotomy for the combination of monotone and anti-monotone maps.

\begin{theorem} \label{thm:p-vs-np-antimono}
    Let $\Gamma$ be a constraint language. If $\Gamma$ is $(\mo\cup\mo')$-preserved by $\pmedian$, then $\udcsp(\Gamma,\mo\cup\mo')$ is in P. Otherwise, $\udcsp(\Gamma,\mo\cup\mo')$ is NP-hard.
\end{theorem}
\begin{proof}
    If $\Gamma$ is $(\mo\cup\mo')$-preserved by $\pmedian$, then it is also $\mo$-preserved by $\pmedian$ hence $\udcsp(\Gamma,\mo\cup\mo')$ is in P by Corollary \ref{corr:medianisinP}.
    Otherwise, Lemma \ref{lem:no_median_perm_3} shows that we can $\mo$-efpp-define all permutations from the domain $\numdom{3}$. In turn, Lemma \ref{lem:perm_3_hardness} shows that $\udcsp(\Gamma,\mo\cup\mo')$ is NP-hard.
\end{proof}

Additionally, we obtain the following dichotomy for $\mcsp(\Gamma)$.

\begin{theorem} \label{thm:p-vs-np-mono}
    Let $\Gamma$ be a constraint language. If $\Gamma$ is $\mo$-preserved by $\pmin$, $\pmax$ or $\pmedian$, then $\mcsp(\Gamma)$ is in P. Otherwise, $\mcsp(\Gamma)$ is NP-hard.
\end{theorem}
\begin{proof}
    If $\Gamma$ is $\mo$-preserved by $\pmin$, $\pmax$ or $\pmedian$, then $\mcsp(\Gamma)$ is in P by Lemma \ref{lem:minandmaxisinP} and Corollary \ref{corr:medianisinP}.
    Otherwise, Lemma \ref{lemma:can_define_reverse} shows that we have domain reversal. Hence, $\mcsp(\Gamma)$ reduces to $\udcsp(\Gamma,\mo\cup\mo')$. Since $\pmedian$ does not $\mo$-preserve $\Gamma$, it also does not $(\mo\cup\mo')$-preserve it, hence Theorem \ref{thm:p-vs-np-antimono} shows that $\mcsp(\Gamma)$ is NP-hard.
\end{proof}

\subsection{Towards a Parameterized Complexity Dichotomy: \\the Connector Polymorphism}
\label{sec:line}

We now turn to the parameterized complexity of $\mcsp(\Gamma)$ which turns out to be a significantly more interesting question and which requires the new algebraic theory developed in Section~\ref{sec:algebra}. First, recall from Section~\ref{sec:unrestr-discussion} that tractability/hardness of $\mcsp(\Gamma)$ is intrinsically linked to whether $\Gamma$ can fgpp-define all permutations (over all possible domains). Second, we manage to describe this via an ordered polymorphism condition and define an operation called the {\em connector polymorphism}. Third, we prove that the connector polymorphism is a necessary condition for FPT in the sense that if $\Gamma$ is {\em not} invariant under the connector polymorphism then $\Gamma$ can fgpp-define all permutations, in turn implying that $\mcsp(\Gamma)$ is W[1]-hard. 

As a simplifying assumption we in the parameterized complexity setting may assume that $\Gamma$ is domain-reversible, i.e., that we in the context of fgpp-definability have access to both monotone and anti-monotone maps. To see this, recall that if $\pmin \in \mpol(\Gamma, \mo)$ or $\pmax \in \mpol(\Gamma, \mo)$ then $\mcsp(\Gamma)$ is in P (Lemma~\ref{lem:minandmaxisinP}) but if $\pmin \notin \mpol(\Gamma, \mo)$ and $\pmax \notin \mpol(\Gamma, \mo)$ then $\Gamma$ is domain-reversible from Lemma~\ref{lemma:can_define_reverse}. 

%In Section~\ref{sec:line_fpt} we turn to the dual question of proving FPT when $\Gamma$ is invariant under the line polymorphism.

\begin{example} \label{ex:flip}
Recall from Section~\ref{sec:unrestr-discussion} that W[1]-hardness for $\mcsp(\Gamma)$ is linked to being able to define all permutations, since this makes it possible to define  a series of universal permutations (such as the flip permutations) which can be used as gadgets in W[1]-hardness reductions. 
% On the other hand, the flip relation, seen as a permutation over $D$, 
%  is extremely powerful -- it is a \emph{universal permutation} for
%  permutations over domains of size $n$ (as the subsequence $(0,
%  \sigma(0)), \ldots, (n-1, \sigma(n-1))$ shows for every permutation
%  $\sigma$ on $\numdom{n}$). Moreover, it is easy to see that it is an {\em involution} ($F(F(x)) = x$ for each $x \in D$)\todo{PJ: I corrected definition of involution. This is what we want?} with exactly $n$ fixpoints $(0,0), \ldots, (n-1,n-1)$.  Hence, being able to define the flip relation is intrinsically linked to being able to define all permutations, and if we know  if $\Gamma$ for each $n \geq 1$ can fgpp-define the graph of every permutation then we can also define the flip relation for each $n \geq 1$. 
  %TODO: the flip permutation is permutation over [n]^2, so already definable by assumption. Precisely, it's an involution with n fixpoints. It's possible to visualize these fixpoints as a line.
  %$F_n(x,y) \equiv \bigwedge_{\sigma \text{ is a permutation over }n} \sigma^\bullet(x,y)$.
%  Then $\mcsp(\{P,F\})$ is W[1]-hard parameterized by the number
 % of variables by a standard reduction from $k$-Clique (e.g., this is
  %the basis for the hardness of \textsc{Paired Min-Cut} and was
  %perhaps first used by Marx and Razgon). 
% The projection relation can be fgpp-defined (using monotone maps) by Boolean equality.
  For a concrete example of this scenario consider the relation 
\[
  \mathsf{Even}_4(a,b,c,d) \equiv a+b+c+d = 0 \pmod 2.
\]
Then we claim that $\mcsp(\{\mathsf{Even_4}\})$ is W[1]-hard. First, for any $a,b \in D$ let $R_{(a \leftrightarrow b)} = \{(i,j) \in D^2 \mid (i = a) \leftrightarrow (j = b)\}$ be the relation constraining its first argument to be $a$ if and only if its second argument is $b$.  We can implement
$R_{(x=i \leftrightarrow y=j)}$ as
\[
  \mathsf{Even}_4([x \geq i], [x \geq i+1], [y \geq j], [y \geq j+1]).
\]
We can then define an arbitrary permutation $\sigma$ by for each $(a,b) \in \sigma^\bullet$ introducing a constraint $R(x,y)$. This generates all permutations so we can implement the flip permutation, and get W[1]-hardness by the reduction considered in Section~\ref{sec:unrestr-hard}.
\end{example}

On a side-note, $\textsf{Even}_4$ is an affine relation so the results in Section~\ref{sec:one-hot-fpt-w1} tell us that the problem 
$\ohcsp(\{\textsf{Even}_4\})$ is in FPT. Furthermore,
we know that $\ohcsp(\{\textsf{Nand}_2\})$ is W[1]-hard. This relation is closed under $\min$ so
$\mcsp(\{\textsf{Nand}_2\})$ is in P by Lemma~\ref{lem:minandmaxisinP} and thus in FPT.
This demonstrates that the sets of constraint languages that are in FPT for
$\ohcsp$ and $\mcsp$, respectively, are incomparable.

We now want to define an algebraic condition which implies that we can fgpp-define all permutations. With the construction of Example~\ref{ex:flip} in mind we for every permutation $\sigma \colon \numdom{n} \to \numdom{n}$ and each $(a,b) \in \sigma^\bullet$ need to force $x$ to equal $a$ if and only if $y$ equals $b$. If we then work over the domain $\numdom{3}$ then we interpret the domain values as $\{<,=,>\}$ and define
define the relation $R_3 = \{(0,0), (0,2), (1,1), (2,0), (2,2)\}$ over $\numdom{3}$ (visualized in Figure~\ref{fig:r3}). It is then straightforward to prove that $R_3$ generates all permutations. 

\begin{lemma} \label{lemma:r3_can_define}
  Let $\sigma \colon \numdom{n} \to \numdom{n}$ be a permutation. Then $\sigma^\bullet \in \fgpp{\{R_3\}}{\mo}$.
\end{lemma}  

\begin{proof}
  For $t \in \numdom{n}$ let $f_t \colon \numdom{n} \to \numdom{3}$ be the ``threshold'' map defined by $f_t(x) = 0$ if $x < t$, $f_t(t)=1$, and $f_t(x)=2$ if $x > t$. This map is clearly monotone so $f_t \in \mo$ for each $t \in \numdom{n}$.
  For each $(a,b) \in \sigma^\bullet$ we then introduce a constraint $R_3(f_a(x), f_b(y))$, stating that $x$ equals $a$ if and only if $y$ equals $b$, but which does not remove any other assignment. 
\end{proof}

\begin{figure}[tbp]
\centering
\begin{minipage}{0.5\textwidth}
\centering
\begin{tikzpicture}[scale=1, every node/.style={scale=1}]
  % Left side nodes
  \node[draw, circle] (L1) at (0,4) {0};
  \node[draw, circle] (L2) at (0,2) {1};
  \node[draw, circle] (L3) at (0,0) {2};
  % Right side nodes
  \node[draw, circle] (R1) at (4,4) {0};
  \node[draw, circle] (R2) at (4,2) {1};
  \node[draw, circle] (R3) at (4,0) {2};
  % Edges representing the relation R_3
  \draw (L1) -- (R1); % (1,1)
  \draw (L1) -- (R3); % (1,3)
  \draw (L2) -- (R2); % (2,2)
  \draw (L3) -- (R1); % (3,1)
  \draw (L3) -- (R3); % (3,3)
  \draw[red, dashed, thick] (L1) -- (R2);
  \draw[red, dashed, thick] (L3) -- (R2);
  \draw[red, dashed, thick] (R1) -- (L2);
  \draw[red, dashed, thick] (R3) -- (L2);
\end{tikzpicture}
\end{minipage}%
\begin{minipage}{0.5\textwidth}
\centering
\renewcommand{\arraystretch}{2}
\begin{tabular}{@{}c|ccc@{}}
\toprule
  & \textbf{0} & \textbf{1} & \textbf{2} \\
\midrule
\textbf{0} & 1 & \textcolor{red}{0} & 1 \\
\textbf{1} & \textcolor{red}{0} & 1 & \textcolor{red}{0} \\
\textbf{2} & 1 & \textcolor{red}{0} & 1 \\
\bottomrule
\end{tabular}
\end{minipage}
% \begin{minipage}{0.2\textwidth}
% \centering
% \renewcommand{\arraystretch}{2}
% \begin{tabular}{@{}c|ccc@{}}
% \toprule
%   & \textbf{0} & \textbf{1} & \textbf{2} \\
% \midrule
% \textbf{0} & 1 & 1 & 1 \\
% \textbf{1} & 1 & \textcolor{red}{0} & 1 \\
% \bottomrule
% \end{tabular}
% \end{minipage}
% \begin{minipage}{0.2\textwidth}
% \centering
% \renewcommand{\arraystretch}{2}
% \begin{tabular}{@{}c|ccc@{}}
% \toprule
%   & \textbf{0} & \textbf{1} \\
% \midrule
% \textbf{0} & 1 & 1 \\
% \textbf{1} & 1 & \textcolor{red}{0}  \\
% \textbf{2} & 1 &1 \\
% \bottomrule
% \end{tabular}
% \end{minipage}
\caption{Bipartite graph and adjacency matrix representation of the relation $R_3$ over $\numdom{3}$. The red, dashed edges and the zeroes represent the four conditions broken by the connector polymorphism.}
\label{fig:r3}
\end{figure}
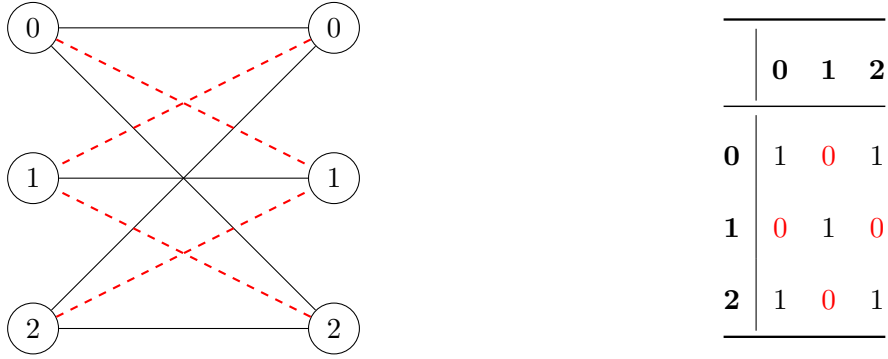

%\begin{lemma}
%  Let $\Gamma$ be a domain-reversible language over $[d]$. If if $R_3 \in \fgpp{\Gamma}{???}$ then $\mcsp(\Gamma)$ is W[1]-hard.
%\end{lemma}
%\begin{proof}
%  ...
%\end{proof}
In particular we for each $n \geq 1$ can define the flip relation from Section~\ref{sec:unrestr-discussion} and obtain W[1]-hardness for $\mcsp(\Gamma)$ which can fgpp-define all permutations. Hence, we are now interested in languages which {\em cannot} fgpp-define $R_3$, and via Theorem~\ref{thm:galois-wrapup} we know that this question can be studied on the operational side. Hence, we define the simplest possible ordered polymorphism which breaks $R_3$, where the intuition is that a simple operation over $[3]$ which does not preserve $R_3$ is given by $f(00, 02, 11, 20, 22) = 01 \notin R_3$. In the following definition, recall that we for a multifunction $f$ and a map family $\MM$ write $f_{\MM}$ for the homomorphic image of $f$ induced by $\MM$.

\begin{definition} \label{def:line}
  The {\em connector pattern} is defined by $\pcon(0, 0, 1, 2, 2) = 0$ and $\pcon(0, 2, 1, 0, 2) = 1$ over $\numdom{3}$. For each $d \geq 2$ we define the {\em connector function} (denoted $\fcon_d$) for the operation $\pcon_{\mo \cup \mo'}$.
\end{definition}

If $\pcon \in \mpol(\Gamma, \mo \cup \mo')$ (i.e., for every $R \in \Gamma$ over domain $\numdom{d}$, $\fcon_d \in \mpol(R)$) then $\Gamma$ is said to have the {\em connector property}. Note that if $\Gamma$ does not have the connector property then there is $R \in \Gamma$ over domain $\numdom{d}$ such that $\fcon_d \notin \mpol(R)$.

\begin{example}
  Consider all monotone and anti-monotone maps from $\numdom{3}$ to $\numdom{3}$. Each such map in combination with any of the two defining identities $\pcon(0, 0, 1, 2, 2) = 0$ and $\pcon(0, 2, 1, 0, 2) = 1$ of the connector pattern gives an entry in the function $\fcon_3$. For example, $\fcon_3$ is defined for:
  \begin{enumerate}
    \item
    $\fcon_3(0, 0, 0, 0, 0) = 0$, $\fcon_3(1, 1, 1, 1, 1) = 1$, $\fcon_3(2, 2, 2, 2, 2) = 2$ (by constant maps),
  \item
    $\fcon_3(0, 0, 1, 2, 2) = 0$ and $\fcon_3(2,2,1,0,0) = 2$ (by identity and order reversal),
  \item
    $\fcon_3(0, 0, 0, 2, 2) = 0$ and $\fcon_3(2,2,0,0,0) = 2$ (by $m(0) = m(1) = 0, m(2) = 2$ and order reversal).
  \end{enumerate}
\end{example}  

By construction, $R_3$ is not invariant under $\fcon_3$, but as we will now prove the relationship between the connector polymorphism and $R_3$ is significantly stronger.

\begin{lemma} \label{lemma:not_line_gives_r3}
  Let $\Gamma$ be a constraint language. If $\Gamma$ does not have the connector property then $R_3 \in \fgpp{\Gamma}{\mo \cup \mo'}$.
\end{lemma}
\begin{proof}
    Since $\Gamma$ does not have the connector property, it is not $(\mo \cup \mo')$-preserved by $\pcon$. Hence, Proposition \ref{prop:no_preserve_f} shows that we can $(\mo \cup \mo')$-fgpp-define a binary relation $R$ on the domain $\numdom{3}$ satisfying $\{00,02,11,20,22\} \subseteq R \subseteq \numdom{3}^2 \setminus \{01\}$.
    We now define
    \[
        R_3(x,y) \equiv R(i(x), i(y)) \land R(r(x), r(y)) \land R(i(y), i(x)) \land R(r(y), r(x))
    \]
    where $r,i \colon \numdom{3} \to \numdom{3}$ are, respectively, the order-reversing permutation and identity function over $\numdom{3}$. We note that the tuples $\{00, 02, 11, 20, 22\} \subseteq R_3$, while the four occurrences of $R$ use $\{01\} \notin R$ to respectively exclude the tuples $\{01,21,10,12\}$. We conclude that $R_3$ is indeed the required relation.
    Finally, Propositions \ref{prop:can_compose} and \ref{prop:language_fpt_reduction} show that $R_3 \in \fgpp{\Gamma}{\mo \cup \mo'}$.
\end{proof}

By combining the results proven so far we obtain the following W[1]-hardness condition.

\begin{theorem} \label{thm:monotone-w1}
  Let $\Gamma$ be a domain-reversible constraint language over $\numdom{d}$.
  If $\Gamma$ does not have the connector property then $\mcsp(\Gamma)$ is W[1]-hard.
\end{theorem}
\begin{proof}
First, if $\fcon_d \notin \mpol(\Gamma)$ then Lemma~\ref{lemma:not_line_gives_r3} 
implies that $R_3 \in \fgpp{\Gamma}{\mo \cup \mo'}$. Then we can fgpp-define all permutations (Lemma~\ref{lemma:r3_can_define}) which gives hardness of $\udcsp(\Gamma, \mo \cup \mo')$ since we can define the flip relation and projection relation from Example~\ref{ex:flip}. Last, we use Proposition~\ref{prop:anti_monotone_maps} to obtain an $\FPT$-reduction to $\mcsp(\Gamma)$.
\end{proof}

Next, we investigate additional consequences of the connector polymorphism, beginning with a structural property which all relations invariant under the connector polymorphism must satisfy.

\begin{lemma} \label{lemma:lineprop-biclique}
  Let $R \subseteq \numdom{k}^2$ be preserved by $\fcon_k$,
  and let $a_1,a_2,b_1,b_2 \in D$ be such that $\{a_1,a_2\} \times \{b_1,b_2\} \subseteq R$.
  Then $R \cap ([a_1,a_2] \times [b_1,b_2])$ is a (not necessarily spanning) biclique.
\end{lemma}
\begin{proof}
  Let $p \in [a_1,a_2]$ be such that $(p,q) \in R$ for some $q \in [b_1,b_2]$. 
  Then $a_1 \leq p \leq a_2$ and $b_1 \leq q \leq b_2$, so the connector polymorphism is applicable and the four distinct applications gives us that $(p,b_1), (p,b_2), (a_1,q), (a_2,q) \in R$. In particular, $p, q$ are in
  the same connected component as $a_1,a_2,b_1,b_2$. Let $G$ be the graph
  corresponding to $R \cap ([a_1,a_2] \times [b_1,b_2])$; then $G$ has a
  unique non-trivial connected component. Assume that $G$ contains an induced
  path on 4 vertices $P_4$, i.e., there are values $u \leq u' \in [a_1,a_2]$ and $v \leq v' \in [b_1,b_2]$ such that $|R \cap (\{u,u'\} \times \{v,v'\})|=3$. By the inversion symmetry of the connector polymorphism (it is defined for both monotone and anti-monotone maps from the pattern to $\numdom{k}$), we assume without loss of generality
  that $(u',v') \notin R$. But then the connector polymorphism applies on $u \leq u' \leq a_2$ and $b_1  \leq v \leq v'$ and shows $(u',v') \in R$, contrary to assumption.
  We conclude that the non-trivial component of $G$ is a biclique. 
\end{proof}

We now complete the classification by showing that any language that can define all permutations is not invariant under $\pcon$.

\begin{lemma} \label{lemma:all_permutations}
  Let $\Gamma$ be a constraint language where $\sigma^\bullet \in \fgpp{\Gamma}{\mo}$ for every permutation $\sigma \colon \numdom{n} \to \numdom{n}$ (for every $n \geq 1$). 
Then $\Gamma$ does not have the connector property.
\end{lemma}
\begin{proof}[Proof sketch.]
%  \todo[inline]{This proof can probably be cleaned up a bit, but I think it's OK so maybe not the highest priority.}
  Let $n$ be a sufficiently large constant and consider the flip
  relation on domain $\numdom{n}^2$ from Section~\ref{sec:unrestr-discussion}, i.e., the bijection $\sigma$ that maps $(a,b)$ into $(b,a)$ for arbitrary $(a,b) \in  \numdom{n}^2$.
%  $(0,0), \ldots, (0,n-1), (1,0), \ldots, (n-1,n-1)$ to
%  $(0,0), \ldots, (n-1,0), (0,1), \ldots, (n-1,n-1)$.\todo{PJ: Improve here.}
  By assumption, $\sigma^\bullet \in \fgpp{\Gamma}{\mo}$, so 
  let $F$ be a witnessing fgpp-definition.
  Since each $R \in \Gamma$ has bounded arity and finite domain, each
  constraint $C$ in $F$ divides the domain $\numdom{n}^2$ into $O(1)$ intervals
  such that values in each interval are treated identically by $C$.
  Let $I_1 < \ldots < I_c$ and $I_1' < \ldots < I_d'$ be the partition
  induced by $C$ of the domain of $x$ and $y$, respectively. 
  Assume that we can find $0 \leq a_1<a_2<a_3 \leq c$
  and $0 \leq b_1 < b_2 < b_3 \leq d$ such that $\sigma$ maps some
  element of $I_{a_2}$ to some element of $I_{b_2}'$, and likewise
  maps elements of both $I_{a_1}$ and $I_{a_3}$ to elements of
  both $I_{b_1}'$ and $I_{b_3}'$, and furthermore $C$ excludes a pair
  $(p,q)$ where $p \in I_{a_1}$ and $q \in I_{b_2}'$.
  Then using only the maps corresponding to these intervals,
  we can define a relation on domain $\numdom{3}$ that contains
  $00$, $02$, $20$, $22$, and $11$ but not $01$.

  We claim that since $\Gamma$ is finite, there must be a constraint
  in $F$ which admits such intervals. Consider a constraint $C$ that
  excludes the pair $((n/4,n/2),(n/2,n/2))$ (assuming for simplicity
  that $n$ is a multiple of 4). Then $C$ is equivalent to a binary
  constraint $R(x,y)$ with $R \subseteq \numdom{n}^2$, and as above,
  induces partitions of the domain into intervals $I_1 < \ldots < I_j$ 
  for $x$ and $I_1' < \ldots < I_\ell'$ for $y$, for constants $j, \ell=O(1)$. 
  Let $I_a$ and $I_b'$ be the first intervals of $x$ and $y$, respectively,
  that contain all values $(i,\cdot)$ for some $i$, and similarly
  let $I_c$ and $I_d'$ be the last intervals of $x$ and $y$, respectively, 
  that contain all values $(i,\cdot)$ for some $i$.
  Furthermore let $(n/4,n/2) \in I_p$, $(n/2,n/4) \in I_q'$, and
  $(n/2,n/2) \in I_r \cap I_s'$. Using a high enough value of
  $n=\Omega(j+\ell)$ we can assume $a \leq p \leq r \leq c$
  and $b \leq q \leq s \leq d$. 
  For $v \in \{a,b,c,d\}$ let $i_v \in [n]$ be a value such that
  $I_v$ respectively $I_v'$ contains all values $(i_v,\cdot)$.
  Then $R(x,y)$ contains a 4-cycle $C_4$ on $(I_a, I_c) \times (I_b', I_d')$
  using edges $(i_a,i_b), (i_a,i_d), (i_c,i_b), (i_c,i_d)$ in $X$
  and their flips $(i_b,i_a), \ldots, (i_d,i_c)$ in $Y$. 
  But then Lemma~\ref{lemma:lineprop-biclique} implies that $R$ in
  $(I_a \cup \ldots \cup I_c) \times (I_b' \cup \ldots \cup I_d')$
  is a biclique, and the pairs $((n/4,n/2),(n/2,n/4)) \in (I_p,I_q')$ 
  and $((n/2,n/2),(n/2,n/2)) \in (I_r,I_s')$ contradict the non-existence
  of $((n/4,n/2),(n/2,n/2))$ in $R(x,y)$. 
  %\todo{MW: Not super clean
  %  but at least there was no dragon here?}
\end{proof}

In summary, a domain-reversible language $\Gamma$ thus defines all permutations (meaning that $\mcsp(\Gamma)$ is W[1]-hard) if and only if $\Gamma$ is not invariant under the connector polymorphism. Hence, the connector polymorphism is a necessary condition for tractability which we now explore further.

\subsection{FPT versus W[1] Dichotomy for Binary Relations}
\label{sec:line_fpt}

In this section, we look at the case where the constraint language satisfies the connector property and all relations are \emph{binary}, which means that all relations are of arity 1 and 2 only. We shall prove that this implies that $\mcsp(\Gamma)$ can be solved in FPT time. The main result is the following theorem.

\begin{theorem}\label{thm:twinwidth}
Let $\Gamma$ denote a binary constraint language.
If $\Gamma$ satisfies the connector property, then $\mcsp(\Gamma)$ 
can be solved in FPT time.
\end{theorem}

To prove the result, we make a detour using the concepts of grid-rank and twin-width. We will show that the connector property implies bounded grid-rank, then use results from literature~\cite{bonnet2024twin} to show that bounded grid-rank implies bounded twin-width, and finally give a dynamic programming algorithm based on bounded twin-width. We begin with the definitions of grid-rank and twin-width.

\begin{definition}[Bonnet et al.~\cite{bonnet2024twin}] \label{def:grid-rank}
    Let $M$ be a matrix. A $d$-{\em division} of $M$ is a partition of its entries into $d^2$ {\em zones} using $d-1$ vertical lines and $d-1$ horizontal lines. A {\em rank-$k$ $d$-division} is a $d$-division where each zone  has at least $k$ non-identical rows or at least $k$ non-identical columns.
    We say that $M$ has \emph{grid-rank} at least $k$ if and only if there exists a subdivision of the rows and columns into a $k \times k$ grid where each submatrix corresponding to a cell of this grid has rank at least $k$. 
\end{definition}

\begin{definition}[Bonnet et al.~\cite{Bonnet:etal:jacm2022}]
  An {\em ordered directed graph} is a directed graph $(V, E)$ equipped with a total order $(V, \leq)$.
  Let $G$ be an an ordered directed graph with labeled edges. Its \emph{twin-width} is defined as follows.
  \begin{enumerate}
    \item
  A \emph{contraction} is an operation where we merge two (not necessarily adjacent) vertices $u,v$ into one vertex $w$. For any other vertex $x$, if the arcs between $x$ and $u$ and the arcs between $x$ and $v$ share the same set of labels, then we also add arcs with these labels between $x$ and $w$. Otherwise, we only add an (undirected) red error edge between $x$ and $w$.
\item
  A \emph{contraction sequence} is a sequence of such contractions, starting with the initial graph and ending with a single vertex.
\item
  The \emph{width} of a contraction sequence is the maximal number of red error edges adjacent to a single vertex at any point during the contraction sequence.
\item
  Finally, the \emph{twin-width} of the graph is the minimal width over all possible contraction sequences.
  \end{enumerate}
\end{definition}

Our proof uses the following graph representation of binary $\mcsp(\Gamma)$ instances.

\begin{definition}
    The {\em \thegraph} \emph{graph} of a binary $\mcsp(\Gamma)$ instance is an ordered directed graph with labeled edges defined as follows.
    For each member of the domain $[n]$, we construct a vertex, ordered in the same way as in $[n]$. We use $k^2$ labels corresponding to pairs of variables $(v,w)$. For any pair of vertices $(i,j)$, we add an arc labelled $(v,w)$ from $i$ to $j$ if and only if setting $v$ to $i$ and $w$ to $j$ does not invalidate any of the constraints between $v$ and $w$.
\end{definition}

We now begin with the proof of Theorem \ref{thm:twinwidth}. We first show that the connector property implies bounded grid-rank.

\begin{lemma}\label{lem:line_gridrank}
    Let $\Gamma$ be a binary constraint language with the connector property. Let $(n,V,C)$ be an instance of $\mcsp(\Gamma)$ and let $G$ be its \thegraph\ graph. Now for any pair of variables $(v,w)$, consider the matrix $M$ defined as the adjacency matrix of $G$ only containing arcs with the label $(v,w)$. Then, the grid-rank of $M$ is only dependent on $\Gamma$.
\end{lemma}
\begin{proof}
     Suppose that the base language $\Gamma$ has the connector property, and we are given an instance $(n, V, C)$ of $\mcsp(\Gamma)$ where $|V| = k$. Let $d$ be the domain size of $\Gamma$. We will show that, for any label, the grid-rank of the adjacency matrix is at most $2d$. Suppose, to the contrary, that the grid-rank of some label $(v,w)$ is at least $2d+1$. Then there exists a subdivision of the matrix into a $(2d+1) \times (2d+1)$ grid where each cell has rank at least $2d+1$.

    Recall that this adjacency matrix $A$ has a 1 at coordinate $(i,j)$ if and only if setting $v$ to $i$ and $w$ to $j$ does not invalidate any constraints. Now let $A_R$ be the adjacency matrix with respect to any individual constraint $R$, where we set a 1 at coordinate $(i,j)$ if and only if setting $v$ to $i$ and $w$ to $j$ does not invalidate $R$. Now, $A$ is the intersection of $A_R$ over all constraints $R$.

    Each individual constraint $R$ is defined using two guarding functions into the base domain $\Gamma$, one for each variable, and thus one that corresponds to the rows of $A_R$ (say $f_{R,r}$)and one corresponding to the columns of $A_R$ (say $f_{R,c}$). Since each guarding function is monotone, and $\Gamma$ contains only $d$ domain values, $f_{R,r}$ can only change value in at most $d-1$ places. Now consider the first $d$ rows of the given $(2d+1) \times (2d+1)$ grid. Since $f_{R,r}$ changes value in at most $d-1$ places, there must be some grid row in which $f_{R,r}$ does not change value and is constant. Let $r_{R,\min}$ be this row. Similarly, there must be a row $r_{R,\max}$ among the last $d$ rows of the grid in which $f_{R,r}$ is constant. Additionally, there are columns $c_{R,\min}$ and $c_{R,\max}$ among the respectively first and last $d$ columns where $f_{R,c}$ is constant.

    Now consider the four grid cells at the intersection of these two rows and two columns. In each of these cells, both guarding functions are constant, so $A_R$ is constant as well. Since this cell has rank at least $2d+1$ in $A$, there must be a nonzero value in the cell. Since $A$ is the intersection of $A_R$ over all $R$, we find that this value is a one in $A_R$ as well, hence $A_R$ must be constant one in the whole cell.

    We now apply Lemma \ref{lemma:lineprop-biclique}, with $a_1,a_2,b_1,b_2$ being any domain value from $\numdom{n}$ whose corresponding row or column in $A_R$ is contained in the grid row or column $r_{R,\min},c_{R,\min},r_{R,\max},c_{R,\max}$, respectively. Since the four grid cells at the intersections are constant one, we have that $\{a_1,a_2\} \times \{b_1,b_2\} \in R$, so we can apply the lemma. This shows that $R$ restricted to the ranges $[a_1,a_2]$ and $[b_1,b_2]$ is a (not necessarily spanning) biclique. In terms of $A_R$, this means that the submatrix of $A_R$ between these rows and columns is rectangular and has rank 1. In particular, the middle cell of the $(2d+1) \times (2d+1)$ grid has rank 1.

    Note that this holds for any relation $R$. Since the middle cell of the grid is the intersection of $A_R$ over all $R$, and thus the intersection of several rank 1 matrices, it must have rank 1 in $A$ too. Since $d$ is at least one, this is a contradiction with the assumption that each cell in the grid has rank at least $2d+1$. We conclude that the grid-rank of $A_R$ is at most $2d$.
\end{proof}

We continue with an algorithm based on twin-width.

\begin{lemma}\label{lem:twinwidth_fpt}
    Let $\Gamma$ be a binary constraint language. Let $(n,V,C)$ be an instance of $\mcsp(\Gamma)$ where $|V| = k$ and \thegraph\ graph $G$. Suppose we are given a contraction sequence of $G$ of width $t$. Then, we can solve this instance in FPT time when parameterized by both $k$ and $t$.
\end{lemma}
\begin{proof}
    We describe a dynamic programming algorithm over the given contraction sequence. During this contraction sequence, we use the following terminology: a \emph{variable} is one of the $k$ variables from the original CSP with unbounded domain. A \emph{vertex} is one the $n$ vertices from the graph before the contractions, or equivalently, an element from the domain $\numdom{n}$. A \emph{bag} is a vertex from the graph during or after the contraction process. Note that a bag always corresponds to some subset of vertices.

    We first describe our states. At each point during the contraction, we look at all induced subgraphs that are connected with respect to error edges and have size at most $k$. We then look at each subset of variables. For each subgraph and subset, we look at all partitions of the variables from the subset over the bags from the subgraph where each bag is assigned at least one variable. These combinations of subgraph, subset, and partition are our states.

    For each state, we wish to solve the following subproblem: does there exist a partial solution to the unbounded-domain CSP where each variable from the subset is assigned a value from the corresponding bag and where all other variables are left unassigned? Observe that, after performing all contractions, there is only one bag left. Hence, the final subproblem consisting of this bag as subgraph, the full subset of all induced variables, and the trivial partition that assigns all variables to this bag asks whether there is a solution to the original unbounded-domain CSP.

    We compute this information as follows. In the base case, before contracting anything, there are no error edges so we only consider subgraphs consisting of a single bag which in turn corresponds to single vertex. For each subset of variables, we can trivially check in FPT time whether assigning these variables to the single vertex is a valid partial solution.

    Now consider a later stage from the contraction sequence. Let $B$ be the latest bag that was created, and let $B_1$ and $B_2$ be the two bags that were merged to create $B$. Consider a subproblem from this stage, consisting of an induced subgraph $G'$, a subset $S$, and a partition $P$. If $B$ is not contained in $G'$, then this subproblem was also a subproblem in the previous stage of the contraction sequence, so we can simply copy the result. Otherwise, let $S_B$ be the (nonempty) subset of variables which is assigned to $B$ in this partition. Recall that any valid solution for this subproblem must assign to each variable from $S_B$ a value from $B$. We now consider all $2^{|S_B|}$ options on whether each variable was assigned a value from $B_1$ or from $B_2$. Since $|S_B| \leq k$, there are FPT many possibilities.

    Consider a partition of $S_B$ over $B_1$ and $B_2$. Let $G''$ be the same induced subgraph as $G'$ but before the contraction of $B_1$ and $B_2$ into $B$. As an edge case, we ensure that each bag from $G''$ is assigned at least one variable by removing $B_1$ or $B_2$ if it is not assigned any variable. Note that $G''$ is not necessarily connected with error edges; it may consist of multiple connected components. We now claim that the subproblem has a solution with respect to this partition if and only if the following holds:

    \begin{itemize}
        \item For each connected component of $G''$, the subproblem restricted to this component has a valid solution.
        \item For each pair of variables $(v,w)$ from $S$ which are assigned to different components in $G''$, there is an arc from the bag assigned to $v$ to the bag assigned to $w$ with the label $(v,w)$.
    \end{itemize}

    For the first direction, suppose that these two properties hold. We now claim that the union of the partial assignments of each connected component is a valid assignment for the initial subproblem. Indeed: consider a pair of variable assignments $v$ set to $i$ and $w$ set to $j$. If $v$ and $w$ belong to the same connected component, then this assignment does not invalidate any relation because the subproblem restricted to this component was valid. If $v$ and $w$ belong to different components, then there is an arc with label $(v,w)$ between their assigned bags, which means that any assignment of $v$ and $w$ to values from these bags does not invalidate any relation between $v$ and $w$.

    For the second direction, suppose that there is a valid solution to the subproblem. For each connected component of $G''$, consider the subassignment containing only the variables assigned to bags from this connected component. This is obviously still a valid partial solution, so the first property holds. Now consider a pair of variables $(v,w)$ assigned to different components, and look at their respective bags. There is no error edge between these bags; if so, $v$ and $w$ would not belong to different components. There also has to be an arc in the color $(v,w)$; if not, then the solution to the initial subproblem would be invalid. This shows that the second property also holds. This completes the second direction and thus the claim.

    Using this claim, we can now compute the solution to each subproblem in $f(k)$ time: for each of the $2^{|S_B|}$ partitions, we simply check if both properties hold, taking at most $O(k^2)$ time per partition. The subproblem is then valid if and only if both properties hold for at least one partition. This almost completes the proof of the dynamic programming algorithm: the only thing that remains is to show that there are only FPT many states, which we will now show.

    There are $n-1$ contractions. After each contraction, consider an induced subgraph that is connected with error edges. This induced subgraph can be obtained by a random walk starting from an arbitrary vertex in at most $2k$ steps: consider a spanning tree and ``walk around'' it. Since the contraction sequence has twin-width $t$, there are at most $t^{2k}$ random walks of this length, and hence at most $t^{2k}$ subgraphs containing a fixed starting vertex, and thus at most $t^{2k}n$ subgraphs in total. Since there are at most $2^k$ subsets of variables, and at most $k^k$ partitions of this subset of the induced subgraph, there are at most $(n-1) \cdot t^{2k}n \cdot 2^k \cdot k^k = (2kt^2)^k \cdot n(n-1)$ states, which is FPT when parameterized by both $k$ and $t$. This completes the proof.

    As a final note, observe that the dependence on $n$ can be made linear if, during the contraction sequence, we only update states whose induced subgraph contains the newly contracted bag.
\end{proof}

We now return to the proof of the main theorem.
\begin{proof}[Proof of Theorem \ref{thm:twinwidth}]
    Let $(n,V,C)$ be an instance of $\mcsp(\Gamma)$ and let $G$ be its \thegraph\ graph. Since $\Gamma$ has the connector property, Lemma \ref{lem:line_gridrank} shows that for each label $(v,w)$, the adjacency matrix of $G$ concerning arcs with this label has bounded grid-rank. We now apply an existing result~\cite{bonnet2024twin} (see, also, Theorem 5.1 in Geniet~\cite{genie2024}) which uses these bounded grid-ranks to (in polynomial time) construct a contraction sequence whose width is only dependent on the largest grid-rank and the number of labels $k^2$. We then use the algorithm from Lemma~\ref{lem:twinwidth_fpt}, noting that it runs in FPT time parameterized by just $k$ since the twin-width is only dependent on $\Gamma$ and $k$.
\end{proof}

Combined with our earlier results, we obtain the following dichotomy for $\mcsp(\Gamma)$ with binary relations.

\begin{theorem} \label{thm:fpt-vs-w1-mono}
Let $\Gamma$ denote a binary constraint language.
If $\Gamma$ has the connector property, then $\mcsp(\Gamma)$ 
can be solved in FPT time.
Otherwise, $\mcsp(\Gamma)$ is W[1]-hard.
\end{theorem}
\begin{proof}
If $\Gamma$ has the connector property, then the result follows from 
Theorem~\ref{thm:twinwidth}. Assume that $\Gamma$ does not satisfy the connector property.
If $\Gamma$ is $\mo$-preserved by $\pmin$ or $\pmax$, then $\mcsp(\Gamma)$
is in P by Lemma~\ref{lem:minandmaxisinP} so we assume that this is not the case.
Lemma~\ref{lemma:can_define_reverse} thus allows us to assume that
$\Gamma$ is domain-reversible and Theorem~\ref{thm:monotone-w1} implies that $\mcsp(\Gamma)$ is W[1]-hard.
\end{proof}

\subsection{Application: Boolean MinCSP}
\label{sec:boolean_mincsp}

%\todo[inline]{PJ: Some kind of opening needed.}
Before continuing with higher-arity generalizations we demonstrate an application of Theorem~\ref{thm:fpt-vs-w1-mono} to replace a step in the FPT algorithm for \textsc{Boolean MinCSP}.
Recall that the \textsc{MinCSP$(\Gamma)$} for a set of relations $\Gamma$ asks if there is an assignment that satisfies all but at most $k$ constraint in a given CSP$(\Gamma)$ instance.

Kim et al.~\cite{KimKPW23fa3} show that for any bijunctive constraint
language $\Gamma$ (i.e., every relation $R \in \Gamma$ can be defined
through a 2-CNF formula) such that for every $R \in \Gamma$, the
\emph{Gaifman graph} of $R$ is $2K_2$-free, the problem \textsc{MinCSP($\Gamma$)}
is FPT parameterized by the number of false constraints $k$. 
The Gaifman graph of a Boolean relation $R \subseteq \{0,1\}^r$
is an undirected graph $G_R$ on vertex set $[r]$ where there is an edge ${i, j} \in E(G_R)$ if
and only if the projection of $R$ onto arguments $i$ and $j$ is non-trivial, i.e., if and only if
there exist values $b_i,b_j \in \{0,1\}$ such that for
every $t \in R$ it is not the case that both $t[i] = b_i$ and $t[j] = b_j$.

The
algorithm works in three phases: (1) reduction of the MinCSP to a graph cut problem,
(2)
reduction of the graph cut problem to an instance where the solution
is an $st$-min cut, and (3) and a complex branching process to solve such min-cut
instances. Of these steps, the reduction to a min-cut instance is a
relatively complex application of the tool of \emph{flow augmentation}~\cite{KimKPW25fa1},
which is the driving force of the FPT algorithm, yet the step that is
arguably the most technically intricate is the final complex branching
step of solving a min-cut instance (Section~3.3 of~\cite{KimKPW23fa3}). 
We illustrate how some of this task can be taken over by the $\mcsp$ framework.

Let us review the setup in more detail. 
The first step use iterative compression to
reduced to a simpler MinCSP, phrased in graph terminology as follows.
An instance of \textsc{Generalized Digraph Pair Cut} (GDPC) 
is a tuple $I=(G, s, t, \cC,\cB, \weight, k, W)$ with the following
meaning. 
\begin{itemize}
\item $G$ is a digraph with distinguished vertices $s, t \in V(G)$
  corresponding to crisp assignments $s=1$ and $t=0$,
\item $\cC \subseteq \binom{V(G)}{2}$ is a set of ``clauses'',
  where a clause $C=\{u,v\}$ corresponds to a binary constraint $(\neg u \lor \neg v)$,
\item $\cB$ is a collection of pairwise disjoint subsets of $E(G) \cup \cC$
  referred to as \emph{bundles},
\item $\weight \colon \cB \to \Z_+$ is the weight of each bundle and
  $W \in \Z_+$ is the maximum total weight, and
\item $k \in \N$ is the parameter.
\end{itemize}
Each bundle $B \in \cB$ corresponds to one constraint in the
initial \textsc{MinCSP} instance. A solution to $I$ is
a set of edges $Z \subseteq E(G)$ that contains an $st$-cut;
an edge $uv \in E(G)$ is violated by $Z$ if $uv \in Z$, and
a clause $C \in \cC$ is violated by $Z$ if both parts of $C$ are
reachable from $s$ in $G-Z$. The \emph{cost} of $Z$ is the number of
bundles $B \in \cB$ that contain a violated clause or edge
and the \emph{weight} of $Z$ is the total weight of these bundles.
The bundles $B \in \cB$ are $2K_2$-free in the natural sense that the
undirected graph formed by the union of edges and clauses of $B$ is
$2K_2$-free. We also assume that there is a bound $b$ on the number of
vertices involved in each bundle. 

In the second step, colour-coding and repeated application of flow
augmentation are used to further impose the condition that $Z$ is a
minimum $st$-cut in $G$. $I$ is then said to be a \emph{mincut instance}. 
We remark that this reduction uses the weight $W$ in an algorithmic way,
i.e., even when we are only interested in the unit-weight version
where we originally have $\weight \equiv 1$ and $W=k$, after the reduction
we have a mincut instance with a non-trivial weight function $\weight'$.

We show how a restricted form of mincut instance can be solved using
$\mcsp$ over a binary language on domain $\numdom{3}$, thereby bypassing the
intricate branching procedures for such instances. We remark that
with some additional case analysis, it is possible to cover further cases,
but the work required to do so arguably detracts from this result's
role as illustration.
In order to do so, we first show that
the algorithm for binary $\mcsp(\Gamma)$ (Theorem~\ref{thm:twinwidth}) can be generalized into a suitable weighted version.

\begin{corollary}
    Consider the weighted version of binary $\mcsp(\Gamma)$ where each domain value has a positive weight and we wish to compute a solution with minimal weight. This can be solved in FPT time when parameterized by both $k$ (number of variables) and $t$ (twin-width of the \thegraph\ graph).
\end{corollary}
\begin{proof}
    We slightly modify the algorithm from Lemma \ref{lem:twinwidth_fpt}. For every state, we now keep track of the smallest weight required to solve the corresponding subproblem. We solve each subproblem by determining the smallest total weight over all partitions of $S_B$ into $B_1$ and $B_2$.
\end{proof}

In the following, the first assumption is the
condition for GDPC being FPT; the second is the outcome of phase two
of the algorithm; and the third is our additional simplifying
assumption.

\begin{lemma}
  Let $I$ be an instance of GDPC with parameter $K$ and bundles of
  arity $b$, and with the following properties.
  \begin{enumerate}
  \item Every bundle is $2K_2$-free
  \item The instance has a solution where the cut edges $Z$ is an
    $st$-min cut
  \item In every bundle violated by $Z$, all vertices are incident to
    violated arcs within the bundle.
  \end{enumerate}
  Then $I$ reduces to an instance $I_F$ of $\mcsp(\Delta)$ with $k$
  variables, where $\Delta$ contains binary relations over domain
  $\numdom{3}$ with the connector property. More precisely, if $I$ is a yes-instance
  then $I_F$ is a yes-instance with probability $1/2^{O(b^2k \log k)}$,
  and if $I_F$ is a yes-instance then so is $I$. 
\end{lemma}
\begin{proof}
  Let $I$ be a mincut instance of GDPC and let $\cP = \{P_1,\ldots,P_\lambda\}$
  be an $st$-max flow in $G$. As pointed out in Section~\ref{sec:applications}, we can formulate
  the min-cut condition as a $\mcsp$ instance with $\lambda$ variables
  over the Boolean language $\{\mathsf{Impl},0\}$, and we can embed $\mathsf{Impl}$ and $0$
  into the domain $\numdom{3}$, implementing $\mathsf{Impl}$ as $\{11, 12, 22\}$ and $0$ as $\{1\}$. 

  Create an instance $I'$ of $\mcsp$ by starting with this mincut-encoding.
  We may assume that $\lambda = O(b^2k)$ as otherwise $Z$ must intersect more than $k$ bundles.
  Denote the bundles that are violated by $Z$ as $\{B_1,\ldots,B_k\} \subseteq \cB$,
  and guess the partition $\cP=\cP_1 \cup \ldots \cup \cP_k$ where $\cP_i$
  contains precisely those paths $P \in \cP$ that intersect $Z$ in the
  bundle $B_i$. We have a one-in-$k^{O(b^2k)}$ chance of guessing the correct partition. 
  Additionally, we guess a map $\chi \colon \cB \to [k]$, with the
  intention that the bundle $B_i$ gets $\chi(B_i)=i$, and for each
  bundle $B \in \cB$ we guess an assignment $\varphi \colon V(B) \to \{0,1\}$,
  with the intention that for every $i \in [k]$, $\varphi^{-1}(1)$ is
  precisely the set of vertices of $B_i$ reachable from $s$ in $G-Z$.
  The probability of these guesses being accurate is $1/k^k$ and
  $1/2^{O(b^2k)}$, respectively. We refer to clauses that are violated
  in $\varphi$ as \emph{soft}, and other clauses as \emph{crisp}.

  First we show how to handle crisp clauses, i.e., clauses known not
  to be violated by $Z$, such as clauses not in bundles and clauses
  satisfied by the assignment $\varphi$. 
  We use some terminology from~\cite{KimKPW23fa3}.
  Let $u \in V(\cP)$ and $v \in V(G)$ be vertices. An \emph{attachment path}
  from $u$ to $v$ is a path from $u$ to $v$ that is disjoint from
  $V(\cP)$ except at $u$. For a vertex $v$ and a path $P_i \in \cP$,
  the \emph{projection} of $v$ to $P_i$ is the earliest vertex $u$ of
  $P_i$ such that there is an attachment path from $u$ to $v$,
  if such a vertex exists. We say that \emph{$v$ projects to $P_i$}, if so,
  and let $\iota(v,i)=j$ be the value of the variable $x_i$ that
  represents the arc whose tail is this projection. 
  Now, for every crisp clause $\{u,v\}$, and every pair $(i,j) \in
  [k]^2$ such that $u$ projects to $P_i$ and $v$ projects to $P_j$,
  we add the constraint $\mathsf{Or}_2([x_i < \iota(u,i)], [x_i < \iota(v,j)])$.
  Note that this can be done within $\mcsp$ using a negative 2-clause
  from $\Delta$. 
  
  We now process the bundles as follows. For any bundle $B$ which is
  incompatible with being the $\chi(B)$:th violated bundle, due to the
  partition of $\cP$ and/or the assignment $\varphi(B)$, we
  re-interpret all clauses of $B$ as crisp, and handle them as above.
  Furthermore, for any arc $uv$ in $B$, representing some value
  $x_i=a$, we add a constraint $Impl([x_i \geq a],[x_i \geq a+1])$ 
  to prevent the arc from being cut. What remains to enforce is that
  (1) for every bundle $B$ not broken up this way, and any two soft
  arcs $e, e' \in B$, the solution cuts $e$ if and only if it cuts $e'$,
  and (2) for any bundle $B$, say with $\chi(B)=i$, if another bundle
  $B'$ with $\chi(B')=i$ is violated, then the solution satisfies all
  soft clauses of $B$. We can handle this using the example relations 
  for constrained min-cuts above. Let $B$ be a bundle and let
  $e, e' \in B$ be soft arcs. Let $e \in Z$ correspond to $X_i=p$
  and $e' \in Z$ correspond to $X_j=q$, where $i \neq j$ by assumption.
  Since $B$ is $2K_2$-free, either $|V(e) \cup V(e')<4$, or there is a
  clause in $B$ on $V(e) \times V(e')$, or there is another arc of $B$,
  soft or crisp, between $V(e)$ and $V(e')$. In the first case, say
  $e=ab$ and $e'=cd$ with not all vertices distinct. Since $e, e'$ are soft,
  $\varphi(a)=\varphi(c)=1$ and $\varphi(b)=\varphi(d)=0$. Thus either
  $a=c$ or $b=d$ or both, and in all cases the coordination of $e$ and
  $e'$ is handled by the solution being an $st$-min cut.
  In the second case, we use the relation $R_1=\{00,02,20,11\}$
  and threshold maps $f_i$ to impose
  \[
    R_1(f_p(x_i), f_q(x_j))  \quad \text{ where } \quad
  f_i(x) = 
  \begin{cases}
    0 & x < i \\
    1 & x = i \\
    2 & x>i.
  \end{cases}
  \]
  Then $x_i=p$ if and only if $x_j=q$, so the edges are coordinated,
  and every other behaviour on $P_i$ and $P_j$ is accepted except
  cutting both paths after $p$ and $q$, since the latter would violate
  the clause on $V(e) \times V(e')$. 
  The third case, say with an arc $V(e) \to V(e')$, is handled
  similarly except using $R_2=\{00, 02, 11, 22\}$; in
  addition to the coordination $X_i=p$ iff $X_j=q$,
  $R_2$ excludes $X_i > p$ and $X_j < q$, which cannot occur in a
  bundled min-cut: if $X_i>p$ then no arc of $B$ is cut by $Z$, 
  so $B$ contains an uncut path from the $s$-side of $Z$ to the $t$-side.

  For a sketch of the correctness proof, assume first that all guesses
  made were consistent with $Z$. This happens with total probability
  $1/2^{O(b^2k \log k)}$. Then the assignment where $x_i=p$ if and
  only if the $p$:th arc of $P_i$ is contained in $Z$ satisfies all
  constraints we have imposed, and the output instance is satisfiable.
  On the other hand, any satisfying assignment to the output instance
  will violate at most $k$ bundles: due to the cut-coordination constraints,
  any cut $Z'$ produced contains all or none of the soft arcs for
  every bundle, and the only situation where a soft clause is violated
  is when it is contained in a violated bundle. 
\end{proof}

\section{The Case of Non-binary Base Languages}
\label{sec:non_binary}

We now investigate the issue of $\mcsp(\Gamma)$ when $\Gamma$ is a
non-binary base language with the connector polymorphism.
We first observe that if all unbounded-domain constraints in an instance of $\mcsp(\Gamma)$
are binary, then Theorem~\ref{thm:twinwidth}
still applies, even if the base language is non-binary.
This is in the same spirit as how the \textsc{Boolean MinCSP} 
example in Section~\ref{sec:boolean_mincsp} reduces non-binary Boolean
relations to non-Boolean binary relations. We formalize this
observation in Section~\ref{sec:monotone-binary-output}.

However, if we are interested in instances of $\mcsp(\Gamma)$ with
non-binary unbounded-domain constraints, then we need to move beyond the
notion of twin-width, since twin-width is intrinsically a property of
binary structures. Furthermore, unlike width notions such as
rank-width or treewidth, there is no immediate and generic way to
translate an $r$-ary relation into a combination of binary relations
that works well with bounded twin-width.

For example, one method of reducing hypergraphs into graphs
is to replace every edge $E$ in a hypergraph $H=(V,\EE)$
by a clique on the vertex set $E$. If the edge size is bounded
(e.g., $|E| \leq r$ for all $E \in \EE$) and $H$ is sparse,
then the resulting graph could have bounded tree-width,
which would imply tractability.
% \todo{Cite Grohe; the CSP is tractable.}
However, our relations can be very dense. For example, the 3-ary
Boolean relation \textsf{Even}$_3(x,y,z) \equiv (x+y+z=0 \mod 2)$
has the connector polymorphism, but a simple unbounded-domain constraint
such as $\mathsf{Even}_3([x \leq n/2], [y \leq n/2], [z \leq n/2])$
would in this representation yield a graph which is a clique.

Another option for a CSP representation is based on constraints, as follows.
Consider a representation with one vertex for every pair $(C,\alpha)$
where $C$ is a constraint $R(x_1,\ldots,x_r)  \equiv R'(f_1(x_1),\ldots,f_r(x_r))$,
and $\alpha \in R$ is a tuple accepted by the constraint.
Add an edge between two vertices $(C,\alpha)$ and $(C',\beta)$
if $\alpha$ and $\beta$ agree on the value of any variables shared
between $C$ and $C'$. Enumerate the values $\alpha$ in, for example,
lexicographical order. But now trivial constraints give graphs of
unbounded twin-width: Consider two binary constraints
$R_1(x,y)$ and $R_2(y,z)$, and assume for simplicity that $R_1$
and $R_2$ are just free relations $\numdom{n}^2$ (the same conclusion
arrives if we only assume that they contain large bicliques). 
Then the edges between $(R_1,\cdot)$ and $(R_2,\cdot)$
almost encode the flip relation, 
and as a relation over the ordered domains $\numdom{n}^2$
it has unbounded twin-width.

%\todo[inline]{If I can: Say something about the status of
%  investigations into non-binary logic structures with bounded
%  expressive power (e.g.\ NIP).}

One potential direction to fix this would be to replace the role of
twin-width by a more restrictive and well-investigated width notion
such as rank-width, in which working with non-binary structures
could be less challenging. Unfortunately, we show in
Section~\ref{sec:unbounded-rankwidth} that this is not feasible. 
Another, more challenging approach, is to investigate a width notion for
general bounded-arity structures, that specialises to twin-width for
binary structures. We pursue this in Section~\ref{sec:projected-gridrank},
where we define the notion of \emph{projected grid-rank}, guided by
the restrictions imposed by the connector polymorphism.

As strong support for Conjecture~\ref{conjecture},
we show two results. First, for a finite language $\Gamma$,
$\fgpp{\Gamma}{\mo}$ has bounded projected grid-rank if and only if
$\Gamma$ has the connector property (see Lemma~\ref{lemma:projected-grid-rank}).
Second, and more impactful, we show in Section~\ref{sec:gridrankramsey} that projections
of relations with bounded projected grid-rank have bounded grid-rank.
That is, for any $k$-variable instance $I$ of $\mcsp(\Gamma)$ over a
language $\Gamma$ with the connector property, and any two variables $x_i, x_j$
in $I$, the binary relation resulting from projecting all satisfying
assignments of $I$ to $x_i$ and $x_j$ has twin-width bounded by $f(k)$.
This makes it highly unlikely that a W[1]-hardness proof would be
possible for such a language.
The proof of this result is a Ramsey argument, relating projected
grid-rank to the standard notion of grid-rank. 

\subsection{Binary Output Relations} \label{sec:monotone-binary-output}

We first confirm that the problem is FPT if restricted to binary unbounded-domain relations.

\begin{lemma} \label{lemma:to-binary}
  Let $\Gamma$ be a constraint language.
  There is a language $\Gamma'$ containing only binary relations
  such that for every binary relation $R \subseteq \numdom{n}^2$
  for some $n \in \N$, $R \in \fgpp{\Gamma}{\mo}$ if and only if
  $R \in \fgpp{\Gamma'}{\mo}$. 
\end{lemma}
\begin{proof}
  Let $R \in \Gamma$ be a constraint of arity $r \in \N$ and domain $\numdom{d}$, and consider
  all the individual binary unbounded-domain constraints $R'(x_1,x_2)$ one can form
  with $R$ over some domain \numdom{n}. These are structured as follows:
  \begin{enumerate}
  \item A map $\iota \colon [r] \to [2]$ that defines whether variable
    $x_1$ or $x_2$ is used in each position $i \in [r]$.
  \item A map $f_i \colon \numdom{n} \to \numdom{d}$ for every $i \in [r]$.
  \end{enumerate}
  There are finitely many options for $\iota$. For the second aspect,
  even though there is an unbounded number of different
  functions $f_i$, in a single constraint $R'(x_1,x_2)$ there are only
  $O(dr)$ distinct behaviours of values of $\numdom{n}$. That is,
  define equivalence classes $\equiv_i$, $i=1,2$, where
  $a \equiv_1 b$ if for every $c \in \numdom{n}$ we have
  $(a,c) \in R'$ if and only if $(b,c) \in R'$ (and similarly for
  $\equiv_2$ over the second position). Then $\equiv_1$ and $\equiv_2$
  both have at most $dr$ equivalence classes each: for each $i \in [r]$
  the function $f_i$ has at most $d-1$ ``breakpoints'',
  where $f_i(j) \neq f_i(j+1)$, and if two values $a, b \in \numdom{n}$
  lie between these breakpoints for every function $f_i$ in the definition,
  then they are equivalent. Thus, every such relation $R'(x_1,x_2)$
  can be represented over the domain $\numdom{dr}$ by keeping only one
  value per equivalence class (and maintaining the domain order).
  Thus we let $\Gamma'$ be the language containing all binary
  relations over the domain $\numdom{dr}$ in $\fgpp{\Gamma}{\mo}$.
  Then, for every constraint
  $R'(x_1,x_2) = R(f_1(x_{\iota(1)}), \ldots, f_r(x_{\iota(r)}))$
  over a domain $\numdom{n}$, 
  there are functions $g_i \colon \numdom{n} \to \numdom{dr}$
  and a relation $R'' \in \Gamma'$
  such that $R'(x_1,x_2) = R''(g_1(x_1), g_2(x_2))$.  
  Here, $g_i$ are simply the injective monotone functions on the
  equivalence classes of $R'$. Conversely, for every relation
  $R \in \fgpp{\Gamma'}{\mo}$ we have $R \in \fgpp{\Gamma}{\mo}$
  by Proposition~\ref{prop:can_compose}.
\end{proof}

Combining this with Theorem~\ref{thm:twinwidth} we get the promised result.

\begin{corollary}
  Let $\Gamma$ be a base language with the connector polymorphism. Then
  $\mcsp(\Gamma)$ restricted to instances where all unbounded-domain
  constraints are binary is FPT parameterized by the number of variables.
\end{corollary}
\begin{proof}
  Let $d$ and $r$ be the maximum domain size and arity of
  a relation in $\Gamma$. Let $\Gamma'$ over domain $\numdom{dr}$
  be as in Lemma~\ref{lemma:to-binary}. We claim that we can reduce an
  instance of $\mcsp(\Gamma)$ with only binary unbounded-domain constraints
  to an instance of $\mcsp(\Gamma')$ in polynomial time. Consider a
  constraint $R'(x_1,x_2) = R(f_1(x_{\iota(1)}), \ldots, f_r(x_{\iota(r)}))$
  for some $\iota \colon [r] \to [2]$ and guarding functions $f_i$.
  We can enumerate the domain breakpoints for $x_1$ and $x_2$ in
  linear time, and exhaustively enumerate the resulting binary
  relation $R''$ over domain $\numdom{dr}$, since there are only
  finitely many possibilities to check. This gives the relation
  $R'' \in \Gamma'$ that we replace $R$ by in the output instance.
  The number of variables is preserved and the reduction takes
  polynomial time. Furthermore, $\Gamma'$ has the connector polymorphism by
  Theorem~\ref{thm:galois-wrapup}, thus Theorem~\ref{thm:twinwidth}
  applies. 
\end{proof}

\subsection{Rank-width is Unbounded} \label{sec:unbounded-rankwidth}

We now investigate whether the role of twin-width in Theorem~\ref{thm:twinwidth}
can be replaced by a more classical width measure, such as rank-width.
Unfortunately, we find that the answer is negative.

Let us define the notions. \emph{Rank-width} and \emph{clique-width}
are width notions that capture graphs that are potentially dense, yet
have a regular structure; they can be seen as a generalization of
graphs of bounded tree-width to also include dense graphs.  
Clique-with and rank-width are functionally equivalent in that a graph
has bounded clique-width if and only if it has bounded rank-width,
although rank-width is quantitatively more powerful and appears easier
to compute; see Oum~\cite{Oum17rankwidth} for an overview. 
%the inequality $r \leq c \leq 2^{r+1}-1$ holds
%where $k$ is rankwidth and $c$ is cliquewidth~\cite[Proposition~6.3]{Oum:Seymour:jctb2006}.

Rank-width is less powerful than twin-width (i.e., there are graph classes with
bounded twin-width and unbounded rank-width), but is better understood algorithmically.
The notion has also been generalized to structures other than graphs,
including matroids and hypergraphs~\cite{JeongKO21}.

Rank-width is defined as follows. 
Let $G$ be a graph on vertex set $V$. A \emph{branch decomposition of $V$} is a
pair $(T,\delta)$ where $T$ is a subcubic tree and $\delta$ is a
bijection between $V$ and the leaves of $T$. For each edge $e$ in $T$,
let $V=A_e \cup \overline{A_e}$ be the partition of $V$ induced by the
components of $T-e$, where the choice of $A_e$ versus $\overline{A_e}$ is arbitrary.
Let $f \colon 2^V \to \NN$ be a symmetric function (i.e., $f(V \setminus S)=f(S)$ for every $S \subseteq V$).
The \emph{$f$-width} of a branch decomposition $(T,\delta)$ of $V$
equals $\max_{e \in E(T)} f(A_e)$, and the \emph{$f$-branch-width} of $G$
is the maximum $f$-width over all branch decompositions of $V(G)$.
Then the rank-width of $G$ is the $f$-branch-width of $G$ where $f(S)$
is the rank of the submatrix $M[S,V \setminus S]$ of the adjacency
matrix $M$ of $G$, evaluated over GF$(2)$. 

We generalize this notion to the setting of having multiple graphs
(or, more generally, relations) over a common ground set, given our application.
That is, given a set of graphs $\GG=\{G_1, \ldots, G_k\}$
over a common vertex set $V$ and a branch decomposition $(T,\delta)$ of $V$,
we let $f(S)=\max_{G \in \GG} f_G(S)$ where $f_G(S)$ is the rank
function of $G$ as above, and consider the $f$-branch-width of $\GG$
for this function $f$. In other words, we have a notion of
\emph{simultaneous rank-width} of the collection of graphs $\GG$.
We show that there is a simple constraint language $\Gamma$
with the connector polymorphism
such that the set of symmetric, binary relations $\mo$-fgpp-definable
over $\Gamma$ has unbounded simultaneous rank-width, even when using
just two relations. 

Define the relation $R_D=\{01, 10, 12, 21\}$.
Then the language of binary relations over a domain $\numdom{n}$ definable
via monotone fgpp-definitions over $R_D$ and $\mathsf{Impl}$ has unbounded rank-width. 

The main obstruction is the following. For even $n \in \Nat$, 
let the \emph{$n$-diamond} be the binary relation
$R_n \subseteq \numdom{n}^2$ where
%\[
%  R_n =
%  \{(a,n/2-a+1) \mid a \in [n/2]\} \cup
%  \{(a,n/2+a) \mid a \in [n/2]\} \cup
%  \{(n/2+a, a) \mid a \in [n/2]\} \cup
%  \{(n/2+a, n-a+1) \mid a \in [n/2]\}.
%\]

\[
\begin{array}{rclc}
R_n & = & \{(a,n/2-a+1) \mid a \in \numdom{n/2}\}  \; \cup \\
 & &  \{(a,n/2+a) \mid a \in \numdom{n/2}\} \; \cup \\
 & & \{(n/2+a, a) \mid a \in \numdom{n/2}\} \; \cup \\
 & &   \{(n/2+a, n-a+1) \mid a \in \numdom{n/2}\}.
\end{array}
\]

Illustrated as a matrix for $n=8$, it looks as follows.
\[
  R_8 =
  \begin{bmatrix}
    0 & 0 & 0 & 1 & 1 & 0 & 0 & 0 \\
    0 & 0 & 1 & 0 & 0 & 1 & 0 & 0 \\
    0 & 1 & 0 & 0 & 0 & 0 & 1 & 0 \\
    1 & 0 & 0 & 0 & 0 & 0 & 0 & 1 \\
    1 & 0 & 0 & 0 & 0 & 0 & 0 & 1 \\
    0 & 1 & 0 & 0 & 0 & 0 & 1 & 0 \\
    0 & 0 & 1 & 0 & 0 & 1 & 0 & 0 \\
    0 & 0 & 0 & 1 & 1 & 0 & 0 & 0 \\
  \end{bmatrix}
\]
Note that this matrix is symmetric, hence it defines a graph. 
Now, as an unordered structure, it is underwhelming -- $R_{4p}$
is isomorphic to $pC_4$, i.e., $p$ disjoint copies of 4-cycles --
but if the decomposition is required to preserve the underlying domain
order, then it is more complex.
In fact, $R_n$ is a canonical example of an ordered structure with
unbounded \emph{stretch-width}~\cite{BonnetD23stretchwidth}
which is a recent width measure for ordered graphs and more generally
ordered binary structures, whose power lies between clique-width and twin-width.
Thus, we can get the stronger result that binary relations in $\fgpp{R_D}{\mo}$ 
have unbounded stretch-width directly from~\cite{BonnetD23stretchwidth}.
However, for completeness, we present a direct proof against
simultaneous rank-width. 

For the result, let $H_n \subseteq \numdom{n}^2$ be the binary relation 
where $(i,j) \in H_n$ if and only if $i \leq j$. This corresponds
to the bipartite adjacency matrix of the \emph{half-graph}.
We observe the following.

\begin{figure}
  \centering
  \begin{tabular}{ccc|cc|ccc}
    \textbf{0} & \textbf{0} & \textbf{0} & 1 & 1 & \textbf{0} & \textbf{0} & \textbf{0} \\
    \hline
    0 & 0 & 1 & \textbf{0} & \textbf{0} & 1 & 0 & 0 \\
    0 & 1 & 0 & \textbf{0} & \textbf{0} & 0 & 1 & 0 \\
    1 & 0 & 0 & \textbf{0} & \textbf{0} & 0 & 0 & 1 \\
    1 & 0 & 0 & \textbf{0} & \textbf{0} & 0 & 0 & 1 \\
    0 & 1 & 0 & \textbf{0} & \textbf{0} & 0 & 1 & 0 \\
    0 & 0 & 1 & \textbf{0} & \textbf{0} & 1 & 0 & 0 \\
    \hline
    \textbf{0} & \textbf{0} & \textbf{0} & 1 & 1 & \textbf{0} & \textbf{0} & \textbf{0} \\    
  \end{tabular}  
  \caption{An illustration of the M-fgpp-definition of the diamond
    constraint $R_n$. Showing $n=8$ and $a=3$ in the definition. The
    boldfaced entries are set to zero by this constraint.}
  \label{fig:build-a-diamond}
\end{figure}

\begin{lemma}
  For every even $n \in \Nat$, $R_n, H_n \in \fgpp{\{R_D,\mathsf{Impl}\}}{\mo}$.
 % where $M_n$ is the set of weakly monotone maps $[n] \to [3]$.
\end{lemma}
\begin{proof}
  For $H_n$, we simply note
  \[
    H_n(x,y) \equiv \bigwedge_{i \in \numdom{n}} \mathsf{Impl}([x \geq i],[y \geq i]).
  \]
  For $R_n$, we can define
  \[
    R_n(x,y) = \bigwedge_{a=0}^{n/2-2} R_D(f_a(x), f_{(n/2-1)-a}(y)) \quad \text{where}
    \quad
    f_a(i) =
    \begin{cases}
      0 & i \leq a \\
      1 & a < i < n-a \\
      2 & n-a \leq i.
    \end{cases}
  \]
  Correctness is easily verified; see Figure~\ref{fig:build-a-diamond} for an illustration.
\end{proof}

We show that there is no branch decomposition of $\numdom{n}$ that
simultaneously has small rank-width for $R_n$ and $H_n$. 

\begin{theorem} \label{thm:unbounded_rankwidth}
  The set of relations $\{R_n,H_n\}$ has unbounded rank-width as $n \to \infty$. 
\end{theorem}
\begin{proof}
  Let the vertex set be $V=\{u_i \mid i \in [n]\} \cup \{v_j \mid i \in [n]\}$.
  Let $H$ be the half-graph and $G$ the diamond graph on $V$ with the given bipartition.
  Let $(T,L)$ be a rank decomposition of $V$ where $L$ is a
  bijection between $V$ and the leaves of the tree.
  Assume that for every edge $e$ of $T$, the rank of the cut across
  $e$ is at most $c$ (for some $c$ independent of $n$) for both graphs.
  For an edge $e$ of $T$, let $V = A_e \cup B_e$ be the partition of
  $[n]$ induced by $e$; the order of $A_e$ versus $B_e$ is irrelevant.
  We first note that for every edge $e$ of $T$, the cut $A_e \cup B_e$
  must be approximately an ordered cut, in the following sense.
  
  Let $e$ be an edge of $T$. Let an \emph{alternating sequence} for $(A_e,B_e)$
  be a sequence $1 \leq i_1 \leq j_1 \leq \ldots \leq i_\ell \leq j_\ell \leq n$
  such that $u_{i_1} \in A_e$, $v_{j_1} \in B_e$, \ldots, $v_{j_\ell} \in B_e$.
  We note that every alternating sequence has $\ell \leq c$, since the biadjacency
  matrix of $H[u_{i_1},\ldots, v_{j_\ell}]$ contains an induced upper triangular
  matrix of rank $\ell$. The same holds for the reverse $(B_e,A_e)$. 

  This has two immediate consequences. First, note that at most $c$
  indices are split by $e$, i.e., for all but $c$ indices $i \in [n]$,
  $u_i$ and $v_i$ are on the same side of the cut. Second, given this fact
  (and ignoring the split indices), the cut $(A_e,B_e)$ is described by 
  at most $2c$ contiguous intervals, i.e., there is a partition
  $[n]-S=I_1 \cup \ldots \cup I_\ell$, where $S$ is the set of split indices
  and for $i<j$, all values of $I_i$ are smaller than all values of $I_j$,
  and where (without loss of generality) $I_1$ is in $A_e$, $I_2$ in $B_e$, {\em etc.}

  Now let $e$ be an edge of the tree such that $|A_e|, |B_e| \geq n/3$.
  Without loss of generality, we can find disjoint $I,J \subseteq [n]$
  with $|I|, |J| = \Omega(n/c)$, where $I<J$, $I \subseteq A_e$
  and $J \subseteq B_e$. Furthermore, there are $I' \subseteq I$
  and $J' \subseteq J$ with $|I'|, |J'| = \Omega(n/c)$,
  such that $I' \subseteq [n/2]$ or $I' \cap [n/2]=\emptyset$
  and the same for $J'$. But now $G$ contains an induced matching
  from $I'$ to $J'$, so the biadjacency matrix of $G$ across $e$
  has rank $\Omega(n/c)$. 
  
  If $n=\Omega(c^2)$, we have arrived at a contradiction. 
\end{proof}

\subsection{Projected Grid-rank}
\label{sec:projected-gridrank}

We now investigate the restrictions imposed by the connector property for
non-binary relations. Inspired by Lemma~\ref{lem:line_gridrank},
which shows that if $\Gamma$ has the connector property
then binary relations $\mo$-fgpp-definable over $\Gamma$
have bounded grid-rank, we define a generalization of grid-rank to
possibly non-binary relations, \emph{projected grid-rank},
and show that for any constraint language $\Gamma$,
the set of relations $\mo$-fgpp-definable over $\Gamma$
has bounded projected grid-rank if and only if
$\Gamma$ has the connector property.
For relations of arity $r=2$, projected grid-rank coincides with
regular grid-rank.

%\todo[inline]{Ideally: We would here make reference to papers in the
%  logic literature, and/or papers related to Markus-Tardos, that
%  discuss why ``$r$-dimensional hypergrid-rank'' is a non-starter.
%}

We note that the notion is much more restrictive than a direct,
na\"ive generalization of grid-rank to $r$-dimensional hypergrids.
Consider the following example.
Let the domain be $\numdom{n}$ and define a relation $R \subseteq \numdom{n}^3$ as
$R(x,y,z) \equiv (x+y=z)$. If we view this relation as a 3-dimensional grid,
then it does not even contain a $2 \times 2 \times 2$ subgrid with
every cell non-empty: let $(x_1,y_1,z_1)$ and $(x_2, y_2, z_2)$
be non-zero points in cells $(1,1,1)$ and $(1,1,2)$, respectively,
and let $(x,y,z)$ be a point in cell $(2,2,1)$.
Then $x>x_2$, $y>y_2$ and $z<z_2$.
However, $R$, together with arbitrary unary
constraints, defines a W[1]-hard language.

Indeed, let $S \subseteq [n]$ be a \emph{Sidon set}, i.e., a set such
that all sums $a+b$, $a, b \in S$ are pairwise distinct.
Such Sidon sets exist with cardinality $\Theta(\sqrt{n})$
and they can be constructed in polynomial time.

\begin{proposition}[\cite{erdos1941problem}]
    \label{prop:golombconstruction}
    Let $p \geq n$ be an odd prime. Then 
    $$S_n = \left\{ pa + (a^2 \bmod p) : a \in [0,n-1] \right\}$$ 
    is a Sidon set.
\end{proposition}

Proposition~\ref{prop:golombconstruction} 
shows that there is a Sidon set
containing $k$ positive integers and whose largest element is at most $2p^2$,
where $p$ is the smallest prime number larger than or equal to $k$. Bertrand's
postulate (see e.g. Chapter 2 in the book by Aigner and Ziegler~\cite{AignerZiegler18}) states that for every natural
number $m$ there is a prime number $p$ satisfying $m \leq p \leq 2m$. Hence,
a Sidon set $S=\{s_1,\dots,s_k\}$ and $s_k - s_1 \geq 8k^2$ can be generated in polynomial time.

We can use $S$ and $R$ to construct an arbitrarily complex binary relation
$R'$, such as the flip permutation, using a single local variable, via
\[
  R'(x,y) \equiv \exists z \colon (x \in S) \land (y \in S) \land (x+y=z) \land (z \in T),
\]
where $T \subseteq S+S$ contains precisely the set of sums $a+b$
such that $(a,b) \in R'$. Thus we can construct an arbitrary symmetric
relation over $S \times S$ where $|S|=\Theta(\sqrt{n})$.
From here, it is easy to complete the W[1]-hardness proof.

Instead, as the name implies, projected grid-rank is defined around
2-dimensional \emph{ordered projections} of a relation $R$. 
We provide the definitions; at the end of the section, we use them to
illustrate that the relation $x+y=z$ over $\numdom{n}$ is not
$\mo$-fgpp-definable for all $n$ over any finite language $\Gamma$ with the connector
property.

\begin{definition}
  Let $n, r \in \N$. 
  An \emph{ordered subset} of $\numdom{n}^r$ is a set $S \subseteq \numdom{n}^r$
  which admits a total ordering $s_1 \prec \ldots \prec s_m$
  which is order-isomorphic to an $m$-point combinatorial line,
  i.e., for each $j \in [r]$ either $s_1[j]<\ldots<s_m[j]$, or
  $s_1[j]=\ldots=s_m[j]$, or $s_1[j]>\ldots>s_m[j]$ holds.
\end{definition}

Grid-rank, that was defined earlier for matrices (Definition~\ref{def:grid-rank}), can be
extended to binary relations~\cite{bonnet2024twin}.

\begin{definition}
  Let $R \subseteq \numdom{n}^2$ be a relation.
  A \emph{$d$-division} of $R$ is a pair of ordered partitions
  $\numdom{n}=I_1 \cup \ldots \cup I_d = J_1 \cup \ldots \cup J_d$,
  where $I_i<I_j$ and $J_i < J_j$ for all $1 \leq i < j \leq d$.
  It has \emph{rank $k$} if every cell contains at least $k$ distinct
  rows and at least $k$ distinct columns. $R$ has 
  \emph{grid-rank $k$} if there exists a $k$-subdivision of rank $k$.
\end{definition}

The following is our generalization of grid-rank. 

\begin{definition}
  Let $A \cup B$ be a partition of $[r]$ for some $r \in \N$
  and let $R \subseteq [n]^r$ be a relation. A \emph{projected grid}
  of $R$ w.r.t.\ $A \cup B$ is a binary relation
  \[
    R' \subseteq S \times T \colon R'=\{(s,t) \in S \times T \mid s \cup t \in R\}
  \]
  where $S$ is an ordered subset of $[n]^A$, $T$ is an ordered subset of $[n]^B$,
  and $s \cup t$ denotes the tuple in $[n]^r$ whose entries in coordinates $A$
  match $s$ and whose entries in coordinates $B$ match $t$. 
  We say that $R$ has \emph{projected grid-rank $k$} if it has a
  projected grid with grid-rank at least $k$. 
\end{definition}

Following the proof of Lemma~\ref{lem:line_gridrank}, it is easy to show 
that relations $\mo$-fgpp-definable over a language with the connector property
have bounded projected grid-rank. 

\begin{lemma} \label{lemma:projected-grid-rank}
  Let $\Gamma$ be a finite constraint language.
  The set of relations $\fgpp{\Gamma}{\mo \cup \mo'}$ 
  has bounded projected grid-rank if and only if
  $\Gamma$ satisfies the connector property.
\end{lemma}
\begin{proof}
  On the one hand, assume that $\Gamma$ does not satisfy the connector
  property. By Lemma~\ref{lemma:not_line_gives_r3},
  $\Gamma$ can $(\mo \cup \mo')$-fgpp-define the
  permutation-generating relation $R_3$, so by Lemma~\ref{lemma:r3_can_define},
  $\Gamma$ $(\mo \cup \mo')$-fgpp-defines the graph of every permutation.
  Thus, even the binary relations from $\fgpp{\Gamma}{\mo \cup \mo'}$
  have unbounded projected grid-rank. On the other
  hand, assume that $\Gamma$ has the connector property.
  Let $r_0$ be the maximum arity of a relation in $\Gamma$,
  let $d$ be the maximum size of a domain of a relation in $\Gamma$,
  and let $k=r_0d$. 
  We claim that every relation in $\fgpp{\Gamma}{\mo \cup \mo'}$ has
  projected grid-rank $O(k)$. Assume to the contrary, and
  let $R \subseteq \numdom{n}^r \in \fgpp{\Gamma}{\mo \cup \mo'}$ for
  some $n \in \N$.
  Let $\numdom{r}=A \cup B$ be a partition, $S$ and $T$ ordered
  subsets of $\numdom{n}^A$ respectively $\numdom{n}^B$,
  and let $R' \subseteq S \times T$ be a projected grid of $R$ w.r.t.\ $A \cup B$.
  Assume that $R'$ has a subdivision defined by intervals $I_1<\ldots<I_\ell$
  and $J_1<\ldots<J_\ell$ for $\ell \geq 2k+3$, where each cell has
  rank at least $\ell$. 
  The proof now follows Lemma~\ref{lem:line_gridrank}.
  Let $\alpha \in \overline{R}$ be an entry from cell $(k+2,k+2)$ such
  that both the row and the column of $\alpha$ are non-zero in the cell;
  this exists by assumption on the rank of the cell.
  Let $R''(f_1(X_{i_1}), \ldots, f_p(X_{i_p}))$ be a constraint
  from the definition of $R$ which eliminates the entry $\alpha$.
  Then both the sequences $S$ and $T$ take at most $k$ distinct values
  from $[d]^A$ resp.\ $[d]^B$ in their combined maps, due to their ordering.
  (Specifically, there are at most $(d-1)p < k$ breakpoints in the
  sequences where at least one map $f_j(X_{i_j})$ changes values.)  
  As in Lemma~\ref{lem:line_gridrank}, there are now coordinates $a_1, b_1 \in [k+1]$
  and $a_2, b_2 \in [k+3,2k+3]$ such that cells $I_{a_1}$, $I_{a_2}$,
  $J_{b_1}$ and $J_{b_2}$ are constant in all maps of this constraint. But then,
  $(I_{a_1},I_{a_2}) \times (I_{b_1},I_{b_2})$ forms a $K_{2,2}$ in $R'$,
  so by the connector property, $R'$ must form a not necessarily spanning biclique 
  over the rectangle $(I_{a_1} \cup \ldots \cup I_{a_2}) \times 
  (J_{b_1} \cup \ldots \cup J_{b_2})$ spanned by these cells. 
  This contradicts our assumptions about $\alpha$.
\end{proof}

Let us illustrate this notion using the relation $x+y=z$. First, we show that
we can support this relation over any finite, fixed domain size $d$. 

\begin{lemma}
  The relation $R \subseteq \numdom{d}^3$ defined as
  $R(x,y,z) \equiv (x+y=z)$ has the connector property.
\end{lemma}
\begin{proof}
  Let $t_{00}, t_{02}, t_{11}, t_{20}, t_{22} \in R$ be such that
  $\fcon_d(t_{00},t_{02},t_{11},t_{20},t_{22})$ is defined. 
  Thus, in every argument $i \in \{x,y,z\}$,
  either $t_{00}[i]=t_{02}[i] \circ t_{11}[i] \circ t_{20}[i]=t_{22}[i]$ for some $\circ \in \{\leq, \geq\}$ or
  $t_{00}[i]=t_{20}[i] \circ t_{11}[i] \circ t_{02}[i]=t_{22}[i]$ for some $\circ \in \{\leq,\geq\}$.
  Refer to a coordinate as \emph{type 1} resp.\ \emph{type 2} if the
  first respectively second case applies; these types are not disjoint.
  Note that there is never precisely one variable \emph{not} of type
  $i$, for $i=1, 2$. Otherwise, for example, assume that $t_{00}[x] \neq t_{02}[x]$
  but $y$ and $z$ are both type 1. Then the tuples $t_{00}$ and $t_{02}$
  differ in only one coordinate, which is impossible.
  Since every coordinate has a type, this is only possible if all
  variables have the same type(s). But then, the connector polymorphism 
  is projective on the tuple, i.e., $\fcon_d(t_{00},\ldots,t_{22}) \in \{t_{00},\ldots,t_{22}\}$.
  Thus $R$ has the connector property. 
\end{proof}

On the other hand, the projected grid-rank of this relation increases with $d$.

\begin{lemma}
  The relation $x+y=z$ over a domain $\numdom{n}$ of arbitrary size has
  unbounded projected grid-rank. 
\end{lemma}
\begin{proof}
  Let $R(x,y,z) \equiv (x+y=z)$ over an arbitrarily large domain $\numdom{n}$.
  We demonstrate a projected grid in $R$ of arbitrary grid-rank.
  Let $A=\{x,y\}$ and $B=\{z\}$, and assume that $n=p^2$ for some $p \in \N$. 
  Let $S$ be a sequence which for every $i \in \numdom{p/2}$, $j \in \numdom{p}$
  contains the pairs $(p^2-(2i+1)p-j, 2ip+2j)$ visited in
  lexicographic order $(i,j)$, i.e., 
  \[
    (p^2-p,0) \prec (p^2-p-1,2) \prec \ldots \prec (p^2-2p+1,2p-2) \prec
    (p^2-3p, 2p) \prec (p^2-3p-1, 2p+2) \prec \ldots \prec (1,p^2-2),
  \]
  which is clearly an ordered subset of $\numdom{n}^A$. 
  For every such pair $(i,j)$, the sum $x+y=p^2-p+j$ is independent of $i$,
  hence the projected grid of $R$ on $S$ and $T=[p^2-p,p^2-1]$ 
  consists of an ascending sequence of length $p$ repeated $\Theta(p)$ times,
  which has unbounded grid-rank for growing $p$.
\end{proof}

\subsection{Bounded Projected Grid-rank Implies Bounded Twin-width}
\label{sec:gridrankramsey}

We conjecture that for every base language $\Gamma$ with the connector polymorphism,
$\mcsp(\Gamma)$ is FPT parameterized by the number of
variables (Conjecture~\ref{conjecture}).
To support this conjecture, we show that if $R$ is a $k$-ary relation
what is $\mo$-fgpp-definable over such a language $\Gamma$,
then the binary projected relation $\pi_{1,2} R$ has twin-width
bounded as a function of $k$.

Let us briefly illustrate why this is an important statement. Assume,
say, that there existed a 3-ary relation $R_n \in \fgpp{\Gamma}{\mo}$
over a domain $\numdom{n^2}$ such that $\pi_{1,2} R_n$ is the flip
relation (if the order $[n]^2$ is mapped into $\numdom{n^2}$ in
lexicographical order). Then there would be a reduction from
\textsc{Multicoloured Clique} with parameter $k$ 
into an instance with $\mcsp(\Gamma \cup \{\mathsf{Eq}\})$ with $O(k^2)$ variables,
following the standard reduction from Section~\ref{sec:unrestr-discussion}:
For every pair $i, j \in [k]$ with $i < j$,
in addition to the variables $v_{i,j}$ and $v_{j,i}$ as in the construction,
we introduce one locally quantified variable $z_{i,j}$
and implement the relation $R_n(v_{i,j}, v_{j,i}, z_{i,j})$. 
This is a priori not contradicted by the statement that binary
relations over $\Gamma$ have bounded twin-width, since the connector
polymorphism, being a partial polymorphism, is not closed under
existential projection, even of a single variable $z_{i,j}$.
However, the statements in this section refute such a construction,
since this would imply that $\pi_{1,2} R_n$ has unbounded twin-width.

More technically, we will show the following theorem.

\begin{theorem} \label{thm:projected-grid-rank-gr}
  There is a function $f \colon \N \times \N \to \N$ 
  such that the following holds:
  For all $d, k \in \N$ and every relation $R \subseteq \numdom{n}^k$
  with projected grid rank at most $d$, 
  the projection $\pi_{1,2} R$ has grid rank at most $f(d,k)$.   
\end{theorem}

Our desired conclusions about $\mcsp(\Gamma)$ are a consequence of this.

\begin{corollary} \label{corollary:twin_width_bounded}
  Let $\Gamma$ be a fixed, finite base language over a finite domain.
  If $\Gamma$ has the connector polymorphism, then for any relation
  $R \in \fgpp{\Gamma}{\mo}$ of arity $k$, 
  the projection $\pi_{1,2} R$ has twin-width bounded
  as $f_\Gamma(k)$. 
\end{corollary}
\begin{proof}
  Let $\Gamma$ be a base language with the connector polymorphism
  and $R \in \fgpp{\Gamma}{\mo}$ be a relation definable over $\Gamma$. 
  By Lemma~\ref{lemma:projected-grid-rank}, $R$ has bounded projected
  grid rank (depending only on $\Gamma$), hence by Theorem~\ref{thm:projected-grid-rank-gr}
  the relation $\pi_{1,2} R$ has grid rank bounded by some $f(\Gamma,k)$.
  Then $\pi_{1,2} R$ has bounded twin-width by~\cite{bonnet2024twin}.
\end{proof}

\subsubsection{Preliminary Tools}

The proof of Theorem~\ref{thm:projected-grid-rank-gr} uses Ramsey arguments; no effort has been made to minimize
the function dependency, as our interest is purely in the dependency
being bounded. The argument combines two aspects, the
existence of a large high-rank grid (cf.\ the applications of the 
Marcus-Tardos theorem in twin-width; see Bonnet et al.~\cite{bonnet2024twin,Bonnet:etal:jacm2022}
or Pilipczuk et al.~\cite{PilipczukSZ22compact}) 
with a requirement that dimensions $3, \ldots, k$ behave in some
regularly ordered way in the grid. This second aspect is closely
inspired by the ``two-dimensional Erd\H{o}s-Szekeres'' theorem
of Fishburn and Graham.

We use the following terminology (to reserve the term \emph{grid} to
its meaning in the rest of the paper). For sets $S_1,\dots,S_d$, refer to 
a set of the form $S_1 \times \dots \times S_d$ as a \emph{cube}. Write
$S^d$ for a product of the form $S \times \dots \times S$ with $d$ factors. A 
$[k]^d$-{\em subcube} of a cube
$S_1 \times \dots \times S_d$ is a subset of $S_1 \times \dots S_d$ of the form 
$S'_1 \times \dots \times S'_d$, where $S'_i$ is a
$k$-element subset of $S_i$.

\begin{theorem}[{\cite{FishburnG93lexicographic}}]
\label{thm:twodim-Erdos-Szekeres}
  For every $d, in \in \N$ there is $N(d,n) \in \N$ such that 
  for every $N \in \N$, $N \geq N(d,n)$, and every injective function
  $f \colon [N]^d \to \N$ there exists a subcube $A_1 \times \ldots \times A_d$
  of $[N]^d$ with $|A_1|=\ldots=|A_d|=n$ such that
  $f$ is monotone lexicographic in the subcube. 
\end{theorem}

Here, the ordering is lexicographic up to two modifications: a
permutation of the dimensions $[d]$ from most significant to least
significant, and a choice of direction (increasing or decreasing)
in every dimension. See also~\cite{BucicSTerdosszekeri} for better bounds.
We will only need the two-dimensional version. 

In certain cases, we also resort to use the more blunt {\em product Ramsey theorem}.

\begin{theorem}[Theorem 11.2.13 in \cite{Bodirsky:Book}] \label{thm:product-ramsey}
  For all positive integers $d$, $c$, $s$, and $m$, there is a
  positive integer $\ell$ such that for every colouring of the
  $[s]^d$-subcubes of $[\ell]^d$ with $c$ colours there exists a
  monochromatic $[m]^d$-subcube $G$ of $[\ell]^d$, 
  i.e., all the $[s]^d$-subcubes of $G$ have the same colour.
\end{theorem}

%We now prove Theorem~\ref{thm:projected-grid-rank-gr}.

The proof of Theorem~\ref{thm:projected-grid-rank-gr} is simplified
by using a variation of projected grid rank:
up to some dependency on $k$, the ``strict'' order
requirement in the projected grid rank definition can be relaxed to a
\emph{weak} order requirement (Lemma~\ref{lemma:weak-order-suffices}). 

\begin{definition}
  A sequence $x_1, \ldots, x_n \in \N$ is \emph{weakly ordered}
  if either $x_1 \leq \ldots \leq x_n$ or $x_1 \geq \ldots \geq x_n$. 
  Let $n, r \in \N$. A \emph{weakly ordered subset} of $\numdom{n}^r$ 
  is a set $S \subseteq \numdom{n}^r$ with an ordering
  $s_1 \prec \ldots \prec s_m$ such that for every $i \in [r]$,
  the sequence $s_1[i], \ldots, s_m[i]$ is weakly ordered. 
\end{definition}

We relax the notion of a projected grid in the natural way. 

\begin{definition}
  Let $A \cup B$ be a partition of $[r]$ for some $r \in \N$
  and let $R \subseteq [n]^r$ be a relation. A \emph{weak projected grid}
  of $R$ w.r.t.\ $A \cup B$ is a binary relation
  \[
    R' \subseteq S \times T \colon R'=\{(s,t) \in S \times T \mid s \cup t \in R\}
  \]
  where $S$ is a weakly ordered subset of $[n]^A$, $T$ is a weakly ordered subset of $[n]^B$,
  and $s \cup t$ denotes the tuple in $[n]^r$ whose entries in coordinates $A$
  match $s$ and whose entries in coordinates $B$ match $t$. 
\end{definition}
%We say that $R$ has \emph{weak projected grid-rank $k$} if it has a
%weak projected grid with grid-rank at least $k$. 

To distinguish this from the regular definition, we refer to a
projected grid as a \emph{proper} projected grid.
For a relation $R \subseteq \numdom{n}^k$, a partition $[k]=A \cup B$,
and a weak projected grid $R' \subseteq S \times T$ with respect to~$A \cup B$,
we let $M_{S,T} \in \{0,1\}^{S \times T}$ denote the matrix where
$M[u,v]=[(u,v) \in R']$.
We refer to $R'$ as being \emph{defined by $(A,B,S,T)$}. 

We exploit a simplifying assumption from Bonnet et al.~\cite{bonnet2024twin}.
For $k \in \N$, $\NN_k$ refers to a class of eight structurally simple
matrices in $\{0,1\}^{k \times k}$; the details are not important to
us, except that they are of full rank.
Let $M \in \{0,1\}^{n \times n}$ be a matrix. 
A \emph{rank-$k$ Latin $d$-division} of $M$ is a $d$-division of $M$
into row parts $I_1 \cup \ldots \cup I_d$ and column parts $J_1 \cup \ldots \cup J_d$
such that for each $i \in [d]$, $I_i$ and $J_i$ are further divided
into $d$ parts $I_i = I_{i,1} \cup \ldots \cup I_{i,d}$
respectively $J_i = J_{i,1} \cup \ldots \cup J_{i,d}$
such that for every $i \in [d]$ there are permutations $\sigma^r_i \colon [d] \to [d]$
and $\sigma^c_i \colon [d] \to [d]$ such that the following hold.
\begin{enumerate}
\item For all $i, j \in [d]$, $M \cap (I_{i,\sigma^r(i)} \times (J_{j,\sigma^c(j)})) \in \NN_k$, and
\item all other submatrices $M \cap (I_{i,a} \times J_{j,b})$ are all-1 or all-0.
  %equal $\mathbf{1}_k$ or $\mathbf{0}_k$.\todo{PJ: Notation mot introduced.}
\end{enumerate}
We refer to the region $I_{i,\sigma^r(i)} \times J_{j,\sigma^c(j)}$ of $M$
as the \emph{active region} of cell $(i,j)$ of $M$. 
A \emph{rank $k$ Latin division} is a rank $k$ Latin $k$-division.

\begin{lemma} \label{lemma:weak-order-suffices}
  There is a function $f(k,d)$ such that the following holds. 
  Let $R \subseteq \numdom{n}^k$ be a relation which contains a
  weak projected grid with grid-rank $f(k,d)$. Then $R$ has projected
  grid-rank at least $d$. 
\end{lemma}
\begin{proof}
  We fix $d$ and prove the fact by induction on $k$.
  Let $R'$ be a weak projected grid defined by $(A,B,S,T)$,
  and assume that $M_{S,T}$ has grid rank at least $f(k,d)$,
  for a value $f(k,d)$ to be determined.
  Let $(I,J)$ be an $f(k,d)$-division of $M_{S,T}$ such that every
  cell has rank at least $f(k,d)$.
  By Lemma~23 of~\cite{bonnet2024twin}, we may assume that $(I,J)$ is Latin.

  If $k=2$, then only a proper projected grid is possible and there is
  nothing to show. Thus, assume $k>2$ and assume that the statement holds
  for every relation of arity less than $k$. If we can find a single
  value, in any coordinate, that covers at least $f(k-1,d)$ whole
  blocks, then we fix that value and proceed by induction.
  
  Otherwise, select a weak projected grid $M_{S',T'}$ of smaller dimension
  as follows. Assume $|B|>1$, otherwise swap the role of rows and columns
  in the following. Let $A'=A$ and $S'=S$.
  Repeat the following $p=f(|A|+1,d)^2$ times: Mark a column in the left-most column
  block from the active region in row block $1$ that is not constant
  in its active region; discard $f(k-1)+2$ column blocks; mark a
  column in the left-most column block from the active region in row
  block $2$ in the same way; and so on, until $p$ columns in the
  active regions of distinct row blocks have been marked. 
  Repeat this sweep a further $p$ times, starting over from row block $1$.
  Partition the columns in the selection according to the $p$ sweeps, and
  partition the rows so that each new block contains $p$ old blocks. 
  This is a $p$-division of the new weakly projected grid $M_{S',T'}$.
  Furthermore, in every cell of this division, there are marked columns
  from the active regions of $p$ distinct row blocks within its
  participating row blocks. These are mutually distinct, since every
  column is non-constant in its active region and constant in all
  other regions. Clearly, we can choose $f(k,d)$ large enough that
  the sequence of sweeps succeeds. Thus the result holds by induction. 
\end{proof}

\subsubsection{Proof of Theorem~\ref{thm:projected-grid-rank-gr}}

We are now ready to prove Theorem~\ref{thm:projected-grid-rank-gr}. 
We proceed by
induction over somewhat more refined objects we refer to as
\emph{weak projected grids with hidden dimensions}. 
Let $R \subseteq D_1 \times \ldots \times D_k$ be a $k$-ary
relation over ordered domains $D_i$, with projected grid rank at
most $d$, for some $d \in \N$.
Let $A, B \subseteq [k]$ be disjoint subsets,
let $S$ be a weakly ordered subset of $\prod_{i \in A} D_i$
and $T$ a weakly ordered subset of $\prod_{j \in B} D_j$.
Let $C=[k] \setminus (A \cup B)$ and $k'=|C|$,
and define $D'=\prod_{i \in C} D_i$. 
Then $(A,B,S,T)$ define a matrix $M_{S,T}$ with rows indexed
by $S$ and columns indexed by $T$ as follows.
If $C=\emptyset$, then $(A,B,S,T)$ defines a weak projected gid
and $M_{S,T}$ is just its matrix. Otherwise, the entries of $M_{S,T}$ are
\[
  M_{S,T}[u,v] = \{x \in D' \mid (u \cup v \cup x) \in R\}
\]
where $(u \cup v \cup x)$ is shorthand for the $k$-dimensional
vector whose entries are taken from $u$, $v$ and $x$ as appropriate.
We refer to $M_{S,T}$ as a \emph{weak projected grid with $k'$ hidden dimensions}. 
Let $M_\pi=M_{S,T}^\pi \in \{0,1\}^{S \times T}$ be the matrix
with entries $M_\pi[u,v]=[M_{S,T}[u,v] \neq \emptyset]$;
this is the \emph{projection} of $M_{S,T}$. 
We show by induction that there is a function $f(d,k) \in \N$
such that for every weak projected grid $M_{S,T}$ of $R$ with $k'$
hidden dimensions, the grid rank of $M_{S,T}^\pi$ is at most $f(d,k')$.

As a base case, consider $|C|=0$. Then as defined, $M_{S,T}^\pi=M_{S,T}$ is
just a weak projected grid of $R$, hence has grid rank at most $d$
by assumption. Thus the base case holds. Now consider a projected
grid $M_{S,T}$ with $k'$ hidden dimensions for some $k' > 0$,
and assume by induction that for every projected grid $M_{S',T'}$
with $\ell < k'$ hidden dimensions, $M_{S',T'}^\pi$ has grid rank at
most $f(d,\ell)$. We will show that if $M_{S,T}^\pi$ has sufficiently 
large grid rank, then there is a projected grid $M_{S',T'}$
with $k'-1$ hidden dimensions such that $M_{S',T'}^\pi$ has grid rank
more than $d'=f(d,k'-1)$, which finishes the induction step by contradiction. 

Let $r \in C$ be one of the hidden dimensions. 
We split into two cases, depending on how $M_{S,T}^\pi$ interacts with the
coordinate $r$. The induction is wrapped up at the very end of this section.

\paragraph*{Cells with diverse labels.}
Let $I_1 \cup \ldots \cup I_m$ and $J_1 \cup \ldots \cup J_m$
be the division of $M_{S,T}^\pi$ into row blocks and column blocks,
where $m=f(d,k')$. For $\ell \in D_r$ let $R_{r=\ell}$ the $(k-1)$-ary
relation formed by fixing coordinate $r$ to $\ell$. 
Say that label $\ell$ is \emph{diverse in cell $(i,j)$} 
if the projection of $R_{r=\ell}$ (instead of $R$) into
the cell $I_i \times J_j$ contains at least $d'$ distinct rows and
columns. 

\begin{lemma} \label{lemma:induction-rich-cells}
  There is a function $g(d')$ such that the following holds. Assume
  that there are $U, W \subseteq [m]$ with $|U|=|W|=g(d')$
  such that every cell $(i,j) \in U \times W$ has a label $\ell_{i,j}$
  that is diverse in the cell. Then $R$ has a weak projected subgrid
  with at most $k'-1$ hidden dimensions and grid rank at least $d'$. 
\end{lemma}
\begin{proof}
  Let $L \in D_r^{U \times W}$ be the matrix where $L[i,j]=\ell_{i,j}$
  is a label that is diverse in cell $(i,j)$. By a product Ramsey argument (Theorem~\ref{thm:product-ramsey}),
  if $g(d')$ is large enough then we can find a submatrix $L'$ of $L$
  of dimension $(d')^2 \times (d')^2$ such that every $2 \times 2$
  submatrix of $L'$ has the same \emph{order type}, with respect to
  the six possible comparisons of the labels in the submatrix.
  This will result in one of the following cases for $L'$.
  \begin{itemize}
  \item All labels are constant
  \item All rows are constant, and all labels are strictly ordered along the columns
  \item All columns are constant, and all labels are strictly ordered along the rows
  \item The labels are ordered in one of the eight possible lexicographical
    orders (row-first/column-first; increasing/decreasing along the rows;
    increasing/decreasing along the columns).
  \end{itemize}
  In the first three cases, we can easily sweep through all rows
  and columns of $M_{S,T}$ to create a weakly ordered subset that ``carries'' the
  hidden label dimension along with it -- that is, say that the label
  is constant along every row of $L'$. Then we let $S'=S$ and combine the label from $D_r$
  into $T$ in a weakly ordered way by ensuring, at every column $c$
  from a column belonging to a block $J_j$ indexed by $L'$, 
  that the value used in $D_r$ at $c$ comes from the label used in column $j$ of $L'$.
  Then, in each cell of $M_{S',T'}$ indexed in $L'$, the resulting weak projected grid
  has rank at least $d'$, implying a weak projected grid of rank $d'$.
  Hence, assume that we are processing a subcube $L'$ with lexicographic ordering.
  Assume w.l.o.g.\ that the order is row-first and let $T'=T$.
  Then let $S'$ contain all rows from $S$, combining the label from $D_r$ with the row
  index as follows: up until the end of the block indexed by entry $id'+j$
  of $L'$, let the label take the value from column block $d'j+i$ of $L'$.
  These label values are weakly ordered by assumption. Now consider the
  division of $S' \times T'$ into $d'$ blocks covering $d'$ entries from $L'$ each. 
  For each cell in this division, there is a cell
  $(i,j) \in U \times W$ such that $S'$ visits its rows using label $\ell_{i,j}$. 
  By assumption, this gives rank at least $d'$ within the cell.
\end{proof}

Thus, if we can find a large enough subgrid in the division of $M_{S,T}^\pi$
where every cell has a diverse label, then the induction applies.
We proceed with the non-diverse case. 

\paragraph*{Cells with no diverse labels.}
Now, assume that we can find a large subgrid in which no cell has a
diverse label. We first observe that a cell of high rank and with no
diverse label has a large collection of labels which induce mutually
distinct rows and columns.

\begin{lemma} \label{lemma:give-many-labels}
  There is a function $h(p)$ such that the following holds. Let
  $(i,j)$ be a cell with no diverse label and with at least $h(p)$
  distinct rows and at least $h(p)$ distinct columns, for some $p \in \N$. 
  Then there is a set $L$ of $p$ distinct labels, plus the following:
  \begin{enumerate}
  \item A selection of $p$ rows $r_1 < \ldots < r_p$ in the cell
    and a bijection $\sigma^r \colon [p] \to L$ such that the rows formed
    from the cell by taking row $r_t$ of $R_{r=\sigma^r(t)}$ are
    pairwise distinct.
  \item  A selection of $p$ columns $c_1 < \ldots < c_p$ in the cell,
    and a bijection $\sigma^c \colon [p] \to L$ such that the columns
    formed from the cell by taking column $c_t$ of $R_{r=\sigma^c(t)}$  are pairwise distinct.
  \end{enumerate}
\end{lemma}
\begin{proof}
  We note the following.
  The cell in $M_{S,T}^\pi$ is the Boolean sum of the matrices formed
  from $R_{r=\ell}$ over all possible labels $\ell$ (by definition of $M_{S,T}^\pi$).
  For a row $r$ in the cell, and a value $\ell \in D_r$, let $(r,\ell)$
  denote row $r$ in cell $(i,j)$ when defined from $R_{r=\ell}$ instead of $R$,
  and define $(c,\ell)$ analogously for a column $c$. If the rank
  $h(p)$ is sufficiently large, then for any $q$ we can find a
  sequence $(c_1,\ell_1)$, \ldots, $(c_q,\ell_q)$ of pairwise distinct columns
  where the $c_i$ are strictly increasing, and since no label is
  diverse in the cell we may also assume that the labels $\ell_i$ are
  pairwise distinct. Now consider the cell formed by only the labels
  $\ell_i$ selected in this process. This matrix has rank growing with
  $q$, hence we may repeat the argument and find, within this same
  label set, a sequence $(r_i,\ell_i')$ of distinct rows where $r_i$
  is strictly increasing and the labels $\ell_i'$ pairwise distinct. Return this final set of labels $\ell_i'$ as $L$. 
\end{proof}

Let us also make the following observation. The \emph{span} of a set of numbers $I$ is the interval $[\min I, \max I]$.

\begin{lemma} \label{lemma:make-disjoint-spans}
  Let $s, t \in \N$ and let $L_1, \ldots, L_t$ be a collection of
  label sets from $\N$, each with at least $st$ labels.
  There are label sets $L_i' \subseteq L_i$ with mutually disjoint
  spans such that $|L_i'|=s$ for each $i$.
\end{lemma}
\begin{proof}
  Repeat the following operation: Let $x$ be the smallest number such that
  one label set $L_i$ contains $s$ numbers $y \leq x$. Replace that label set
  $L_i$ by its prefix up to $x$, name this $L_i'$ and put it aside,
  and remove values up to $x$ from all other label sets. In every step, at
  most $s$ labels are removed from every remaining set. Therefore,
  when $q \leq t$ label sets remain, each contains at least $qs$ labels
  and the process can continue until termination. 
\end{proof}

We now give the inductive argument for the many-labels case.

\begin{lemma} \label{lemma:induce-poor-labels}
  There is a function $h(d')$ such that the following holds.
  Assume that the division of $M_{S,T}^\pi$ admits a $h(d') \times h(d')$ grid
  where the cells $(i,j)$ are given labels sets $L(i,j)$ with mutually
  disjoint spans, with $|L(i,j)| \geq d'$ for all $(i,j)$. 
  Then there is a weak projected subgrid of $R$ with at most $k'-1$
  hidden dimensions and grid rank at least $d'$. 
\end{lemma}
\begin{proof}
  Since the label sets have disjoint spans, we can treat them as
  distinct integers, and use the two-dimensional Erd\H{o}s-Szekeres
  theorem (Theorem~\ref{thm:twodim-Erdos-Szekeres}) to find a subgrid of dimension $(d')^3 \times (d')^3$ in
  which the label intervals are lexicographically ordered. 
  Assume w.l.o.g.\ that the major order is along the columns (i.e.,
  for every column $c$, the intervals in column $c+1$ all come after
  all intervals in column $c$). Similarly to Lemma~\ref{lemma:induction-rich-cells},
  we may now form a $d'$-division consisting of $(d')^2$ blocks each,
  such that in each such large block, at least $d'$ distinct cells
  have been marked in a sweep. The only distinction is, when selecting
  a label $\ell \in D_r$ to use for coordinate $r$ in the sweep, 
  then we need to select a label $\ell$ from the corresponding set of
  labels $L(i,j)$ such that the resulting column is distinct from all previously chosen
  labelled columns in the large block. This is possible since each
  cell $(i,j)$ contains $d'$ distinct labels producing pairwise
  distinct columns; hence at every stage, one available label will
  give a column distinct from previously selected columns in the large
  block, and the disjoint spans guarantee that all selections yield
  a weakly ordered subset. Thus we have produced a weak projected grid
  with a $d'$-division of rank at least $d'$ in every cell, and with
  at most $k'-1$ hidden dimensions.
\end{proof}

\paragraph*{Conclusion of proof.}
We can now finish the induction. Let $d''=f(k,d')$ be sufficiently large,
and let $(I,J)$ be a $d''$-division of $M_{S,T}^\pi$ 
with rank $d''$ in each cell.
Define an auxiliary matrix $H \in \{0,1\}^{d'' \times d''}$,
where $H[i,j]$ indicates whether cell $I_i \times J_j$ has a diverse
label or not. By the basic product Ramsey argument, for any $p=g(k,d')$
we may assume that we have a monochromatic subcube $U \times W$ of $H$
with $|U|=|W|=p$. 
If this is a subcube where every cell has a diverse label, we invoke
Lemma~\ref{lemma:induction-rich-cells}, ensure $p \geq g(d')$ for the
function $g$ from this lemma, and produce a weak projected grid with
grid rank at least $d'$.
Otherwise, $U \times W$ forms a subcube where no cell has a diverse
label. Let $h(d')$ be the function from Lemma~\ref{lemma:induce-poor-labels}.
By beginning with a size $|U|=|W|=p$ larger than $h(d')$, we can use
Lemma~\ref{lemma:give-many-labels} to give label sets of length $h(d')^2d'$,
then restrict $U \times W$ to a smaller subcube $U' \times W'$
of dimension $h(d') \times h(d')$, 
then apply Lemma~\ref{lemma:make-disjoint-spans} to produce label sets
of length $d'$ with mutually disjoint spans. 
Lemma~\ref{lemma:induce-poor-labels} now finishes the induction.

\section{Discussion}
\label{sec:discussion}

We introduced the $\udcsp(\Gamma, \MM)$ framework to represent CSPs with few variables but with unbounded domain size. These problems have appeared in the literature as important, intermediate problems in both hardness and tractability contexts, but has not been systematically studied as an independent formalism. To be able to obtain general results we introduced a novel algebraic approach based on partial multifunctions and proved that the polymorphism side boiled down to partial multifunctions (definable as patterns) interpreted over different domains with the help of the map family $\MM$. In the monotone world these patterns turned out to be particularly simple and resulted in {\em order polymorphisms} where the ordering of the input arguments is sufficient to determine the output.

For $\ucsp$ and $\ohcsp$ we gave complete (parameterized and classical) complexity dichotomies. Here, the $\ohcsp$ problem turned out to have strong similarities with the problem of finding a solution of weight $k$. Curiously, the parameterized dichotomy turned out to be equivalent and given by weak separability, but the classical complexity differed, showing that the two problems are indeed computationally different. We continued with the richer $\mcsp(\Gamma)$ problem where we first gave a classical complexity dichotomy via the classical $\min, \max$ and $\median$ operations. In the parameterized setting we proved W[1]-hardness if $\Gamma$ does not have the connector property, and proved a complementary FPT result for binary base languages $\Gamma$. However, even the binary setting is rather expressive and we exemplify this by showing how one of the hardest steps in the \textsc{Boolean MinCSP} dichotomy by Kim et al.~\cite{KimKPW23fa3} can be simplified by a reduction to an $\mcsp$ problem. 

We conclude that the $\udcsp(\Gamma,\mo)$ framework is extremely expressive and, based on the connections we already made in this paper, seems to be just the right blend since it is rich enough to connect to powerful methods and deep questions in parameterized complexity while still being tameable with universal algebra.
We close the paper by discussing some future research directions.

\paragraph{A formal connection to parameterized \textsc{MinCSP}?}

Let us begin with the question that motivated the study. 
As mentioned in the introduction, there appears to be a connection
between the parameterized complexity of \textsc{MinCSP$(\Gamma)$}
parameterized by the number of false constraints
and the complexity of $\udcsp(\Gamma',\MM)$-problems expressible in the
\textsc{MinCSP}, for suitable $\Gamma'$ and $\MM$ (especially $\MM=\mo$).
It would be very interesting to work out this connection more formally.
For example, is there a direct algebraic proof that if $\Gamma$
belongs to one of the FPT classes for \textsc{Boolean MinCSP$(\Gamma)$},
then $\Gamma$ cannot express a suitably defined FPT-reduction
from any problem $\mcsp{\Gamma'}$ where $\Gamma'$ does not have the
connector property? 

We note that the classical complexity of \textsc{MinCSP} and \textsc{Valued CSP}
problems is characterized algebraically by \emph{fractional} or
\emph{weighted polymorphisms}~\cite{FullaZ16galois},
which are quite different from the partial polymorphism patterns
considered in this paper. Furthermore, they are too coarse-grained to 
distinguish between FPT and W[1]-hard cases of \textsc{MinCSP$(\Gamma)$}.
Consequentially, existing parameterized complexity characterizations
of \textsc{MinCSP$(\Gamma)$}-type problems have all been non-algebraic in
nature, either purely combinatorial in their dichotomy descriptions~\cite{ChitnisEM17,KimKPW23fa3,OsipovW23equality},
or giving only a partial dichotomy (e.g., classifying a finite list of cases)~\cite{Dabrowski:etal:ipec2023,OsipovPW24pointalgebra},
or else only considering the complexity up to constant-factor FPT approximations~\cite{BonnetEM16}.
The latter is a more well-behaved question, algebraically speaking,
than the existence of an exact FPT-algorithm since constant-factor FPT approximability is preserved under efpp-definitions.
More algebraically guided tools are expected to be required for FPT/W[1]-dichotomies
for \textsc{MinCSP$(\Gamma)$} going much beyond this list. 

For a particularly intriguing case, \textsc{Temporal CSPs} are a class
of infinite-domain CSPs over the ordered domain $(\Q, <)$.
As with finite-domain CSPs, there is a notion of a language-restricted
version \textsc{Temporal CSP$(\Gamma)$} for a constraint language $\Gamma$, 
in which case $\Gamma$ contains a finite set of relations
definable over $(\Q,<)$ in first-order logic. 
The set of tractable temporal CPSs was characterized by Bodirsky and K\'ara~\cite{BodirskyK10temporal}
(covered in a more modern way in the book of Bodirsky~\cite{Bodirsky:Book}).
A very natural follow-up question is to determine the temporal
languages $\Gamma$ such that the \textsc{MinCSP} over $\Gamma$ is FPT
parameterized by the solution cost. 
However, this task has proven highly challenging so far, as temporal CSPs
have significant structural complexity. 
The results of this paper should help in this regard. 

\paragraph{A dichotomy for non-binary relations.}

The most intriguing, and also the hardest, open question is whether our dichotomy for $\mcsp$ can be strengthened to arbitrary non-binary relations, and since our hardness condition is still valid the missing case is when $\Gamma$ has the connector property. Classical graph parameters such as rank-width or clique-width are not useful in this setting since we have proved that there exist families of languages with the connector property that have unbounded rank-width. We therefore investigated a higher-arity generalization of grid-rank, {\em projected grid-rank}, and proved that   the connector property coincides with bounded projected grid-rank. Through a non-trivial Ramsey argument we proved that  projected grid-rank coincides with bounded twin-width of binary projections. Constructively using this property seems hard, however, without simultaneously generalizing twin-width and contraction sequences to higher-arity relations, and via our Ramsey argument we know that we cannot actually escape twin-width. For a concrete example where the complexity of $\mcsp(\Gamma)$ is unknown, consider e.g. $\Gamma = \{R_{1/3}\}$ for $R_{1/3} = \{001, 010, 100\}$. This relation has the connector property but is non-binary, and could be a manageable starting point towards resolving the full dichotomy.

%\paragraph{Projected grid-rank versus twin-width.}
%\todo[inline]{This discussion point depends on whether the Ramsey argument works out or not.}

\paragraph{More classes of $\udcsp(\Gamma, \MM)$ problems.}
We concentrated on three map families: $\all$, $\mo$, and Boolean $\oh$. What other map families result in interesting $\udcsp(\Gamma, \MM)$ problems? For example, it seems likely that non-Boolean generalizations of $\oh$ could lead to problems capturing CSPs parameterized by solution size~\cite{Bulatov:Marx:sicomp2014}. 
Another application concerns the signed logics from Section~\ref{sec:applications}. They have been studied
with other partial orders in mind than those that we have considered, cf. \cite[Section 7]{Beckert:etal:survey2000} and \cite[Section 7]{Chepoi:etal:ejc2010}. Typical examples are the full class of lattices and subclasses such as distributive and modular lattices.
For even more general classes it is easy to formulate map families with respect to an algebraic structure, e.g., by associating the base domain $\numdom{d}$ and the input domain $\numdom{n}$ with groups (or some other algebraic structures) and then consider all group homomorphisms. This question can also be generalized with category theory where the basic algebraic objects (e.g., groups or orders) would correspond to categories, and the map family would correspond to sets of morphisms.

\section*{Acknowledgements}

Peter Jonsson and Jorke de Vlas are partially supported  by the Swedish Research Council (VR) under grant 2021-04371. Victor Lagerkvist is partially supported by VR under grant VR-2022-03214.
The authors wish to thank \'Edouard Bonnet, Marcin Pilipczuk and Szymon Toru\'nczyk 
for helpful discussions about non-binary twin-width notions.

\bibliographystyle{abbrv}
\bibliography{bib}

\end{document}